\documentclass[authoryear,11pt,final]{elsarticle}%
\RequirePackage{natbib}
\RequirePackage{hyperref}
\RequirePackage{hypernat}
\usepackage{amsmath,mathtools}
\usepackage{nicefrac}
\usepackage{amssymb}
\usepackage{appendix}
\usepackage{subcaption}
\usepackage[capitalise,noabbrev]{cleveref}
\crefname{assumption}{Assumption}{Assumptions}
\crefname{lem}{Lemma}{Lemmas}
\crefname{align}{}{}
\crefname{appsec}{Appendix}{Appendices}
\crefname{assumptionc}{Assumption}{Assumptions}
\crefname{lemmac}{Lemma}{Lemmas}
\crefname{propositionc}{Proposition}{Propositions}
\crefname{remarkc}{Remark}{Remarks}
\crefname{corollaryc}{Corollary}{Corollaries}
\usepackage[shortlabels]{enumitem}
\makeatletter
\let\@afterindenttrue\@afterindentfalse
\makeatother
\usepackage{autonum}
\usepackage{dsfont}
\usepackage{pgfplots,tikz}
\pgfplotsset{compat=1.9}
\usepackage{booktabs,siunitx}
\usepackage[minskip=6pt,proofsymbol=\ensuremath{_\blacksquare}]{coolthms}
\theoremmarkup[\bf][\it]
\definetheorem[remarkc]{remark}{Remark}[section]
\definetheorem[theoremc]{theorem}{Theorem}[section]
\definetheorem[corollaryc]{corollary}{Corollary}[section]
\definetheorem[lemmac]{lemma}{Lemma}[section]
\definetheorem[propositionc]{proposition}{Proposition}[section]
\definetheorem[assumptionc]{assumption}{Assumption}[section]
\theoremstyle{empty}
\theoremsymbol{\ensuremath{_\blacksquare}}
\RequirePackage{amssymb}
\theoremheaderfont{\scshape}
\theorembodyfont{\normalfont}
\theoremseparator{}

\newcommand{\loading}[1]{\lambda_{#1}}

\newcommand{\Loading}{\Lambda}
\newcommand{\hatLoading}{\hat\Lambda}
\newcommand{\autocorr}{\gamma}

\newcommand{\SN}{\mathds{N}}
\newcommand{\SR}{\mathds{R}}

\newcommand{\lawMP}[1][0]{\mathrm{P}_{#1,n,T}^{\text{MP}}}
\newcommand{\law}[1][0]{\mathrm{P}_{#1,n,T}}
\newcommand{\lawPANIC}[1][0]{\mathrm{P}_{#1,n,T}^{\text{PANIC}}}
\newcommand{\expec}{\mathrm{E}}

\newcommand{\var}[1][]{\operatorname{var}_{#1}}
\newcommand{\ACFgamma}[1]{\gamma_{\eta,#1}}
\newcommand{\ACFgammaF}[1]{\gamma_{f,#1}}

\newcommand{\vareps}{\Sigma_\varepsilon}
\newcommand{\vareta}[1][]{\Sigma_{\eta #1}}
\newcommand{\varF}[1][]{\Sigma_{f #1}}

\newcommand{\proxvareps}{\Psi_\varepsilon}
\newcommand{\proxvarepsT}{\Psi_{\varepsilon}}

\newcommand{\proxvarepsgen}[1][]{\psi_{\varepsilon #1}}

\newcommand{\hatproxvarepsgen}[1][]{\hat\psi_{\varepsilon #1}}
\newcommand{\auxproxvarepsgen}[1][]{{\psi^*_{\varepsilon #1}}}

\newcommand{\longrunvarmatrixeta}{\Omega_\eta}
\newcommand{\longrunvarmatrixetaT}{\Omega_{\eta}}
\newcommand{\LReta}{\Psi_\eta}
\newcommand{\LRetaT}{\Psi_{\eta}}
\newcommand{\hatlongrunvarmatrixeta}{\hat\Omega_\eta}
\newcommand{\longrunvarF}[1]{\omega_{f,#1}^2}
\newcommand{\longrunvarFT}[1]{\omega_{f,#1,T}^2}
\newcommand{\longrunvarmatrixF}{\Omega_{F}}

\newcommand{\longrunvareta}[1]{\omega_{\eta,#1}^2}
\newcommand{\longrunvaretaT}[1]{\omega_{\eta,#1,T}^2}
\newcommand{\longrunvaretaTroot}[1]{\omega_{\eta,#1,T}}

\newcommand{\oslongrunvareta}[1]{\delta_{\eta,#1}}
\newcommand{\oslongrunvaretaT}[1]{\delta_{\eta,#1,T}}

\newcommand{\longrunvar}[1]{\omega^2_{#1}}
\newcommand{\longrunvarT}[1]{\omega^2_{#1,T}}
\newcommand{\oslongrunvarfT}[1]{\delta_{f,#1,T}}

\newcommand{\hatlongrunvareta}[1]{\hat{\omega}_{\eta,#1}^2}
\newcommand{\hatoslongrunvareta}[1]{\hat{\delta}_{\eta,#1}}

\newcommand{\limitfactorloadings}{\Psi_\Loading}

\newcommand{\CSorigMP}{\Delta_{n,T}^{\text{MP}}}
\newcommand{\CSorigPANIC}{\Delta_{n,T}^{\text{PANIC}}}

\newcommand{\FIorigMP}{ J_{n,T}^{\text{MP}}}
\newcommand{\FIorigPANIC}{ J_{n,T}^{\text{PANIC}}}
\newcommand{\CSMPLRV}{\tilde\Delta_{n,T}^{\text{MP}}}

\newcommand{\auxCS}{\Delta^*_{n,T}}
\newcommand{\etaCS}{\Delta_{n,T}}
\newcommand{\hatCS}{\hat{\Delta}_{n,T}}

\newcommand{\FIMP}{ \tilde J_{n,T}^{\text{MP}}}
\newcommand{\hatJ}{ \hat J_{n,T}}

\newcommand{\rnorm}[1]{\left\| #1 \right\|_{\text{spec}}}

\newcommand{\id}[1][]{I_{#1}}
\newcommand{\kron}{\otimes}
\newcommand{\diag}{\operatorname{diag}}
\newcommand{\trans}{^\prime}

\newcommand{\ones}{\iota}
\newcommand{\tr}{\operatorname{tr}}
\newcommand{\mineig}[1]{\lambda_{\text{min}}\left(#1\right)}
\newcommand{\maxeig}[1]{\lambda_{\text{max}}\left(#1\right)}
\newcommand{\Frob}[1]{ \left\| #1 \right\|_F}

\newcommand{\vto}{\stackrel{d}{\longrightarrow}}
\newcommand{\jointlimits}{(n,T\to\infty)}

\newcommand{\calA}{\mathcal{A}}
\newcommand{\A}{A}

\newcommand{\rH}{\mathrm{H}}
\newcommand{\rd}{\mathrm{d}}
\newcommand{\crosseta}[1]{\eta_{\cdot,#1}}

\newcommand{\m}{m}

\newcommand{\asMP}{assumption:serial correlation,assumption:factor_loadings,assumption:rates,ass:MP}
\newcommand{\asPANIC}{assumption:serial correlation,assumption:factor_loadings,assumption:rates,ass:BN}
\newcommand{\asEither}{assumption:serial correlation,assumption:factor_loadings,assumption:rates,assumption:framework}

\newcommand{\hatlrv}{\hat\omega}
\newcommand{\hatoslrv}{\hat\delta}
\newcommand{\hatsquaredomegas}{\hat\phi}
\newcommand{\lrv}{\omega}
\newcommand{\oslrv}{\delta}
\newcommand{\squaredomegas}{\phi}
\newcommand{\bfe}[1][i]{E_{#1}}
\newcommand{\BNbfe}[1][i]{\hat E_{#1}}
\newcommand{\BNbfem}[1][i]{\hat E_{-1,#1}}
\newcommand{\BNrho}{\hat\rho^+}
\newcommand{\hatGamma}{\hat\Loading}
\newcommand{\MPrho}{\rho^+_{\text{pool}}}
\newcommand{\tUMP}{t_{\text{UMP}}}
\newcommand{\tUMPEMP}{t_{\text{UMP}}^{\text{emp}}}

\begin{document}
%
%
%
\begin{frontmatter}
\title{
Local Asymptotic Equivalence of the Bai and Ng (2004) and Moon and Perron (2004)
Frameworks for Panel Unit Root Testing
}
\author[T]{Oliver Wichert}\ead{O.wichert@uvt.nl}
\author[D]{I. Gaia Becheri}\ead{I.G.Becheri@tudelft.nl}
\author[T]{Feike C. Drost}\ead{F.C.Drost@uvt.nl}
\author[T]{Ramon van den Akker}\ead{R.vdnAkker@uvt.nl}
\address[T]{Tilburg University, Department of Econometrics \& Operations Research}
\address[D]{Zurich Insurance Group Ltd.}
%

\begin{abstract}
This paper considers unit-root tests in
large $n$ and large $T$ heterogeneous panels with
cross-sectional dependence generated by unobserved factors. 
We reconsider the two prevalent approaches in the literature, that of  \cite{MoonPerron2004} and the PANIC setup proposed in \cite{BaiNg2004}. While these have been considered as completely different setups, we show that, in case of Gaussian innovations, the frameworks are  asymptotically equivalent in the sense that both experiments are locally asymptotically normal (LAN) with the same central sequence. 
Using Le Cam's theory of statistical experiments we determine the local asymptotic power envelope and derive an optimal test jointly in both setups. We show that the popular
\cite{MoonPerron2004} and \cite{BaiNg2010} tests only attain the power envelope in case there is no heterogeneity in the long-run variance of the idiosyncratic components. The new test is asymptotically uniformly most powerful irrespective of possible heterogeneity. Moreover, it turns out that for any  test, satisfying a mild regularity condition,  the  size and local asymptotic power are the same under both data generating processes. Thus, applied researchers do not need to decide on one of the two frameworks to conduct unit root tests. 
Monte-Carlo simulations corroborate our asymptotic results and document 
significant gains in finite-sample power if the variances of the idiosyncratic shocks differ substantially among the cross sectional units.
\end{abstract}
\end{frontmatter}
Testing for unit roots is an important aspect of time series and panel data analysis.\footnote{See, for example, the textbook \cite{Choi2015} for an overview.}
The presence of unit roots not only determines how to proceed for correct statistical inference but can also have serious policy implications. A well-known problem with univariate unit roots tests is their low power. In the last two decades, increased data availability led to the development of panel unit root tests that increase the statistical power by exploiting the cross-sectional data dimension.

The ``first generation'' of panel unit root tests does not allow for cross-sectional dependence, i.e., panel units are assumed to be independent of each other.\footnote{See, for example, the surveys \cite{Banerjee1999}, \cite{BaltagiKao2000}, \cite{Choi2006}, \cite{BreitungPesaran2008}, and \cite{WesterlundBreitung2013}. Local asymptotic powers of first generation tests have been considered in, for example, \cite{Breitung2000}, and \cite{Madsen2010}. A large scale Monte Carlo study to assess finite-sample powers was conducted in \cite{HlouskovaWagner2006}.}
For many, if not most, empirical applications, however, the assumption of cross-sectional independence is not only empirically hard to justify but has also non-trivial implications for the properties of test statistics. In fact, as shown in 
\cite{OConnell1998} and
\cite{Gutierrez2006}, the dependence between cross-section units can compromise the validity of  ``first generation tests''. For this reason, a ``second generation'' of tests, which are also valid in case of cross-sectional dependence, gained a foothold in the literature.\footnote{Popular second generation tests are proposed in, among others, \cite{PhillipsSul2003}, \cite{BaiNg2004,BaiNg2010}, \cite{MoonPerron2004},  \citeauthor{BreitungDas2005} (\citeyear{BreitungDas2005}, \citeyear{BreitungDas2008}), \cite{Pesaran2007} and \cite{PesaranSmithYamagata2013}.}

This paper reconsiders the two leading second generation classes of data generating processes, namely the PANIC framework proposed in \cite{BaiNg2004} and the framework proposed in \cite{MoonPerron2004}, henceforth MP. Both setups allow for cross-sectional dependence through common, unobserved factors. MP uses an autoregressive structure with the factors appearing in the innovations (errors). For PANIC, the factors are part of the ``mean specification''€™. Consequently, the PANIC framework allows for non-stationarity generated by the factors and for non-stationarity generated by the idiosyncratic components. This is in contrast to the MP framework, for which the factors and the idiosyncratic components have the same order of integration. As \cite{BaiNg2010}, \cite{PesaranSmithYamagata2013} and \cite{Westerlund2015}, this paper will focus on testing for unit roots in the idiosyncratic components. 

This paper offers four main contributions. Firstly, our results imply that 
\emph{for all tests (satisfying a mild regularity condition) it suffices to determine the asymptotic size and local power in one of the frameworks, since the same behaviour automatically holds for the other one.}\footnote{For unit-root testing, it thus is irrelevant to make a specific choice for the data generating process. This is good news for the practitioner, who no longer must decide between two competing frameworks that are typically hard to distinguish based on finite samples.}
Previous papers are based on either MP or PANIC for the construction of test statistics. However, using our first main result, the (same) local asymptotic power function is automatically obtained for the other framework as well.\footnote{In particular, we use the first main result to show that the unit root tests proposed in \cite{MoonPerron2004} are equivalent (in terms of asymptotic size and power) to the tests proposed in \cite{BaiNg2010}.
A first study on the comparison of the behavior of these tests is present in Bai and Ng (2010),
but, to our best knowledge, the equivalence  has not been observed
before.
}
These conclusions are obtained by showing that the PANIC and MP experiments are \emph{both Locally Asymptotically Normal (LAN) with the same central sequence and Fisher information}.\footnote{This means that the limit experiment (in the Le Cam sense) is a Gaussian shift experiment; see, for example, 
\cite{vdVaart2000}.
For unit root problems in (univariate) time
series, limit experiment theory has been exploited by, amongst others, \cite{Jansson2008}
and \cite{BDvdA}.
}

Secondly, exploiting the general theory for LAN experiments we easily obtain the \emph{local asymptotic power envelope, which is, in view of the first main result, the same for PANIC and MP}. This result extends the work by \cite{MoonPerronPhillips2007}, \cite{BDvdA},  \cite{MPP2014}, and \cite{JuodisWesterlund2018} on  first generation frameworks, to the second generation.  It turns out that the level of the local asymptotic power envelope only depends on the (local) deviation to the unit root. The level of the power envelope is thus not affected by the nuisance parameters. \footnote{Westerlund (2015) observed that the local asymptotic power of the tests considered in that paper do depend on the presence of serial and cross-sectional dependence (see Remark 2 in that paper). Consequently, these tests are not globally optimal.} We also provide a new derivation, using our LAN-result, of the local asymptotic power of the popular \cite{MoonPerron2004} and \cite{BaiNg2010} tests.\footnote{\cite{Westerlund2015} derived, via ``triangular array asymptotics'', the local asymptotic power function of the tests proposed in \cite{BaiNg2010}. Using our LAN-result we provide a short derivation at the expense of slightly less general conditions.} A comparison of the power functions to the power envelope shows that these original tests are only optimal in case there is no heterogeneity in the long-run variances of the idiosyncratic components.

Thirdly, we propose a \emph{new test that is asymptotically uniformly most powerful} irrespective of possible heterogeneity in the long-run variance of the idiosyncratic components. Our test is motivated by our expansion of the likelihoods, underlying the LAN results, and its optimality is easily proved by exploiting LAN-theory.  Compared to the tests proposed in \cite{MoonPerron2004} and \cite{BaiNg2010}, the size of the power gains depends on the amount of heterogeneity in the long-run variances of the idiosyncratic components. We report numerical asymptotic powers for commonly encountered amounts of heterogeneity and use Monte-Carlo experiments to show that the new test compares favorably also in finite samples.

Finally, to obtain the LAN result for the PANIC case, we first show that \emph{the model in which we observe both the panel units and the common factors is equivalent to that where the factors are unobserved}. This in contrast to other data generating processes, used in the literature on panel unit roots, where observing factors or correlated covariates does yield additional power; see, for example, \citet{PesaranSmithYamagata2013}, \citet{Becheri:2015gi}, and \cite{JuodisWesterlund2018}. Moreover, \emph{for both the MP and PANIC framework, our results imply that the local asymptotic power envelope for the setting in which all nuisance parameters (this includes incidental intercepts, factor loadings, and coefficients of the linear filters generating serial dependence) are known, can be attained}. In other words, we demonstrate that we are in an adaptive setting.

To obtain the local and asymptotic equivalence of the MP and PANIC frameworks, we need to impose some restrictions. First, we assume that the driving innovations are Gaussian. Second, we impose the deviations to the unit root, under the alternative hypothesis, to be the same for all panel units. And third, we do not allow for (incidental) trends. The Gaussianity facilitates a relatively easy proof of the LAN-result and it seems to be rather difficult to generalize this assumption; even for first-generation frameworks no results are available yet. Nevertheless, for the proposed asymptotically uniformly most powerful test, we stress that Gaussianity is not required to obtain asymptotic size, implying validity of the test in a non-Gaussian environment.  In view of \cite{ BDvdA} we do not expect that a generalization to random deviations from the unit root, under the alternative hypothesis, would affect our main results. The Monte Carlo results seem to confirm this conjecture for finite-samples. Allowing for incidental trends leads to a different convergence rate (see \cite{MoonPerronPhillips2007}) and thus requires a new asymptotic expansion of the log-likelihood ratios, which is beyond the scope of this paper. 

The paper is organized as follows. Section~\ref{sec:setup} presents the model and assumptions. Section~\ref{sec:limexp} derives the common approximation to the local likelihood ratios in the two experiments and derives its limiting distribution. \Cref{sec:UMPtest} introduces our new UMP test based on the limit experiment. \Cref{sec:tests_comparison} computes the local asymptotic power functions of the tests proposed in \cite{MoonPerron2004} and \cite{BaiNg2010} and \cref{sec:MC} compares their asymptotic and finite-sample power to those of the new UMP test. \Cref{sec:concl} concludes. All proofs are organized in several appendices.
\section{Setup and Assumptions}\label{sec:setup}
\subsection{Data-generating processes}
\label{sub:frameworks}
\noindent
We consider observations $Z_{it}$,  $i=1,\dots,n$ and $t=1,\dots,T$, generated by the components specification
\begin{align}
Z_{it}&=m_i+Y_{it}, \label{Zeq} \\
Y_{it}&=\sum_{k=1}^K  \loading{ki}  F_{kt} +  E_{it},\label{Yeq}\\
E_{it}&=\rho E_{i,t-1}+\eta_{it},\label{Eeq}\\
F_{kt}&=\rho_k F_{k,t-1}+f_{kt},\label{Feq}
\end{align}
with $\loading{ki}$ the loading of (unobserved) factor $\{F_{kt}\}$ on panel unit $i$, and $K\in\SN$ being the fixed number of factors. The $m_i$ are fixed effects and we assume zero starting values: $E_{i0}=0$ and $F_{k0}=0$.\footnote{\cite{BaiNg2004} assume that the initial values are bounded while \cite{MoonPerron2004} assume zero starting values for the $\varepsilon_{it}$ (see \cref{eqn:epsilon}). We refer to Section~6.2 in \cite{MoonPerronPhillips2007} for  a discussion on why relaxing initial conditions can be problematic in a panel context and do not pursue this issue further, except by noting that our tests are invariant with respect to the $m_i$.}
 The  assumptions on the innovations $\eta_{it}, f_{kt}$ and factor loadings $\lambda_{ki}$ are discussed in \cref{sub:assumptions} below. 
This setup covers the most widely used setups for second-generation panel unit root tests: for $\rho_k=1, k=1,\dots,K,$ we obtain the PANIC framework of \cite{BaiNg2004} (`PANIC') and with  $\rho_k=\rho, k=1,\dots,K,$ we obtain the framework of \cite{MoonPerron2004} (`MP'),
  in which we can also rewrite the DGP as
\begin{align}
Z_{it}&=m_i+Y_{it},  \\
  Y_{it}&=\rho Y_{i,t-1}+\varepsilon_{it},\\
  \varepsilon_{it}&= \sum_{k=1}^K  \loading{ki}  f_{kt} +  \eta_{it}.\label{eqn:epsilon}
\end{align}
In both frameworks, the hypotheses will be phrased in terms of $\rho$.

\begin{remark}
The PANIC framework does not require the factors to have a unit root. Therefore,  when considering the PANIC framework, we allow the $\{f_{kt}\}$ to be over-differenced (see \cref{assumption:framework} below).
\end{remark}
\begin{remark}
	In both frameworks, we do not allow for `heterogeneous alternatives', i.e., we impose that $\rho$ does not differ across panel units. This helps to unify the treatment of the two setups: A more general MP framework  where $Y_{it}=\rho_i Y_{i,t-1}+\varepsilon_{it}$, can no longer be rewritten in the PANIC form of \cref{Zeq,Yeq,Eeq,Feq}. \cite{BDvdA} prove, for the case without factors, unobserved heterogeneity in the autoregressive parameters has no impact on the power envelope or optimal tests. Therefore, in \cref{sec:MC} we also investigate the performance of our tests in the presence of heterogeneous alternatives; those results seem to confirm their conclusion that there is no impact on power also for the general factor case.
	\end{remark}

\subsection{Matrix notation}
To write the model in matrix form we need some additional notation. We write $I_n$ and $I_T$ for identity matrices of dimension $n$ and $T$, respectively, while $\ones$ denotes a $T$-vector of ones. 
Introduce the $n$-vectors $\loading{k}=(\loading{k1},\dots,\loading{kn})\trans$, $k=1,\dots,K$ and
 the $n\times K$ matrix $\Loading=(\loading{1},\dots,\loading{K})$.
Collect the observations as
$Y= (Y_{11},Y_{12}, \dots,Y_{1T},\dots,Y_{n1},\dots,Y_{nT})\trans$. We also write
$Y_{-1} =(Y_{10},Y_{11}, \dots,Y_{1,T-1},\dots,Y_{n0},\dots,Y_{n,T-1})\trans$, $\Delta Y = Y-Y_{-1}$,
and define $\varepsilon$, $\eta$, $E$, $E_{-1}$, $\Delta E$, $Z$, $Z_{-1}$, and $\Delta Z$ analogously.
Write $m=(m_1,\dots,m_n)\trans$, $\eta_i = (\eta_{i1},\dots,\eta_{iT})\trans$, $i=1,\dots,n$, 
$f_k =(f_{k1},\dots,f_{kT})\trans$, $k=1,\dots,K$,
  and  denote their corresponding covariance matrices by $
  \varF[,k]=\var f_k\in\SR^{T\times T}$
 and \begin{align}
  \vareta=\diag(\vareta[,1],\dots,\vareta[,n])\text{, with }
  \vareta[,i]=\var \eta_i\in\SR^{T\times T}.
\end{align}
The long-run variances of $\{f_{kt}\}$ and $\{\eta_{it}\}$ are denoted by $\longrunvarF{k}$ and $\longrunvareta{i}$, respectively. In addition, we define the approximate long-run variances $\longrunvarFT{k}= \ones'\varF[k] \ones/T$ and $\longrunvaretaT{i}= \ones'\vareta[,i] \ones/T$. 
For a given $T$, these ignore the contribution of any autocovariances further than $T$ apart.
We will use the approximate long-run variances to simplify notation and the structure of our proofs. 
We add the subscript T to the approximate versions to emphasize the difference and define 
\begin{align}
  \longrunvarmatrixeta=\diag(\longrunvaretaT{1},\dots,\longrunvaretaT{n})
  \text{ and }
  \longrunvarmatrixF=\diag(\longrunvarFT{1},\dots,\longrunvarFT{K}).
\end{align}
 
  In addition to this `vectorized' notation, it will also be useful to consider the observations as $T\times n$ matrices.  Thus, let $\tilde \eta=(\eta_1,\dots,\eta_n)$, and define $\tilde \varepsilon$, $\tilde Y$, $\tilde Z$, $\tilde E$, $\tilde f=(f_1,\dots,f_K)$, and $\tilde F$ analogously. 
With this notation, \cref{eqn:epsilon} can be rewritten as
\begin{align}\label{tildeFacModel}
  \tilde \varepsilon=\tilde f\Loading'+\tilde\eta,
  \end{align}
while for the vectorized versions we have
\begin{align}
  \varepsilon=\sum_{k=1}^K \loading{k} \otimes f_k + \eta.
\end{align}

Finally, we introduce the $T \times T$ matrix $\A$ by $\A_{st}:=1$ if $s>t$ and $0$ otherwise
and we put $\calA:= I_n\otimes A\in\SR^{nT\times nT}$, i.e.
\begin{align}
\A= \begin{pmatrix}
0& 0& \dots& 0\\
1& 0& \dots& 0\\
\vdots& \ddots& \ddots& \vdots\\
1& \dots& 1& 0
\end{pmatrix}
\text{ and }
\calA=\begin{pmatrix}
\A &  0_{T\times T} & \hdots &  0_{T\times T} \\
0_{T\times T} & \A  & \hdots &  0_{T\times T} \\
\vdots       &     \ddots  & \ddots  & 0_{T\times T} \\
0_{T\times T}  & \hdots & 0_{T\times T} & \A
\end{pmatrix}.
\end{align}
 The matrix $A$ can be considered a cumulative sum operator and premultiplying the vectorized panel with $\calA$ takes the cumulative sum in the time direction for each panel unit, i.e., we have  $\tilde Y_{-1}=\A \Delta \tilde Y$ and $Y_{-1}=\calA \Delta Y$. It is also related to `approximate one-sided long-run variances', which we can define by $\oslongrunvaretaT{i}=\tr[
A\Sigma_{\eta,i}/T]$ and $\oslongrunvarfT{k}
=\tr[A\Sigma_{f,k}/T]$. Note $A+A'= \ones\ones'-I_T$, so that, analogous to the long-run variances,  we have $2\oslongrunvaretaT{i}
=\longrunvaretaT{i}
-\ACFgamma{i}(0)$.
  
\subsection{Assumptions}\label{sub:assumptions}
Now we can formally state the full specifications of our DGPs \cref{Zeq,Yeq,Eeq,Feq}. The distributional assumptions on the time series of the factors $\{f_{kt}\}$ and idiosyncratic shocks $\{\eta_{it}\}$
are given in \cref{assumption:serial correlation} and we formulate the assumptions on the (deterministic) factor loadings $\loading{ki}$  in
\cref{assumption:factor_loadings}. \Cref{assumption:rates} specifies the joint asymptotics we consider in this paper. Finally, \cref{assumption:framework} differentiates between the two setups discussed in \cref{sub:frameworks}.

%
\begin{assumption}\label{assumption:serial correlation}\text{ } 
\begin{enumerate}[(a),topsep=6pt]
\item Each factor innovation, indexed $k=1,\dots,K$, is a zero-mean ergodic  stationary time series $\{f_{kt}\}$ independent of the other factors and all idiosyncratic parts. Its autocovariance function $\ACFgammaF{k}$ satisfies 
\begin{align}
  \sum_{m=-\infty}^\infty (|m|+1) | \ACFgammaF{k}(m)| <\infty
\end{align}
and is such that the variance of each factor innovation $\{f_{kt}\}$ is strictly positive. 
\item \Label{partb} For each panel unit $i\in\mathds{N}$, the idiosyncratic part $\{\eta_{it}\}$ is a Gaussian zero-mean stationary time series independent of the other idiosyncratic parts and all factors. The autocovariance function $\ACFgamma{i}$ satisfies \begin{align}\label{bddACFEta}
  \sup_{i\in\mathds{N}}\sum_{m=-\infty}^\infty (|m|+1) | \ACFgamma{i}(m)| <\infty
\end{align}
 and is such that the eigenvalues of the $T\times T$ covariance matrices are uniformly bounded away from zero, i.e., $\inf_{i,T} \mineig{\Sigma_{\eta,i}}>0$.
\end{enumerate}
\end{assumption}
\begin{remark}
The imposed restrictions on serial correlation are sometimes phrased in terms of spectral densities. Note that our assumption on the boundedness of the eigenvalues is implied by the spectral density being uniformly bounded away from zero (see, for example, Proposition 4.5.3 in \cite{BrockwellDavis}). Similarly, they are sometimes phrased in terms of linear processes on which analogous assumptions are imposed; see, for example, Assumption C in \cite{BaiNg2004} and Assumption 2 in \cite{MoonPerron2004}. Finally, note that a collection of causal ARMA processes satisfies \cref{assumption:serial correlation} if the roots are uniformly bounded away from the unit-circle.
\end{remark}
%
\begin{remark}\label{rem:bddLRV}
Note that, under \cref{partb},  the long-run variances of the $\{\eta_{it}\}$, $\longrunvareta{i}$, are also uniformly bounded and uniformly bounded away from zero.\footnote{The former directly follows from \cref{bddACFEta} whereas the latter follows from $
 \longrunvareta{i}= \lim_{T\to\infty}\frac{1}{T} \ones'\Sigma_{\eta,i} \ones\ge 
 \lim_{T\to\infty}\frac{1}{T}\mineig{\Sigma_{\eta,i}} \ones' \ones
 \ge\inf_{i,T}\mineig{\Sigma_{\eta,i}}>0$.\label{footnoteLRV}} 
Moreover, the one-sided long-run variances
\begin{align}
\oslongrunvareta{i} =\sum_{m=1}^{\infty} \ACFgamma{i}(m)=\frac{1}{2}\left( \longrunvareta{i} - \ACFgamma{i}(0)\right),\quad   i\in\SN,
\end{align}
 are also well-defined. 
\end{remark}
As already announced, we also need to impose some stability on the factor loadings $\loading{ki}$, which we assume to be fixed. \Cref{assumption:factor_loadings} is standard in the literature, c.f.\ Assumption A in \cite{BaiNg2004} or Assumption 6 in \cite{MoonPerron2004}. It is commonly referred to as the factors being `strong'.
\begin{assumption}\label{assumption:factor_loadings}
There exists a positive definite $K\times K$ matrix $\limitfactorloadings$ such that
$
\lim_{n\to\infty}  \frac{1}{n} \Loading\trans \Loading = \limitfactorloadings.$
Moreover, $\max_{k=1,\dots,K} \sup_{i\in\SN} |\loading{ki}|<\infty$.
\end{assumption}


\Cref{assumption:rates} below specifies the asymptotic framework we consider throughout this paper. We follow \cite{MoonPerron2004}, \cite{BaiNg2010}, and \cite{Westerlund2015} in considering  large `macro panels', where both $n$ and $T$ go to infinity, but $T$ will be the larger dimension. We derive all our results using joint asymptotics, which yields more robust results than taking sequential limits where first $T\to\infty$ and subsequently $n\to\infty$. 

\begin{assumption}\label{assumption:rates}
We consider joint asymptotics (in the \cite{PhillipsMoon1999} sense) with $n/T\to 0$. 
\end{assumption}
Finally, \cref{assumption:framework} below specifies that we either operate in the PANIC (case \ref{ass:BN}) or in the MP (case \ref{ass:MP}) framework. In the PANIC framework, we allow the long-run variance of the factor innovations to be zero, so that we consider both integrated and and stationary factors. This is ruled out in the MP case to enforce that the factors have the same order of integration as the idiosyncratic parts. 
\begin{assumption}
\label{assumption:framework} One of the below holds:
\begin{enumerate}[(a),topsep=6pt]	
\item For each factor ${F_k}, k=1,\dots,K,$ we have $\rho_k=1$, or, \Label{ass:BN}
\item For each factor $k=1,\dots,K$, we have  $\rho_k=\rho$. Moreover, $\{f_{kt}\}$ is \emph{Gaussian} and its long-run variance exists and is strictly positive.\Label{ass:MP}
\end{enumerate}
	
\end{assumption}

\section{Limit Experiment and Power Envelope}\label{sec:limexp}
%
\noindent
In this section we show that the likelihood ratios related to the unit root hypothesis, for the MP and for the PANIC framework, exhibit the same local asymptotic expansion. Both experiments are proved to be Locally Asymptotically Normal (LAN) with the same central sequence and (asymptotic) Fisher information. 
This result allows us to treat the two setups jointly and to obtain three main results. Firstly, we derive the asymptotic power envelopes. Secondly, in \cref{sec:UMPtest} we obtain asymptotically optimal, feasible tests. Thirdly, the LAN result allows us to show that any test, satisfying a mild regularity condition, has the same, perhaps nonoptimal, local asymptotic power function under both data generating processes.

 We phrase our hypotheses about $\rho$ in \cref{Zeq,Yeq,Eeq,Feq} using the local parameterization 
\begin{align}\label{localAlternatives}
	\rho=\rho^{(n,T)}=1+\frac{h}{\sqrt{n}T}.
\end{align}
As shown below, these rates lead to contiguous alternatives, which allow us to obtain the (local) power of our tests.
The unit root hypothesis can be reformulated in terms of the ``local parameter'' $h$:
\begin{align}
\rH_0:\, h=0 \text{ versus }  \rH_a:\, h<0.
\end{align}  

 In both setups, we start by considering the likelihood ratio for observing $Z_{it}$ in case $\rho$ is the only unknown parameter. Hence, the number of factors $K$, the factor loadings $\loading{ki}$, the autocovariance functions, and the fixed effects $m_i$ are considered as known in this section. 
 We will first show, for each model separately, that its likelihood ratio satisfies an expansion, under the null hypothesis, of the form $\log  \rd \law[h]/\rd \law= h\etaCS - h^2 J/2 + o_P(1) $ with Fisher-information $J=1/2$. In \cref{asymnorm}, we consider the limiting distribution of their common central sequence $\etaCS$ and will conclude that both experiments enjoy the LAN-property. In \cref{sec:UMPtest} we demonstrate that the conclusions of this section also hold for the model of interest, where all the parameters are unknown, i.e., that the nuisance parameters can be adaptively estimated.

\subsection{Expanding the likelihood in the PANIC setup}\label{expandPANIC}
For the PANIC case, for now, consider the factors as known. Just as for the other parameters, we show  in \cref{sec:UMPtest} that the resulting likelihood ratio can still be approximated by an observable version (up to a negligible term).\footnote{ 
The result implies that observing the factors in the PANIC framework will not lead to an increase in power. This contrasts with the situation in the MP setting, for which \cite{Becheri:2015gi} report higher powers with observed factors and \cite{JuodisWesterlund2018} show power gains when observing covariates correlated to the innovations.
} Denote the joint law of $F$ and $Z$ under \cref{assumption:serial correlation,assumption:factor_loadings,assumption:rates,ass:BN} by $\lawPANIC[h]$. Using  $\eta\sim N(0,\vareta)$ and $\eta=\Delta E- h E_{-1}/(\sqrt{n}T)$, we obtain the log-likelihood ratio
\begin{align}
\log \frac{ \rd \lawPANIC[h]}{\rd \lawPANIC}&= \frac{h}{\sqrt n T}  \Delta E\trans\calA'\vareta^{-1}\Delta E -\frac{h^2}{2 n T^2}  \Delta E\trans \calA'\vareta^{-1} \calA \Delta E\\
&=:h \CSorigPANIC
-\frac{1}{2} h^2 \FIorigPANIC.
\end{align}
Note, from \cref{tildeFacModel}, $\Delta\tilde E =\Delta\tilde Y -\Delta\tilde F \Loading'$ ,
implying $\Delta E$ is indeed observable in this PANIC framework (with observed factors as considered here). Moreover, under $\lawPANIC$, $\Delta E=\eta$.
We now show that we can replace variances by long-run variances, to obtain simpler versions of the central sequence and empirical Fisher information.
\begin{lemma}\label{lemma:PANICCS}
	Suppose that \cref{\asPANIC} hold.
Then we have, under $\lawPANIC$,
$(\CSorigPANIC,\FIorigPANIC) =(\etaCS,\frac{1}{2})+o_P(1)$, where
\begin{align}
  \etaCS=\frac{1}{\sqrt n T}  \Delta E\trans\calA'\LReta^{-1}\Delta E
 -\frac{1}{\sqrt n}\sum_{i=1}^n \frac{\oslongrunvaretaT{i}}{\longrunvaretaT{i}},\text{ with } \LReta^{-1}=\longrunvarmatrixeta^{-1}\kron I_T.\end{align}
\end{lemma}

\begin{remark} 
\label{remark:approx}
	The simplified central sequence $\etaCS$ is the result of substituting $\vareta^{-1}$ by $\LReta^{-1}$. It is, however, not the case that $\LReta^{-1}$ is a ``good''  approximation to $\vareta^{-1}$. As evident in $\etaCS$, the replacement necessitates a correction term for the central sequence to be centered. This term arises due to the fact that, contrary to $\vareta^{-1/2}\Delta E$, $\LReta^{-1/2}\Delta E$ does exhibit serial correlation. What we can show is that $\calA\trans \LReta$ approximates $\calA\trans\vareta$ well. This is thanks to $\longrunvaretaT{i}$ being roughly equal to the column sums of $\vareta[,i]$.
\cref{lemma:generalLRV} phrases this phenomenon in a general context.
\end{remark}

In the following subsections we  show that $\etaCS$ also approximates the central sequence in the \cite{MoonPerron2004} setup. 

\subsection{Expanding the likelihood in the \cite{MoonPerron2004} setup}\label{expaMP}
Let us denote the law of $Z$ under \cref{\asMP} by $\lawMP[h]$.
Then the log-likelihood ratio of $\lawMP[h]$ with respect to $\lawMP[0]$
is given by, using  $\varepsilon\sim N(0,\vareps)$ and $\varepsilon=\Delta Y- h Y_{-1}/(\sqrt{n}T)$,
\begin{align}
\log \frac{ \rd \lawMP[h]}{\rd \lawMP}&= \frac{h}{\sqrt n T} \Delta Y'\calA'\vareps^{-1}\Delta Y -\frac{h^2}{2 n T^2} \Delta Y'\calA' \vareps^{-1}\calA\Delta Y\\
&=:h\CSorigMP  -\frac{1}{2}h^2 \FIorigMP.
\end{align}


In this more complicated model, we simplify the central sequence and also the Fisher information in two steps. The first is analogous to the approximation in the PANIC setup, i.e., we replace variances by long-run variances. Note that thanks to our independence assumptions, the $nT\times nT$ covariance matrix of the $\varepsilon$ can be written as
\begin{align}\label{eqn:OmegaEpsilon}
\vareps=\var \varepsilon =\sum_{k=1}^K \left( \loading{k}\loading{k}\trans   \otimes \varF[,k]\right) + \vareta.
\end{align}
 Replacing $\varF[,k]$ by $\longrunvarFT{k} \id[T]$ and $\vareta[,i]$ by $\longrunvaretaT{i}\id[T]$ in \cref{eqn:OmegaEpsilon} we obtain the simplified versions of central sequence
\begin{align}\label{eqn:CS}
\CSMPLRV:=
 \frac{1}{\sqrt{n}T}  \Delta Y\trans \calA\trans \proxvareps^{-1}\Delta Y
-\frac{1}{\sqrt n}\sum_{i=1}^n \frac{\oslongrunvaretaT{i}}
{\longrunvaretaT{i} }
,
\end{align}
where  the $nT\times nT$ matrix $\proxvareps$ is defined by
\begin{align}\label{eqn:Omegavarepsilonprox}
\proxvareps:=\proxvarepsgen\kron\id[T]:=
\left( \Loading \longrunvarmatrixF \Loading\trans + \longrunvarmatrixeta\right)\kron \id[T]
,
\end{align}
with $\longrunvarmatrixeta=\diag( \longrunvaretaT{1},\dots,\longrunvaretaT{n})$ and $\longrunvarmatrixF=\diag(\longrunvarFT{1},\dots,\longrunvarFT{K})$. The following lemma demonstrates that applying these replacements to the central sequence and Fisher information do not affect their asymptotic behavior.
\begin{lemma}\label{lemma:proxyCS}
	Suppose that \cref{\asMP} hold.
Then we have, under $\lawMP$, $(\CSorigMP,\FIorigMP) =(\CSMPLRV,\frac{1}{2})+o_P(1)$.
\end{lemma}

Exploiting the Sherman-Morrison-Woodbury formula we obtain
\begin{align}\label{eqn:Omegavarepsilonproxinv}
\proxvareps^{-1}=\proxvarepsgen^{-1}\kron\id[T] =
\left(
\longrunvarmatrixeta^{-1} - \longrunvarmatrixeta^{-1} \Loading \left(
\longrunvarmatrixF^{-1} + \Loading\trans  \longrunvarmatrixeta^{-1} \Loading \right)^{-1}\Loading\trans\longrunvarmatrixeta^{-1}
\right)
\kron \id[T] .
\end{align}
%

Note that removing $\longrunvarmatrixF^{-1}$ from \cref{eqn:Omegavarepsilonproxinv} yields a projection matrix corresponding to `projecting out the factors'. 
Thus, basing a central sequence on such a projection matrix would simplify approximating it based on observables by removing the need to estimate $\longrunvarmatrixF^{-1}$ and, more importantly, by ensuring that the factors are projected out. The next lemma shows that using such a projection version ${\auxproxvarepsgen}^{-1}$ of $\proxvarepsgen^{-1}$ in the central sequence does not change its asymptotic behaviour.
\begin{lemma}\label{lemma:auxCS}
	Suppose that \cref{\asMP} hold.
Then we have, under $\lawMP$,
	 $\CSMPLRV =\auxCS+o_P(1)$, where
\begin{align}
  \auxCS=&\frac{1}{\sqrt n T} \Delta Y'\calA'({\auxproxvarepsgen}^{-1}\kron I_T)\Delta Y
 -\frac{1}{\sqrt n}\sum_{i=1}^n \frac{\oslongrunvaretaT{i}}{\longrunvaretaT{i}}\text{, with}\\
  {\auxproxvarepsgen}^{-1} =& \longrunvarmatrixeta^{-1} - \longrunvarmatrixeta^{-1} \Loading \left(
\Loading\trans  \longrunvarmatrixeta^{-1} \Loading \right)^{-1}\Loading\trans\longrunvarmatrixeta^{-1}.\label{auxinvdev}\end{align}
\end{lemma}
\subsection{Asymptotic normality}\label{asymnorm}

Having simplified each framework's central sequence and Fisher information  separately, we are now ready to show that they are asymptotically equivalent and the central sequences converge to a normal distribution.
We begin this section by showing that the central sequence in the MP framework is asymptotically equivalent to the one in the PANIC framework.\begin{lemma}\label{lemma:NormalityCS}
	Suppose that \cref{\asEither} hold.
Then we have, under $\lawPANIC$ and $\lawMP$,
$\auxCS =\etaCS+o_P(1)$.
\end{lemma}
Finally, we consider the weak limit of the central sequence $\etaCS$ (and therefore also of $\auxCS$), showing that both experiments are locally asymptotically normal.
\begin{proposition}\label{prop:limitexperiment}
Suppose that \cref{assumption:serial correlation,assumption:factor_loadings,assumption:rates,assumption:framework} hold.
Then we have, under $\lawPANIC$ and $\lawMP$,
  $\etaCS\vto N(0,J)$ with $J=\frac{1}{2}$.
\end{proposition}
\begin{remark}\label{remark:equivalence} Under the null hypothesis, the model equations of both models coincide. Hence, the additional distributional  \cref{ass:MP} implies that under the null, the MP framework is a special case of the PANIC framework. Therefore, it is sufficient to show the desired convergence for $\lawPANIC$.  This principle applies to all calculations under the hypothesis. As the central sequences are equal as well and thanks to the LAN result below, it even extends to many calculations under alternatives, through Le Cam's Third Lemma.
	\end{remark}

	\Cref{prop:limitexperiment} is an important result as it establishes that the unit root testing problem in both models is locally asymptotically normal, i.e., it is asymptotically equivalent to testing $h=0$ against $h<0$ based on one observation $X\sim N(Jh,J)$. This equivalence prescribes how to perform asymptotically optimal inference and yields the asymptotic local power envelope and the power functions of various test statistics: The asymptotic representation theorem \cite[see, for example][Ch. 9]{vdVaart2000} implies that in our framework no unit root test can have higher power than the optimal test in the limit experiment.
This best test is clearly rejecting for small values of $X$, leading to a power (for a level-$\alpha$ test) of $
  \Phi(\Phi^{-1}(\alpha)-J^{1/2}h).$
Thus, with $J=1/2$, this constitutes the power envelope for our unit root testing problems:\footnote{As this section assumes the nuisance parameters to be known, for now we can only present an upper bound on the attainable power. In \cref{sec:UMPtest} we show that the power envelope of \cref{cor:theoretical_power_envelope} can be attained.}
\begin{corollary}\label{cor:theoretical_power_envelope}
Suppose that \cref{assumption:serial correlation,assumption:factor_loadings,assumption:rates,ass:BN}  hold. Let $\phi_{n,T}=\phi_{n,T}( Z_{11},\dots,Z_{nT})$ be a sequence of tests and denote their powers, under $\lawPANIC[h]$, by
$\pi_{n,T}(h)$.
If the sequence $\phi_{n,T}$ is asymptotically of level $\alpha\in(0,1)$,  i.e.
$ \limsup_{n,T\to\infty}   \pi_{n,T}(0)\leq
\alpha$,  we have, for all $h\leq 0$,
\begin{align}\label{eqn:powerenvelope}
\limsup_{n,T\to\infty} \pi_{n,T}(h)\leq \Phi\left(\Phi^{-1}(\alpha)-\frac{h}{\sqrt{2}}\right).
\end{align}
Replacing \cref{ass:BN} by \cref{ass:MP}, the same bound applies to powers under $\lawMP[h]$.
\end{corollary}
The above power envelope would be reached by any of our previously introduced central sequences.\footnote{This always holds in LAN experiments and follows from Le Cam's Third Lemma \cite[see, for example,][Ch. 6]{vdVaart2000}.} In the next section we show that we can approximate these central sequences based on observables, yielding a feasible test that attains the asymptotic power envelope.


\section{An Asymptotically UMP Test}\label{sec:UMPtest}

In the previous section we derived a testing procedure that reaches the power envelope for the unit root testing problem. This test, however, is not feasible when the nuisance parameters are unknown. In this section, we demonstrate how to estimate the nuisance parameters to obtain a feasible version that also attains the power envelope. We provide a feasible version of $\auxCS$, which is motivated by the likelihood ratio in the MP experiment. As \cref{auxinvdev} projects out the factors, basing our feasible version on $\auxCS$ instead of $\etaCS$ spares us the approximation of the idiosyncratic parts. 

Recalling our LAN results in \cref{sec:limexp} and that the central sequences are asymptotically equivalent across the two setups (see \cref{lemma:NormalityCS}) it is clear that a feasible version of $\auxCS$ would be optimal.
 Therefore, we show that replacing all nuisance parameters with estimates does not change the limiting behavior of $\auxCS$. Specifically, we need estimates $\hatLoading$ of the factor loadings, as well as estimates $\hatoslongrunvareta{i}$ and $\hatlongrunvareta{i}$ of the (one-sided) long-run variances of each idiosyncratic part. The feasible test statistic is then \begin{align}\label{hatCS}
  \hatCS =& \frac{1}{\sqrt n T} \sum_{t=2}^T \sum_{s=2}^{t-1} \Delta Z_{\cdot, s}\trans \hatproxvarepsgen^{-1} \Delta Z_{\cdot, t} -\frac{1}{\sqrt n}\sum_{i=1}^n \frac{\hatoslongrunvareta{i}}{\hatlongrunvareta{i}}\text{, where}\\
\label{hatinvdev}
  \hatproxvarepsgen^{-1} :=& \hatlongrunvarmatrixeta^{-1} - \hatlongrunvarmatrixeta^{-1}\hat\Loading (\hat\Loading\trans \hatlongrunvarmatrixeta^{-1} \hat\Loading)^{-1}\hat\Loading\trans\hatlongrunvarmatrixeta^{-1}.
\end{align}

\begin{assumption}\label{assumption:estimators}
Let $\hatoslongrunvareta{i}$, $\hatlongrunvareta{i}$ and $\hat\Loading$  be 
estimators of $\oslongrunvareta{i}$, $\longrunvareta{i}$ and $\Loading$ satisfying, under $\lawMP$ and $\lawPANIC$,
\begin{enumerate}
 \item  \label{nuisanceDeltas}$\sup_{i\in\mathbb{N}} |\hatoslongrunvareta{i} - \oslongrunvareta{i}|^2 = o_P(1/ n) $, 
  
 \item \label{nuisanceOmegas}$\sup_{i\in\mathbb{N}} |\hatlongrunvareta{i} - \longrunvareta{i}|^2 = o_P(1/ n) $, and
 \item  \label{nuisanceLoading}for a $K\times K$ matrix $H_K$ satisfying $\Frob{H_K}=O_P(1)$ and $\Frob{H_K^{-1}}=O_P(1)$, we have $\Frob{\Loading H_K -\hat \Loading}=o_P(1)$.
\end{enumerate}
\end{assumption}
Under suitable restrictions on the bandwidth and the kernel, \cref{nuisanceDeltas,nuisanceOmegas} hold for kernel spectral density estimates; see Remark 2.9 in \cite{MPP2014}. \cref{nuisanceLoading} is stronger that the results in \cite{MoonPerron2004}, so we show in \Cref{lem:FactEst} that it indeed holds under our assumptions.

\begin{lemma}\label{lem:FactEst}
Let $\bar\Loading$ be $\sqrt{n}$ times the $n\times K$ matrix containing the $K$ orthonormal eigenvectors corresponding to the $K$ largest eigenvalues of $\frac{\Delta\tilde Z'\Delta\tilde Z}{nT}$. Take $
  \hat\Loading=\frac{\Delta\tilde Z'\Delta\tilde Z}{nT}\bar\Loading.$
There exists a $K\times K$ matrix $H_K$ such that, under $\lawMP$ and $\lawPANIC$, $
  \Frob{\Loading H_K -\hat \Loading}=o_p(1)$ 
and both $\Frob{H_K}$ and $\Frob{H_K^{-1}}$ are $O_P(1)$.
\end{lemma}

\begin{remark}
These factor estimates are the same 	as those used in \cite{MoonPerron2004} and correspond to factor estimates based on classical principal component analysis.
\end{remark}

\begin{remark}\label{rotations}	
The factors are only identified up to a `rotation' $H_K$. Note that $\auxCS$ is (indeed) invariant under such rotations, as  ${\auxproxvarepsgen}^{-1}$ also equals
\begin{align}
  \longrunvarmatrixeta^{-1} - \longrunvarmatrixeta^{-1} \Loading H_K\left( H_K'
\Loading\trans  \longrunvarmatrixeta^{-1} \Loading H_K\right)^{-1}H_K'\Loading\trans\longrunvarmatrixeta^{-1}.
\end{align}
\end{remark}

\begin{lemma} \label{lemma:hatCS}
Under \cref{assumption:estimators,assumption:factor_loadings,assumption:framework,assumption:serial correlation,assumption:rates}
we have, under $\lawMP$ and $\lawPANIC$, $
\hatCS=\auxCS+o_P(1).$
\end{lemma}
Although \cref{lemma:hatCS} only concerns adaptivity under the null hypothesis $\rH_0$, we can use Le Cam's First Lemma to obtain that, thanks to contiguity, also under $\lawMP[h]$ or $\lawPANIC[h]$, $\hatCS$ has the same limiting distribution as $\auxCS$, so that tests based on $\hatCS$ will be uniformly most powerful. Formally, the size and power properties of our optimal test follow from the following theorem.

\begin{theorem}Let $\tUMP=\sqrt{2}\hatCS$. 
Under \cref{assumption:estimators,assumption:factor_loadings,assumption:framework,assumption:serial correlation,assumption:rates}  we have, under $\lawMP[h]$ and $\lawPANIC[h]$,
\begin{align}
\tUMP \vto N\left(\frac{1}{\sqrt{2}} h,1\right).
\end{align}
Rejecting $\rH_0$ for $\tUMP\leq \Phi^{-1}(\alpha)$, $\alpha\in (0,1)$,
and $\tUMP$ an asymptotic power of leads to an asymptotic power of $\Phi\left(\Phi^{-1}(\alpha)-\frac{h}{\sqrt{2}}\right)$, implying that $\tUMP$ is asymptotically uniformly most powerful.
\end{theorem}
\begin{remark}
The asymptotic size of our test can also be obtained under much weaker assumptions not exploiting Gaussianity, see \cref{footnoteGaussianity,remarkGaussianity}. In such a situation, our test is still valid although perhaps nonoptimal. For optimal inference with non-Gaussian innovations a new analysis of the likelihood ratio would be needed, but this is not feasible here.
\end{remark}

\begin{remark}
	Note that the limiting distribution of $\tUMP$ does not depend on the autocorrelations or the heterogeneity of the long-run variances.  This shows that the decrease in asymptotic power attributed to these features, for example in Remark 2 of \cite{Westerlund2015} was due to the specific tests under consideration rather than being a feature of the unit root testing problem.
\end{remark}

\begin{remark}
	Note that $\hatCS$ only involves differenced data, so that our test is invariant with respect to the incidental intercepts $m_i$. 
\end{remark}

Here is one way to obtain the UMP test in practice:
\begin{enumerate}
\item Compute an estimator $\hat K$ of the number of common factors on the basis of the observations $\Delta Z_{\cdot t}$, $t=2,\dots,T$ using information criteria from \cite{BaiNg2002}.\footnote{As $\jointlimits$, these criteria select the correct number of factors with probability one. Therefore, we can treat the number of factors as known in our asymptotic analyses.}
\item Use the observations $\Delta Z_{\cdot t}$, $t=2,\dots,T$, and $\hat K$ to determine the factor loadings  $\hatLoading$ and the factor residuals  $\hat\eta_{\cdot t}$, $t=2,\dots,T$,  using principal components.
\item Determine estimates $\hatlongrunvareta{i}$ of $\longrunvareta{i}$  and estimates $\hatoslongrunvareta{i}$ of $\oslongrunvareta{i}$
from $\hat\eta_{\cdot t}$, $t=2,\dots,T$, using kernel spectral density estimates.  
Let $\hat\Omega=\diag(\hatlongrunvareta{1},\dots,\hatlongrunvareta{n} )$.
\item Calculate the estimated central sequence $\hatCS$ as in \cref{hatCS} and reject when $\tUMP=\sqrt{2}\hatCS\leq \Phi^{-1}(\alpha)$. Alternatively, based on small sample considerations, also estimate the empirical  Fisher information \begin{align}
	\hatJ:= \frac{1}{ n T^2}\sum_{t=2}^T \sum_{s=2}^{t-1} \Delta Z_{\cdot, s}\trans \hatproxvarepsgen^{-1}  \sum_{u=2}^{t-1}\Delta Z_{\cdot, u},
\end{align}
and reject the null hypothesis when $\tUMPEMP:=\hat \Delta_{n,T}/\sqrt{\hatJ}  \leq \Phi^{-1}(\alpha)$.
\end{enumerate}
\begin{remark}
Although the uniformly most powerful test $\tUMP$ does not require a complicated estimate of the known $J=1/2$, it can be undersized in small samples, whereas the empirical version $\tUMPEMP$ behaves very well in most DGPs, both in terms of size and power. Thus we recommend to use the $\tUMPEMP$ in small samples. See \cref{sec:MC} for details. 
\end{remark}

\section{Comparing Powers Across Tests and Frameworks}\label{sec:tests_comparison}
This section derives the asymptotic powers of commonly used tests in both the \cite{MoonPerron2004} and the \cite{BaiNg2004} frameworks. We start by formalizing our observation that local powers are equal across the two frameworks.
\begin{corollary}\label{equivalence}
	Let $t_{n,T}$ be a test statistic that, under $\lawPANIC$, converges in distribution jointly with $\etaCS$. Then, for all $x\in\mathds{R}$, and all $h$, \begin{align}
 \lim_{\jointlimits} \lawMP[h][t_{n,T}\le x]=\lim_{\jointlimits}\lawPANIC[h][t_{n,T}\le x].
\end{align}
If, more specifically, $t_{n,T}\overset{\lawPANIC}{\to} N(\mu,\sigma^2)$ and if $t_{n,T}$ and $\etaCS$ are jointly asymptotically normal under $\lawPANIC$ with asymptotic covariance $\sigma_{\Delta,t}$, its limiting distribution under local alternatives is given by 
\begin{align}
  t_{n,T}\overset{\lawPANIC[h]}{\to} N(\mu+h\sigma_{\Delta,t},\sigma^2)\text{, and } t_{n,T}\overset{\lawMP[h]}{\to} N(\mu+h\sigma_{\Delta,t},\sigma^2).
\end{align}
Thus, rejecting for small values of $t_{n,T}$ leads to an asymptotic power for a level-$\alpha$ test of $\Phi(\Phi^{-1}(\alpha)-h\sigma_{\Delta,t}/\sigma)$ in both frameworks.
\end{corollary}
Once again, our result on the asymptotic equivalence of the two experiments allows us to obtain results for both frameworks at the same time. By demonstrating the joint normality under the null as in \cref{equivalence} we obtain simple proofs of the powers of commonly used tests in these frameworks, without ever relying on triangular array calculations. 

To show the elegance of this approach, we include here the full proof of the first part of this lemma. The second part follows immediately from a more specific version of Le Cam's third lemma, which directly prescribes the desired normal distribution under alternatives. We can use this simple way to obtain powers under local alternatives thanks to our LAN results of \cref{sec:limexp}.
\begin{proof}
	Denote the weak limit of $(t_{n,T},\etaCS)$ under $\lawMP$ by $(t,\Delta)$. Thanks to our results in \cref{sec:limexp}, both $(t_{n,T},\frac{ \rd \lawPANIC[h]}{\rd \lawPANIC})$ and $(t_{n,T},\frac{ \rd \lawMP[h]}{\rd \lawMP})$ converge in distribution to $(t,\exp(h\Delta-h^2/4))$. By a general form of Le Cam's third lemma, the distribution of $t_{n.T}$ under local alternatives only depends on this joint limiting law and is thus equal across the two frameworks.\footnote{See Theorem 6.6 in \cite{vdVaart2000}. 
}
	\end{proof}


%
%
Before we apply \cref{equivalence} to derive asymptotic powers, we first describe the relevant test statistics in some detail. We focus on the tests proposed in \cite{BaiNg2010} (`BN tests') and \cite{MoonPerron2004} (`MP tests').
Following these papers, we denote
\begin{align}\label{notation:add_BN}
 \lrv^2 &= \lim_{n \to \infty} \frac{1}{n} \sum_{i=1}^n \longrunvareta{i}, \enskip
 \squaredomegas^4 = \lim_{n \to \infty} \frac{1}{n} \sum_{i=1}^n \left(\longrunvareta{i}\right)^2,\enskip
\oslrv = \lim_{n\to\infty} \frac{1}{n} \sum_{i=1}^n \oslongrunvareta{i},
\end{align}
all assumed to be positive, and their estimated counterparts
\begin{align}
 \hatlrv^2 =  \frac{1}{n} \sum_{i=1}^n \hatlongrunvareta{i}, \enskip
 \hatsquaredomegas^4 =  \frac{1}{n} \sum_{i=1}^n \left(\hatlongrunvareta{i}\right)^2 \text{, and }
 \hatoslrv =  \frac{1}{n} \sum_{i=1}^n \hatoslongrunvareta{i}.\end{align}
Finally, we define $\omega^4 = (\omega^2)^2$ and $\hat\omega^4 = (\hat\omega^2)^2$.

Both the MP and BN tests rely on a two stage procedure. In the first stage, the unobserved idiosyncratic innovations $E$ are estimated. Subsequently, a pooled regression procedure is used to estimate the (pooled) autoregression parameter.
This pooled estimator is then used to construct a $t$-test.
 The main difference between the MP and the BN procedures  lies in the way the idiosyncratic innovations are estimated.

\cite{BaiNg2010} propose to estimate the idiosyncratic errors $E$ by the PANIC approach introduced in \cite{BaiNg2004},
which in turn relies on principal component analysis applied to the differences $\Delta Y_{it}$. Denoting this estimator of $\bfe$ by $\BNbfe$, the BN tests are 
\begin{align}
P_a =& \frac{\sqrt n T (\BNrho-1)}{\sqrt{2\hatsquaredomegas^4/\hatlrv^4}} \text{ and } \\
P_b =& \sqrt n T (\BNrho-1)\sqrt{\frac{1}{nT^2}\sum_{i=1}^n \BNbfem[i]^\prime\BNbfem[i]\frac{\hatlrv^2}{\hatsquaredomegas^4}}\text{, where}\\
\BNrho =& \frac{\sum_{i=1}^n \BNbfem[i]^\prime\BNbfe[i]- nT \hatoslrv}{\sum_{i=1}^n \BNbfem[i]^\prime \BNbfem[i]}\end{align}
is a  bias-corrected pooled estimator for the autoregressive coefficients.

\begin{remark}
	Recall that $\tUMPEMP$ is a modification of $\tUMP$ that replaces the asymptotic Fisher Information $J=1/2$, with its finite sample equivalent in the MP setup, $\FIMP$. The resulting statistics can be considered a version of $P_b$: In the case of homogeneous long-run variances, inserting the true long-run variances into $\tUMPEMP$ yields $P_b$. Conversely, $\tUMPEMP$ is a version of $P_b$ that takes into account the heterogeneity in the long-run variances.
\end{remark}

The MP tests are based on a different estimator of $\rho$. 
The  idiosyncratic components $\bfe[i]$ are estimated by projecting the data on the space orthogonal to the common factors. Let $\hat\Loading$ be a consistent estimators for $\Loading$ as defined in \cite[p. 89-90]{MoonPerron2004}, and $Y_{\cdot,t} = (Y_{1t}, \dots, Y_{nt})^\prime$. Then the MP test statistics are given by
\begin{align}
t_a =& \frac{\sqrt n T (\MPrho-1)}{\sqrt{2\hatsquaredomegas^4/\hatlrv^4}},\text{ and } \\ t_b =& \sqrt n T (\MPrho-1)\sqrt{\frac{1}{nT^2}\sum_{t=1}^T Y_{\cdot,t-1}^\prime Q_{\hat\gamma} Y_{\cdot,t-1}\frac{\hatlrv^2}{\hatsquaredomegas^4}}\text{, where}\\
\MPrho =& \frac{\sum_{t=1}^T Y_{\cdot,t}^\prime Q_{\hat\gamma} Y_{\cdot,t-1} - nT \hatoslrv}{\sum_{t=1}^T Y_{\cdot,t-1}^\prime Q_{\hat\gamma} Y_{\cdot,t-1}},
\text{ and } Q_{\hat\gamma} = I - \hatGamma(\hatGamma^\prime\hatGamma)^{-1}\hatGamma^\prime.\end{align}

We are now ready to compute the asymptotic behaviour of the MP and BN tests under local alternatives by an application of \cref{equivalence}. 
The power of the MP tests in the MP framework has been derived in \cite{MoonPerron2004} and that of the BN tests in the PANIC framework has been derived in \cite{Westerlund2015}.
Given our LAN result, we can provide simple independent proofs of these results. These rely on the second part of \cref{equivalence}; we demonstrate the required joint asymptotic normality in a supplementary appendix.
More importantly, our approach also leads to new results, namely the asymptotic powers of the MP test in the PANIC framework and the asymptotic powers of the BN tests in the MP framework. In fact, those results can be considered an immediate consequence of the first part of \cref{equivalence} and the existing power results in the literature.


\begin{proposition}\label{prop:local_powers_BN_and_MP}
Suppose that \cref{assumption:serial correlation,assumption:factor_loadings,assumption:rates,assumption:framework,assumption:estimators} hold.
Then, under $\lawPANIC[h]$ or $\lawMP[h]$, as $\jointlimits$, the test statistics $P_a, P_b, t_a$, and $t_b$ all converge in distribution to a normal distribution with mean $h\sqrt{\frac{\omega^4}{2\phi^4}}$ and variance one. Rejecting for small values of any of these statistics leads to an asymptotic power for a level-$\alpha$ test of $\Phi(\Phi^{-1}(\alpha)-h\sqrt{\frac{\omega^4}{2\phi^4}})$ in both frameworks.
\end{proposition}

\begin{remark}
It turns out that the powers are equal, no matter which test statistic and which framework is considered. We have discussed in some detail that, for a given test, the equality of powers across frameworks is a general phenomenon. The fact that in each framework, the power of the MP tests is equal to that of the BN tests, on the other hand, is a `coincidence'.
Originally, the MP tests have been developed for the MP experiment, whereas the BN tests are designed for the PANIC experiment. 
It has been noted in \cite{BaiNg2010} that the MP tests are valid in term of size in the PANIC setup for testing the idiosyncratic component of the innovation for a unit root but their (local and asymptotic) power in the PANIC framework has not been considered. More discussion on the use of the MP tests in the PANIC setup can be found in \cite{BaiNg2010} and \cite{GengenbachPalmUrbain2010}. Similarly, to the best of our knowledge there are no studies on the power of the BN tests in the MP framework.
\end{remark}
  The Cauchy-Schwarz inequality implies $\frac{\omega^4}{\phi^4} \leq  1$, thus Proposition~\ref{prop:local_powers_BN_and_MP} shows that, in general, the local asymptotic power of the MP and BN tests lies below the power envelope. In fact, they are all asymptotically UMP only when $\frac{\omega^4}{\phi^4} = 1$. This condition is satisfied when the long-run variances of the idiosyncratic shocks $\eta_{it}$ are homogeneous across $i$. The proposed test $\tUMP$ is asymptotically UMP irrespective of possible heterogeneity. In \cref{sec:MC} we assess whether the asymptotic power gains, compared to the MP and BN tests, are also reflected in finite samples for realistic parametric settings.

\section{Simulation results}\label{sec:MC}
This section reports the results of a Monte-Carlo study with three main goals: firstly, to assess the finite sample performance of our proposed test $\tUMP$, secondly, to see how the asymptotic equivalence between the \cite{MoonPerron2004} and PANIC setups is reflected in finite samples, and, finally, to check the robustness of our results to deviations from our assumptions. 
\subsection{The DGPs}
We generate the data from \cref{Zeq,Yeq,Eeq,Feq} with $m_i=0$.\footnote{Recall that our tests are invariant with respect to $m_i$.} Using sample sizes $n=25,50,100$ and $T=n,2n,4n$, we simulate both the MP and the PANIC setups. Recall that, for a local alternative $h$, we take $\rho=1+\frac{h}{\sqrt{n}T}$ in both setups. In the MP case we also set $\rho_k=\rho$, whereas in the PANIC case we set $\rho_k=1$ under the null and all alternatives. 
The factor loadings $\Lambda$ are drawn from a normal distribution with mean $K^{-1/2}$ and covariance matrix $K^{-1}I_K$.\footnote{As done in \cite{MoonPerron2004}, we scale by $\sqrt{K}$ to ensure the contribution of the factors is comparable across specifications.}  Most of the simulations are run with $K=1$ but we also explore what happens with more factors. Throughout this section we assume the number of factors to be known.\footnote{This number can be estimated consistently, so this makes no difference for the asymptotic analysis. See, for example, Section 2.3 in \cite{MoonPerron2004} and Section 5 in \cite{BaiNg2010} for a discussion of this issue.} For the innovation processes $f_{kt}$ and $\eta_{it}$ we examine Gaussian i.i.d., MA(1), and AR(1) processes. We fix the MA or AR parameter at 0.4 and set the variance such that the long-run variances of the $f_{kt}$ equal one, and the long-run variance of the $\eta_{it}$ is $\omega_i^2$. The $\omega_i^2$ are drawn i.i.d.\ from a lognormal distribution whose parameters are chosen to match different values of $\omega^4/\phi^4$ and a mean of one.\footnote{Recall from \cref{sec:tests_comparison} that the asymptotic relative efficiency of the existing tests compared to our UMP test depends on the heterogeneity of the long-run variances and more specifically on the ratio $\omega^4/\phi^4$. Therefore, the sample size at which it becomes worthwhile to estimate the heterogeneous long-run variances (i.e., use the asymptotically UMP tests suggested here) mainly depends on this ratio.  
 We present simulation results for $\sqrt{\omega^4/\phi^4}$ between 0.6 and 1, where lower values indicate more heterogeneity. A cursory look at a few typical applications reveals that these ratios are mostly between 0.6 and 0.8 and match the skewed nature of the lognormal distribution.}

\begin{figure}[!htb]
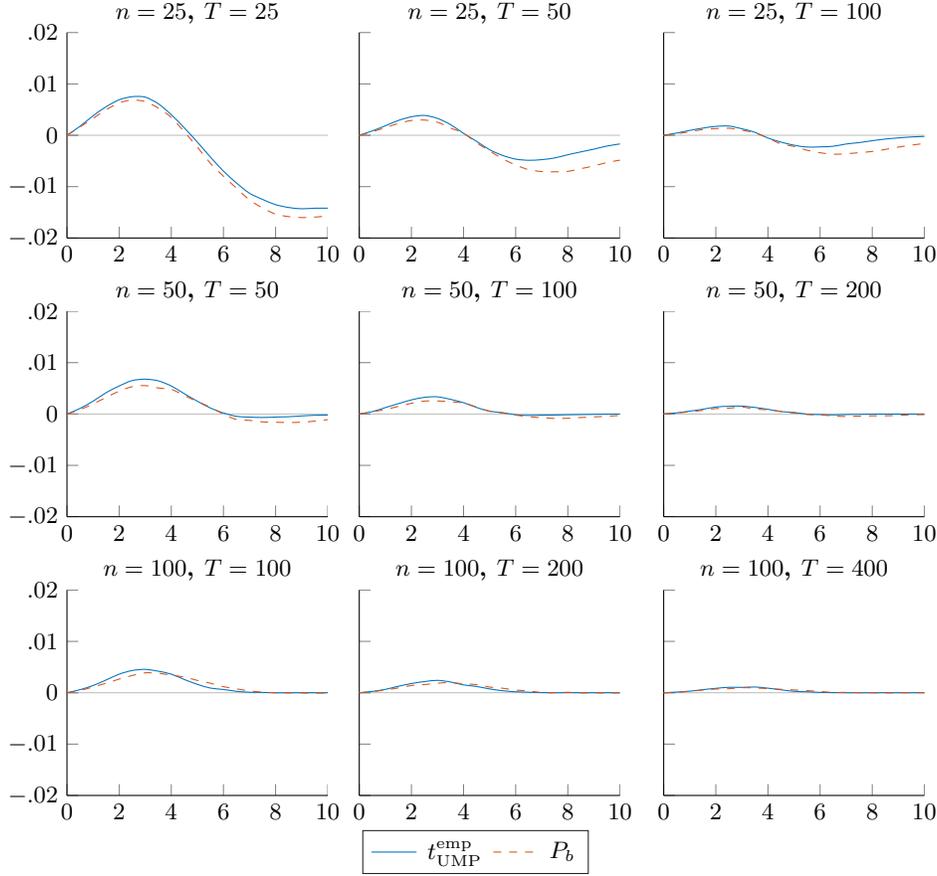
\footnotesize\begin{subfigure}[b]{0.36\textwidth}\input{graphs/stamps/NpowersMPDif25_25_1_0.8_iid_0.4_iid_0.4_MP_DifMP.tex}\end{subfigure}\begin{subfigure}[b]{0.32\textwidth}\input{graphs/stamps/NpowersMPDif50_25_1_0.8_iid_0.4_iid_0.4_MP_DifMP.tex}\end{subfigure}\begin{subfigure}[b]{0.32\textwidth}\input{graphs/stamps/NpowersMPDif100_25_1_0.8_iid_0.4_iid_0.4_MP_DifMP.tex}\end{subfigure}\\ \begin{subfigure}[b]{0.36\textwidth}\input{graphs/stamps/NpowersMPDif50_50_1_0.8_iid_0.4_iid_0.4_MP_DifMP.tex}\end{subfigure}\begin{subfigure}[b]{0.32\textwidth}\input{graphs/stamps/NpowersMPDif100_50_1_0.8_iid_0.4_iid_0.4_MP_DifMP.tex}\end{subfigure}\begin{subfigure}[b]{0.32\textwidth}\input{graphs/stamps/NpowersMPDif200_50_1_0.8_iid_0.4_iid_0.4_MP_DifMP.tex}\end{subfigure}\\ \begin{subfigure}[b]{0.36\textwidth}\input{graphs/stamps/NpowersMPDif100_100_1_0.8_iid_0.4_iid_0.4_MP_DifMP.tex}\end{subfigure}\begin{subfigure}[b]{0.32\textwidth}\input{graphs/stamps/NpowersMPDif200_100_1_0.8_iid_0.4_iid_0.4_MP_DifMP.tex}\end{subfigure}\begin{subfigure}[b]{0.32\textwidth}\input{graphs/stamps/NpowersMPDif400_100_1_0.8_iid_0.4_iid_0.4_MP_DifMP.tex}\end{subfigure}\\ \begin{subfigure}[b]{\textwidth}\centering\input{graphs/stamps/lNpowersMPDif400_100_1_0.8_iid_0.4_iid_0.4_MP_DifMP.tex}\end{subfigure}\normalsize\caption{\label{DifMPNpowersMPDifiidBN0.811}Difference between powers in the MP vs the PANIC framework as a function of $-h$ with i.i.d.\ factor innovations and i.i.d.\ idiosyncratic parts and $\sqrt{\omega^4/\phi^4}=0.8$. Based on \num{1000000} replications.}\end{figure}
\subsubsection{The test statistics}
In addition to the tests proposed in \cref{sec:UMPtest}, $\tUMP$ and $\tUMPEMP$, we consider the MP tests of \cite{MoonPerron2004} and the BN tests of \cite{BaiNg2010}. However, the powers and sizes of the (MP) $t_b$ and (BN) $P_b$ tests were very similar also in finite samples, so we only report results for $P_b$. We omit the comparison with $P_a$ and $t_a$ since they tend to show large biases in terms of size (see, for example, the Monte Carlo studies in \cite{GengenbachPalmUrbain2010} and \cite{BaiNg2010}).  


The sizes of all considered tests are highly sensitive to estimation of the (one-sided) long-run variances. 
We have considered a variety of methods, for example, using a Bartlett or quadratic spectral kernel and selection of the bandwidth according to the \cite{NeweyWest1994} or the \cite{Andrews1991} rule with/without various forms of prewhitening.  Whereas the differences from using different kernels are small, the selection of both the bandwidth and the prewhitening are essential. Our preferred method employs a Bartlett kernel with prewhitening.\footnote{As in \cite{MPP2014}, the prewhitening model is selected based on the BIC between four simple ARMA models. 
} There is a size-power tradeoff between using the \cite{Andrews1991} and the \cite{NeweyWest1994} bandwidth selection: The \cite{Andrews1991} bandwidth leads to higher powers for the smallest sample sizes, but an oversized test when the innovations have a strong MA component. The decision which bandwidth to use thus depends on the preferences of the researcher. In this section, all results are based on the \cite{Andrews1991} bandwidth. However, the sizes and powers based on the \cite{NeweyWest1994} bandwidth can be found in a supplementary appendix.

\begin{figure}[htb!]\footnotesize\begin{subfigure}[b]{0.3466\textwidth}\input{graphs/stamps/Npowers25_25_1_0.8_iid_0.4_iid_0.4_BN_.tex}\end{subfigure}\begin{subfigure}[b]{0.3267\textwidth}\input{graphs/stamps/Npowers50_25_1_0.8_iid_0.4_iid_0.4_BN_.tex}\end{subfigure}\begin{subfigure}[b]{0.3267\textwidth}\input{graphs/stamps/Npowers100_25_1_0.8_iid_0.4_iid_0.4_BN_.tex}\end{subfigure}\\ \begin{subfigure}[b]{0.3466\textwidth}\input{graphs/stamps/Npowers50_50_1_0.8_iid_0.4_iid_0.4_BN_.tex}\end{subfigure}\begin{subfigure}[b]{0.3267\textwidth}\input{graphs/stamps/Npowers100_50_1_0.8_iid_0.4_iid_0.4_BN_.tex}\end{subfigure}\begin{subfigure}[b]{0.3267\textwidth}\input{graphs/stamps/Npowers200_50_1_0.8_iid_0.4_iid_0.4_BN_.tex}\end{subfigure}\\ \begin{subfigure}[b]{0.3466\textwidth}\input{graphs/stamps/Npowers100_100_1_0.8_iid_0.4_iid_0.4_BN_.tex}\end{subfigure}\begin{subfigure}[b]{0.3267\textwidth}\input{graphs/stamps/Npowers200_100_1_0.8_iid_0.4_iid_0.4_BN_.tex}\end{subfigure}\begin{subfigure}[b]{0.3267\textwidth}\input{graphs/stamps/Npowers400_100_1_0.8_iid_0.4_iid_0.4_BN_.tex}\end{subfigure}\\ \begin{subfigure}[b]{\textwidth}\centering\input{graphs/stamps/lNpowers400_100_1_0.8_iid_0.4_iid_0.4_BN_.tex}\end{subfigure}\normalsize\caption{\label{NpowersiidBN0.811}Size-corrected power of unit-root tests as a function of $-h$ for varying sample sizes in the PANIC framework with i.i.d.\ factor innovations and i.i.d.\ idiosyncratic parts and $\sqrt{\omega^4/\phi^4}=0.8$. Based on \num{100000} replications.}\end{figure}
\subsection{Sizes}
\Cref{table:sizesAn1_0.4_0.4_BN_} reports the sizes of our tests for the baseline DGP based on the Andrews bandwidth.
Many other specifications can be found in the supplemental appendix. Recall that the sizes depend considerably on how the long-run variances are estimated. Using the method described above, the sizes of $\tUMPEMP$ reasonable across most DGPs and generally comparable to those of $P_b$. $\tUMP$, on the other hand, is undersized in many specifications, so that we focus on its empirical version $\tUMPEMP$ in the remainder. Only in the MA(1) example, both $\tUMPEMP$ and $P_b$ are oversized ($\tUMPEMP$ is more oversized for the smallest sample sizes and marginally less oversized in the larger ones). Thus, when a strong MA component is suspected, we recommend to use tests based on the \cite{NeweyWest1994} bandwidth. Generally, the  \cite{NeweyWest1994} bandwidth provides better sizes, especially in the MA case. However, small sample powers are slightly lower. Both sizes and powers based on the \cite{NeweyWest1994} bandwidth can be found in a supplementary appendix.

\sisetup{table-format=2.1}
\begin{table}
\resizebox{\textwidth}{!}{
\begin{tabular}{@{}S[table-figures-decimal=0,table-figures-integer=3]S[table-figures-decimal=0,table-figures-integer=3]S[table-figures-decimal=1,table-figures-integer=1]cSSScSSScSSS@{}}
\toprule
&&&&\multicolumn{3}{c}{i.i.d.}&&\multicolumn{3}{c}{AR(1)}&&\multicolumn{3}{c}{MA(1)}\\ \cmidrule{5-7} \cmidrule{9-11} \cmidrule{13-15}
{$n$} & {$T$} & {$\sqrt{\omega^4/\phi^4}$} & {} & {$\tUMP$} & {$\tUMPEMP$} & {$P_b$} & {} & {$\tUMP$} & {$\tUMPEMP$} & {$P_b$} & {} & {$\tUMP$} & {$\tUMPEMP$} & {$P_b$} \\
\midrule
25.0000 & 25.0000 & 0.6000 &  & 0.6278 & 2.8300 & 3.1393 &  & 1.7592 & 4.5324 & 4.1730 &  & 2.2143 & 6.9812 & 5.6133 \\
25.0000 & 50.0000 & 0.6000 &  & 1.3642 & 4.6906 & 4.0342 &  & 1.7387 & 4.8714 & 3.5933 &  & 3.0835 & 8.9345 & 6.2100 \\
25.0000 & 100.0000 & 0.6000 &  & 1.8153 & 5.4963 & 4.6073 &  & 2.3184 & 6.0900 & 4.1149 &  & 3.8533 & 10.0570 & 6.6694 \\
50.0000 & 50.0000 & 0.6000 &  & 1.9867 & 4.3321 & 3.6601 &  & 2.5053 & 4.4973 & 3.5272 &  & 5.2619 & 9.9448 & 6.6498 \\
50.0000 & 100.0000 & 0.6000 &  & 2.5939 & 5.1117 & 4.2180 &  & 2.8715 & 5.2431 & 3.7118 &  & 6.0827 & 10.9587 & 6.9966 \\
50.0000 & 200.0000 & 0.6000 &  & 2.9027 & 5.4544 & 4.5615 &  & 3.3695 & 5.9381 & 4.0929 &  & 5.2512 & 9.1976 & 6.0557 \\
100.0000 & 100.0000 & 0.6000 &  & 3.2141 & 4.9931 & 4.1720 &  & 3.3284 & 4.8677 & 3.8115 &  & 9.0711 & 13.1456 & 8.2002 \\
100.0000 & 200.0000 & 0.6000 &  & 3.5535 & 5.2922 & 4.4927 &  & 3.6555 & 5.2959 & 4.0748 &  & 6.9634 & 9.9569 & 6.6438 \\
100.0000 & 400.0000 & 0.6000 &  & 3.6374 & 5.3489 & 4.5417 &  & 4.3000 & 6.1066 & 4.4998 &  & 4.9484 & 7.0682 & 5.1413 \\
25.0000 & 25.0000 & 0.8000 &  & 0.8577 & 3.1353 & 3.4974 &  & 1.8084 & 4.2566 & 4.6914 &  & 2.3562 & 6.6509 & 6.3883 \\
25.0000 & 50.0000 & 0.8000 &  & 1.7769 & 5.0621 & 4.6238 &  & 1.6805 & 4.3540 & 3.9552 &  & 3.1463 & 8.2977 & 7.1776 \\
25.0000 & 100.0000 & 0.8000 &  & 2.2970 & 5.8434 & 5.2254 &  & 2.1766 & 5.2984 & 4.5579 &  & 3.8913 & 9.2820 & 7.7690 \\
50.0000 & 50.0000 & 0.8000 &  & 2.3612 & 4.6210 & 4.2170 &  & 2.3909 & 4.1594 & 4.2157 &  & 5.1408 & 9.2920 & 8.2863 \\
50.0000 & 100.0000 & 0.8000 &  & 3.0065 & 5.4040 & 4.8028 &  & 2.6241 & 4.6078 & 4.2616 &  & 5.8681 & 10.1062 & 8.4963 \\
50.0000 & 200.0000 & 0.8000 &  & 3.2812 & 5.6841 & 5.1922 &  & 3.0695 & 5.2160 & 4.6565 &  & 5.0026 & 8.3802 & 7.1121 \\
100.0000 & 100.0000 & 0.8000 &  & 3.4702 & 5.1371 & 4.6244 &  & 3.0940 & 4.4182 & 4.3760 &  & 8.7044 & 12.3456 & 10.3841 \\
100.0000 & 200.0000 & 0.8000 &  & 3.7985 & 5.4521 & 4.9804 &  & 3.3254 & 4.7271 & 4.5268 &  & 6.5692 & 9.2071 & 7.8790 \\
100.0000 & 400.0000 & 0.8000 &  & 3.9236 & 5.5008 & 5.0867 &  & 3.9261 & 5.4829 & 5.0494 &  & 4.7237 & 6.5561 & 5.8989 \\
25.0000 & 25.0000 & 1.0000 &  & 0.9962 & 3.3244 & 3.9258 &  & 1.8804 & 4.2743 & 5.4298 &  & 2.3944 & 6.5077 & 7.1917 \\
25.0000 & 50.0000 & 1.0000 &  & 2.0143 & 5.2155 & 5.1491 &  & 1.6980 & 4.2196 & 4.5150 &  & 3.2266 & 8.1170 & 8.2168 \\
25.0000 & 100.0000 & 1.0000 &  & 2.5867 & 5.9941 & 5.8380 &  & 2.1366 & 5.0396 & 5.0741 &  & 3.8652 & 8.9230 & 8.8367 \\
50.0000 & 50.0000 & 1.0000 &  & 2.4784 & 4.6938 & 4.5690 &  & 2.3690 & 4.0249 & 4.9541 &  & 5.1579 & 9.1845 & 10.0636 \\
50.0000 & 100.0000 & 1.0000 &  & 3.1383 & 5.4317 & 5.2045 &  & 2.5755 & 4.4297 & 4.7568 &  & 5.8210 & 9.8498 & 9.9993 \\
50.0000 & 200.0000 & 1.0000 &  & 3.4412 & 5.7348 & 5.5827 &  & 2.9529 & 4.9738 & 5.1350 &  & 4.9495 & 8.1587 & 8.1262 \\
100.0000 & 100.0000 & 1.0000 &  & 3.5983 & 5.2415 & 4.9421 &  & 2.9685 & 4.2383 & 4.9281 &  & 8.5661 & 12.0875 & 12.5631 \\
100.0000 & 200.0000 & 1.0000 &  & 3.9398 & 5.5272 & 5.3232 &  & 3.2433 & 4.5877 & 4.9112 &  & 6.4836 & 9.0147 & 9.1229 \\
100.0000 & 400.0000 & 1.0000 &  & 4.0304 & 5.5980 & 5.4843 &  & 3.8030 & 5.2766 & 5.4550 &  & 4.5878 & 6.3552 & 6.4043 \\
\midrule\multicolumn{3}{c}{Mean abs. dev. from 5\%}  &  &2.3077 &0.58924 &0.57321 &  &2.2752 &0.51656 &0.6014 &  &1.3964 &4.0578 &2.7073\\
\bottomrule
\end{tabular}}
\caption{Sizes (in percent) of nominal 5\% level tests  with no heterogeneity in the alternatives. Based on \num{1000000} replications. Andrews Bandwidth.}
\label{table:sizesAn1_0.4_0.4_BN_}
\end{table}

\begin{figure}[htb!]\footnotesize\begin{subfigure}[b]{0.3466\textwidth}\input{graphs/stamps/Npowers25_25_1_1_iid_0.4_iid_0.4_BN_Dif.tex}\end{subfigure}\begin{subfigure}[b]{0.3267\textwidth}\input{graphs/stamps/Npowers50_25_1_1_iid_0.4_iid_0.4_BN_Dif.tex}\end{subfigure}\begin{subfigure}[b]{0.3267\textwidth}\input{graphs/stamps/Npowers100_25_1_1_iid_0.4_iid_0.4_BN_Dif.tex}\end{subfigure}\\ \begin{subfigure}[b]{0.3466\textwidth}\input{graphs/stamps/Npowers50_50_1_1_iid_0.4_iid_0.4_BN_Dif.tex}\end{subfigure}\begin{subfigure}[b]{0.3267\textwidth}\input{graphs/stamps/Npowers100_50_1_1_iid_0.4_iid_0.4_BN_Dif.tex}\end{subfigure}\begin{subfigure}[b]{0.3267\textwidth}\input{graphs/stamps/Npowers200_50_1_1_iid_0.4_iid_0.4_BN_Dif.tex}\end{subfigure}\\ \begin{subfigure}[b]{0.3466\textwidth}\input{graphs/stamps/Npowers100_100_1_1_iid_0.4_iid_0.4_BN_Dif.tex}\end{subfigure}\begin{subfigure}[b]{0.3267\textwidth}\input{graphs/stamps/Npowers200_100_1_1_iid_0.4_iid_0.4_BN_Dif.tex}\end{subfigure}\begin{subfigure}[b]{0.3267\textwidth}\input{graphs/stamps/Npowers400_100_1_1_iid_0.4_iid_0.4_BN_Dif.tex}\end{subfigure}\\ \begin{subfigure}[b]{\textwidth}\centering\input{graphs/stamps/lNpowers400_100_1_1_iid_0.4_iid_0.4_BN_Dif.tex}\end{subfigure}\normalsize\caption{\label{DifNpowersiidBN111}(Size-corrected) power gains from using $\tUMPEMP$ over $P_b$ for varying values of $\sqrt{\omega^4/\phi^4}$ and sample sizes in the PANIC framework with i.i.d.\ factor innovations and i.i.d.\ idiosyncratic parts. Based on \num{100000} replications.}\end{figure}
	
\subsection{Powers}
We start this subsection by investigating the finite-sample differences between the MP and the PANIC setups. Recall that we have shown that the asymptotic, local power functions are the same and that (under some regularity conditions) all tests have the same asymptotic power in the MP framework as they do in the PANIC framework. \cref{DifMPNpowersMPDifiidBN0.811} compares the powers of $\tUMPEMP$ and $P_b$ across the two frameworks.
 Indeed, also in small samples the powers are very similar. Moreover, both a larger $n$ and a larger $T$ contribute to reduce the difference. When the factor is stationary under the hypothesis, the difference is considerably smaller still. Noting the small scale on the $y$ axis in these plots, in the remainder we will only present results for the PANIC framework, as the lines would otherwise be mostly indistinguishable.

We now turn to comparing the performance of the UMP tests to existing ones. As discussed in \cref{sec:UMPtest}, we need to estimate the individual long-run variance of each idiosyncratic part in order to attain the power envelope. Of course, this becomes easier with a larger time series dimension and is more beneficial when the long-run variances differ substantially between series.

\cref{NpowersiidBN0.811} presents the baseline power results for a medium amount of heterogeneity ($\sqrt{\omega^4/\phi^4}=0.8$). It is evident that even for relatively small samples using the optimal test pays off: except for $n=T=25$, the power of $\tUMPEMP$ is uniformly higher than that of $P_b$. 

%

Next, \cref{DifNpowersiidBN111} presents the power difference between the optimal test and $P_b$ for varying degrees of heterogeneity. As expected, the higher the amount of heterogeneity, the more beneficial it is to use the optimal test, also in finite samples. In the case of perfect homogeneity, the losses from estimating individual long-run variances are minor, except for the $n=T=25$ case.


In the supplemental appendix we investigate the effects of serial correlation and multiple factors. Qualitatively, the power results are not affected by these variations in the DGP. 
We also consider the robustness of our results to deviations of our assumptions:  we consider the power against heterogeneous alternatives and investigate the effects of non-Gaussian innovations.

\section{Conclusion and Discussion}\label{sec:concl}
This paper shows that the MP and PANIC frameworks are equivalent, for unit root testing, from a local and asymptotic point of view. Using the underlying LAN-result, the local asymptotic power envelope for the MP and PANIC frameworks readily follows. We show that the tests proposed in \cite{MoonPerron2004} and \cite{BaiNg2010} only attain this bound in case the long-run variances of the idiosyncratic component are sufficiently homogeneous. We develop an asymptotically uniformly most powerful test; a Monte Carlo study demonstrates that this test also improves on existing tests for finite-samples.

To obtain the local and asymptotic equivalence of the MP and PANIC frameworks, we need to impose some restrictions. First, we assume that the driving innovations are Gaussian. Second, we impose the deviations to the unit root, under the alternative hypothesis, to be the same for all panel units. And third, we do not allow for (incidental) trends. The Gaussianity facilitates a relatively easy proof of the LAN-result and it seems to be rather difficult to generalize this assumption; even for first-generation frameworks no results are available yet.  For the proposed asymptotically uniformly most powerful test, we stress that Gaussianity is not required for its validity.  In view of \cite{ BDvdA} we do not expect that imposing constant deviations to the unit root, under the alternative hypothesis, affects our main results. The Monte Carlo results seem to confirm this conjecture for finite-samples. To allow for incidental trends the proper strategy seems to be to first determine the maximal invariant (i.e. determine which part of the observations is invariant with respect to the incidental trends), and to analyze if the resulting maximal invariant satisfies a LAN-result (yielding the power envelope).  On basis of \cite{MoonPerronPhillips2007} we expect that the reduction of the data to the maximal invariant will result in a different localizing rate compared to the situation in which there are no incidental trends. This indicates that the generalization to incidental trends really requires a separate analysis.
\section*{References}
\bibliographystyle{ecta.bst}
\bibliography{papers}

\clearpage
\begin{appendices}
\crefalias{section}{appsec}
\section{Proof of Main Results}
\subsection{Preliminaries}\label{app:proof_prelims}
\noindent
This section present some preliminary results that are heavily exploited in the proofs of our main results.

First, we  recall some elementary  results from linear algebra (throughout we only consider real matrices); see, e.g.,  \cite{Lutkepohl1996}  and \cite{MagnusNeudecker1999}.
Let $\tr [C]$ denote the trace of a square, real matrix $C$ and let $\mineig{C}$ (and $\maxeig{C}$) denote the minimal (maximal) eigenvalue of a symmetric, real matrix $C$. For any real matrix $C$, let
$\Frob{C}= \sqrt{ \tr \left[ C\trans C\right]}=\Frob{C'}$ denote its Frobenius norm, while $\rnorm{C}= \sqrt{ \maxeig{C\trans C}}=\rnorm{C'}$ denotes its spectral norm. Recall $\rnorm{C}\le \Frob{C}$. 
%

The inequality $\Frob{CD}\leq \rnorm{C}\Frob{D}$ is immediate from Raleigh's quotient. It follows that the Frobenius is submultiplicative, $\Frob{C D}\le\Frob{C}\Frob{D}$. Moreover, the identity $\Frob{C\kron D}=\Frob{C}\Frob{D}$ easily follows from the alternative interpretation of the Frobenius norm being the square-root of the sum of all squared individual matrix entries.  
 Finally, we note that for square matrices $\langle C,D\rangle_F=\tr[C'D]$ defines an inner product, so  we have the Cauchy-Schwarz inequality $|\tr[C'D]|\le \Frob{C}\Frob{D}$.

%
Next, we present a general lemma on approximating variances with long-run variances. The results we present in this appendix are the main keys to many proofs in \cref{sec:limexp}.
Moreover, they may be of general interest.
\begin{lemma} \label{lemma:generalLRV}%
Consider an indexed collection of stationary time series $\{X^{(h)}_t\}$, $h\in\mathcal{H}$. 
Denote the $T\times T$ covariance matrix of $(X^{(h)}_1,\dots,X^{(h)}_T)$ by $\Sigma_h$, the $m$-th autocovariance of $\{X^{(h)}_t\}$ by $\autocorr_h(m)$, and its long run variance by $\omega_h^2<\infty$. Also write $\omega^2_{h,T}= \ones'\Sigma_h \ones/T$.  
If $\sup_{h\in\mathcal{H}} \sum_{\m=-\infty}^{\infty} (|\m|+1)|\autocorr_h(\m)|<\infty$, then 
\begin{enumerate}
	\item $\sup_{h\in\mathcal{H}} |\longrunvarT{h}-\longrunvar{h}|=O(T^{-1})$,\label{part:difLRV}
	\item $\sup_{h\in\mathcal{H}} \Frob{A'(\Sigma_h-\longrunvar{h}I_T)}+\sup_{h\in\mathcal{H}} \Frob{A(\Sigma_h-\longrunvar{h}I_T)}=O(\sqrt{T})$,\label{part:generalLRV}
	\item $\sup_{h\in\mathcal{H}} \Frob{A'(\Sigma_h-\longrunvarT{h}I_T)}+\sup_{h\in\mathcal{H}} \Frob{A(\Sigma_h-\longrunvarT{h}I_T)}=O(\sqrt{T})$,  \label{part:generalLRVT}
	\item $\sup_{h\in\mathcal{H}} \Frob{A'\Sigma_h}+\sup_{h\in\mathcal{H}} \Frob{A\Sigma_h}=O(T).$\label{part:generalOSLRV}
\end{enumerate}
\end{lemma}

\begin{proof}
\label{proof:generalLRV}
\Cref{part:difLRV} follows from $\longrunvarT {h}=\frac{1}{T}\sum_{m<T}(T-|m|)\gamma_h(m)$ and $\longrunvar {h}=\sum_{m=-\infty}^\infty\gamma_h(m)$, so
\begin{align}
   |\longrunvarT{h}-\longrunvar{h}|
  &=\left|\frac{1}{T}\sum_{m=-\infty}^{\infty}(\min(|m|,T)\gamma_h(m)\right|,
\end{align}
which is indeed $O(T^{-1})$ uniformly in $h$.

For \cref{part:generalLRV}, tedious but elementary calculations yield 
\begin{align}
\MoveEqLeft\Frob{A(\Sigma_h-\omega_h^2I_T)}^2
	 =
	\Frob{A'(\Sigma_h-\omega_h^2I_T)}^2\\
	&=\sum_{s=1}^{T} \sum_{t=1}^{T}\left( \sum_{\m=s-t+1}^{T-t} \autocorr_h(\m)-\omega_h^2 1_{s<t} \right)^2 \\
	&= \sum_{s=1}^{T-1}\Bigg(\sum_{t=1}^{s}\left(\sum_{\m=s+1}^{T}\autocorr_h(\m-t)\right)^2 \\
	&\quad+\sum_{t=s+1}^{T}\left(\sum_{\m=-\infty}^{s}\autocorr_h(\m-t)+\sum_{\m=T+1}^{\infty}\autocorr_h(\m-t)\right)^2\Bigg)\\
	&=\sum_{s=1}^{T-1}\sum_{t=1}^{T-s}\left(\left(\sum_{\m=s}^{T-t}\autocorr_h(\m)\right)^2+\left(\sum_{\m=s}^{\infty}\autocorr_h(\m)+\sum_{\m=t}^{\infty}\autocorr_h(\m)\right)^2\right)\\
	&\le 5T\sum_{s=1}^{T}\left(\sum_{\m=s}^{\infty}|\autocorr_h(\m)|\right)^2\\
	&\le5T\left(  \sum_{\m=-\infty}^{\infty} |\autocorr_h(\m)|\right) \sum_{\m=1}^{\infty}\min(\m,T)|\autocorr_h(\m)|.
\end{align}
Taking suprema, \cref{part:generalLRV} follows immediately from this bound. \Cref{part:generalLRVT} follows by combining the first two parts and $\Frob{A}=\sqrt{\frac{T(T-1)}{2}}=O(T)$. The order on $\Frob{A}$ also yields
\begin{align}
  \sup_{h\in\mathcal{H}} \Frob{A'\Sigma_h}
  \le& \sup_{h\in\mathcal{H}} \Frob{A'(\Sigma_h-\omega_h^2I_T)}+
  \sup_{h\in\mathcal{H}} \omega_h^2\Frob{A'}\\
  =&O(\sqrt{T})+O(1)O(T).
\end{align}
Again, the second part of \cref{part:generalOSLRV} is analogous.
 \end{proof}

Recall the covariance matrices $\vareta$ and $\vareps$ and their rough approximations $\LReta$ and $\proxvareps$ defined in \cref{lemma:PANICCS,eqn:Omegavarepsilonprox}, respectively. The following three lemmas use \cref{lemma:generalLRV} to show that these approximations do work well when considering partial sums.

\begin{lemma}\label{lemma:eigenvalues}
Under \cref{partb}, $\rnorm{\vareta^{-1}}$, $\rnorm{ \LReta^{-1}}$, $\rnorm{\vareps^{-1}},$ and $\rnorm{\proxvareps^{-1}}$ are all $O(1)$ as $n,T\to\infty$.
\end{lemma}
\begin{proof}\label{lem:normsep}
Note that $\vareps-\vareta$ and $\proxvareps-\Psi_\eta$ are positive semidefinite. Hence $\mineig{\vareps}\geq\mineig{\vareta}  \geq \inf_{i,T}  \mineig{\vareta[,i]}>0$ and, using \cref{rem:bddLRV} (\cref{footnoteLRV}) and \cref{part:difLRV} of \cref{lemma:generalLRV}, 
\begin{align}
  \mineig{\proxvareps }\geq &\mineig{ \Psi_\eta}=\mineig{ \longrunvarmatrixeta\kron \id[T]}
  =\min_{i=1,\dots,n}\longrunvaretaT{i}\\
  \ge& \inf_{i\in\mathds{N}}\longrunvareta{i}
  -\sup_{i\in\mathds{N}}|\longrunvaretaT{i}-\longrunvareta{i}|
  \to \inf_{i\in\mathds{N}}\longrunvareta{i}>0.
\end{align}
This shows the boundedness of all four norms.
\end{proof}

\begin{lemma}
\label{lemma:etaLRV}
	Under \cref{partb} we have, as $n,T\to\infty$,
	\begin{align}
	 \Frob{  \calA\trans\left(  \vareta  -\LRetaT  \right) }
	+ \Frob{  \calA\left(   \vareta-\LRetaT   \right) }
	=O(\sqrt{nT})=o(\sqrt{n}T).
	\end{align}
\end{lemma}
\begin{proof}
	Using block diagonality and \cref{lemma:generalLRV}, we obtain the bound
	\begin{align}
	\Frob{  \calA\trans\left( \vareta  -\LRetaT    \right) }^2&=\sum_{i=1}^n \Frob{ \A\trans (\vareta[,i]- \longrunvaretaT{i} \id[T])    }^2\\
	&\le n \sup_{i\in \mathds{N}}\Frob{ \A\trans (\vareta[,i]-\longrunvaretaT{i} \id[T])}^2
	=O(nT).
	\end{align}
	The other part is analogous; every $\calA'$ and $A'$ are replaced by $\calA$ and $A$, respectively.
\end{proof}
\begin{lemma}
\label{lemma:specificLRV}
	Under \cref{assumption:serial correlation,assumption:rates,assumption:factor_loadings} we have, as $n,T\to\infty$,
	\begin{align}
	 \Frob{  \calA\trans\left(\vareps-   \proxvarepsT  \right) }
	+ \Frob{  \calA\left(  \vareps  -\proxvarepsT  \right) }
	=O(n\sqrt{T})=o(\sqrt{n} T).
	\end{align}
\end{lemma}

\begin{proof}\label{proof:specificLRV}
	From the definitions of $\vareps$ and $\proxvarepsT$ we obtain
	\begin{align}
	\calA\trans\left(\vareps - \proxvarepsT \right)=\sum_{k=1}^K \calA\trans\left( \loading{k}\loading{k}\trans \kron \left( \varF[,k]  -\longrunvarFT{k}\id[T]\right)
	\right)+\calA\trans \left( \vareta - \longrunvarmatrixetaT \kron \id[T] \right),
	\end{align}
	which yields the bound $\Frob{ \calA\trans\left(\vareps - \proxvarepsT \right) }\leq I+II$
	with
	\begin{align}
	I&= \sum_{k=1}^K \Frob{  \left( \loading{k}\loading{k}\trans \kron  \A\trans \left( \varF[,k]  -\longrunvarFT{k}\id[T]\right)
	\right)  } \text{ and }
	II= \Frob{ \calA\trans \left( \vareta - \longrunvarmatrixetaT \kron \id[T] \right) }.
	\end{align}
	
Part $II$ is already treated in \cref{lemma:etaLRV}. For part $I$, again using \cref{lemma:generalLRV}, we  get a slightly weaker bound since for the factor part there is no block diagonality:
	\begin{align}
	I=& \sum_{k=1}^K \Frob{  \loading{k}\loading{k}\trans }  \Frob{\A\trans \left( \varF[,k]  -\longrunvarFT{k}\id[T]\right)}\\
	  \leq& \sum_{k=1}^K   \loading{k}\trans\loading{k}
	 \Frob{\A\trans \left( \varF[,k]  -\longrunvarFT{k}\id[T]\right)}
		=O(n\sqrt{T})
		=o(\sqrt{n}T).
	\end{align}
	The proof for $\Frob{  \calA\left(  \vareps -\proxvarepsT   \right) }$ is analogous.
	\end{proof}

We now present a general weak convergence result for partial sums using joint asymptotics. \Cref{prop:limitexperiment} is a special case of \cref{lemma:normality} with $a_{i,n,T}=1$. We provide  \cref{lemma:normality} in general terms here as it might be of independent interest and we also use it in the proof of \cref{prop:local_powers_BN_and_MP} to demonstrate the joint convergence of $P_a$ and the local likelihood ratio.
\begin{lemma}\label{lemma:normality}
Let $a_{i,n,T}$ be a bounded sequence of non-random numbers and $ \frac{1}{n}\sum_{i=1}^n a_{i,n,T}^2 \to \alpha$. Then, under  $\lawMP$ or $\lawPANIC$, as $\jointlimits$,
\begin{align}\frac{1}{\sqrt n}\sum_{i=1}^n \frac{a_{i,n,T}}{\longrunvaretaT{i}}
\left( \frac{1}{T} \sum_{t=1}^T\sum_{s=1}^{t-1} \eta_{is} \eta_{it}
-
\oslongrunvareta{i}
\right) \vto N(0, \alpha/2 ).\end{align}
\end{lemma}

\begin{proof}
First consider the case of $a_{i,n,T}$ being identically equal to one and observe that this implies convergence of $\etaCS$.
Recall $A+A'= \ones\ones'-I_T$ and $2\oslongrunvaretaT{i}=\longrunvaretaT{i}-\ACFgamma{i}(0)$, hence, with $\longrunvaretaT{i}=\frac{1}{T}\ones'\Sigma_{\eta,i}\ones$,
\begin{align}
  \etaCS=&\frac{1}{\sqrt{n}T}\sum_{i=1}^{n}\frac{1}{\longrunvaretaT{i}}\eta_i'\frac{A+A'}{2}\eta_i-\frac{1}{\sqrt{n}}\sum_{i=1}^{n}\frac{\oslongrunvaretaT{i}}{\longrunvaretaT{i}}\\
  =&\frac{1}{2\sqrt{n}}\sum_{i=1}^{n}\left(\left(\frac{\ones'\eta_i}{\sqrt{T}\longrunvaretaTroot{i}}\right)^2-1\right)
  -\frac{1}{2\sqrt{n}}\sum_{i=1}^{n}\frac{1}{\longrunvaretaT{i}}\left(\frac{1}{T}\eta_i'\eta_i-\ACFgamma{i}(0)\right).
\end{align}
Observe that $X_{i,T}:=\frac{\ones'\eta_i}{\sqrt{T\longrunvaretaT{i}}}\sim N(0,1)$ and are independent across $i\in\mathds{N}$. Thus, for each $T$, $\frac{1}{\sqrt{2n}}\sum_{i=1}^{n}(X_{i,T}^2-1)$  has the same distribution as $\frac{1}{\sqrt{2n}}\sum_{i=1}^{n}(X_i^2-1)$, where $X_i^2\overset{iid}{\sim}\chi^2(1)$. Therefore, as the latter converges to a standard normal distribution as $n\to\infty$ (CLT), so does the former under joint limits. Thus, the first, leading term converges in distribution to  $N(0,1/2)$.

Asymptotic negligibility of the second, mean-zero term follows from
\begin{align}
  \sup_i \var(\frac{1}{T}\eta_i'\eta_i)=&\frac{2}{T^2}\sup_i \tr[\vareta[,i]^2]=\frac{2}{T^2}\sup_i\Frob{\vareta[,i]}^2\\
  =&\frac{2}{T}\sup_i \left|\sum_{m=-(T-1)}^{T-1} (1-\frac{|m|}{T})\ACFgamma{i}^2(m)\right|
  =O(T^{-1}).
\end{align}

For general $a_{i,n,T}$ we can apply a double array CLT, see 1.9.3 in \cite{serfling1980approximation}, to the first (slightly adapted) term in the expansion. The Lindeberg condition is readily verified since we have a weighted sum of i.i.d.\ centered $\chi^2$ variables. Asymptotic negligibility of the second remainder term follows from the boundedness condition on the $a_{i,n,T}$. 
\end{proof}
\begin{remark}\label{remarkGaussianity}
We can obtain the same conclusion without requiring Gaussian innovations: as long as the Lindeberg condition holds, for example thanks to higher moment conditions, the same Theorem 1.9.3 of \cite{serfling1980approximation} applies.  
\end{remark}

We conclude this subsection by taking care of important terms that appear repeatedly in the remainder.
\begin{lemma}
	\label{lemma:extrasFeasible}
Suppose that Assumptions~\ref{assumption:serial correlation}-\ref{assumption:rates} hold. Then, under $\lawMP$ or $\lawPANIC$ and as $n,T \to \infty$, we have
	\begin{enumerate}	\item $  \Frob{ \left(\frac{1}{n} \Loading\trans \longrunvarmatrixeta^{-1}\Loading\right)^{-1}}=O(1)$,\label{lem:bddFacInv}
	\item $\Frob{\sum_{t=2}^{T}\eta_{\cdot, t}}=O_P(\sqrt{nT})$,\label{sumeta}
	\item $\Frob{\sum_{t=2}^T f_{\cdot,t}}=O_P(\sqrt{T})$,\label{sumf}
		 \item \label{sumLoadingeta}$\Frob{\iota'\tilde\eta \longrunvarmatrixeta^{-1} \Loading} = O_P(\sqrt {n T})$, and
	\item \label{eta1Loading}$\Frob{\tilde\eta \longrunvarmatrixeta^{-1} \Loading}=O_P(\sqrt{nT})$.
	\end{enumerate}
		\end{lemma}
\begin{proof}
For \cref{lem:bddFacInv},  recall that $K$ is fixed, so that the norm we consider is irrelevant.
   As
   \begin{align}
  \Loading\trans \longrunvarmatrixeta^{-1}\Loading
  =\sum_{i=1}^{n}\frac{1}{\longrunvaretaT{i}}\lambda_i\lambda_i'
  \ge \frac{1}{\sup_{i\in\mathds{N}}\longrunvaretaT{i}}\sum_{i=1}^{n}\lambda_i\lambda_i',
\end{align}
the smallest eigenvalue of $\Loading\trans \longrunvarmatrixeta^{-1}\Loading$ is larger than that of $\Loading\trans \Loading$. Thus,
\begin{align}
  \rnorm{ \left(\frac{1}{n} \Loading\trans \longrunvarmatrixeta^{-1}\Loading\right)^{-1}}
  \le \sup_{i\in\mathds{N}}\longrunvaretaT{i}\rnorm{ \left(\frac{1}{n} \Loading\trans \Loading\right)^{-1}}
  \to \sup_{i\in\mathds{N}}\omega^2_i\rnorm{\limitfactorloadings^{-1}}<\infty,
\end{align}
thanks to \cref{assumption:factor_loadings,assumption:serial correlation}.

\Cref{sumeta} follows from
\begin{align}
  \expec\Frob{\sum_{t=1}^{T}\eta_{\cdot, t}}^2
  =\expec\Frob{\tilde\eta'\iota }^2
  =\iota '\expec \tilde\eta\tilde\eta'\iota 
  =\iota '\sum_{i=1}^{n}\expec \eta_i\eta_i'\iota 
  =T\sum_{i=1}^{n}\longrunvaretaT{i}=O(nT).
\end{align}
Note that the expectation of $\Frob{\sum_{t=2}^{T}\eta_{\cdot, t}}^2$ is given by $(T-1)\sum_{i=1}^{n}\omega^2_{\eta,i,T-1}$ and is thus of the same order.

\cref{sumf} can be obtained along a similar line of proof.

 For \cref{sumLoadingeta}, note $\expec \tilde\eta'\iota \iota \tilde\eta=T\longrunvarmatrixeta$, so that
\begin{align}
\expec\Frob{\iota'\tilde\eta \longrunvarmatrixeta^{-1} \Loading}^2 
&=\tr{\expec [\tilde\eta'\iota \iota \tilde\eta]\longrunvarmatrixeta^{-1}\Loading\Loading'\longrunvarmatrixeta^{-1}}\\
&=T\tr{\Loading\Loading'\longrunvarmatrixeta^{-1}}
\le T\Frob{\Loading}^2\rnorm{\longrunvarmatrixeta^{-1}}=O(nT).
\end{align}

\Cref{eta1Loading} follows similarly from $\expec \crosseta{t}\crosseta{t}'=\diag(\ACFgamma{1}(0),\dots,\ACFgamma{n}(0))=: D$, so 
\begin{align}
   E \Frob{\tilde\eta \longrunvarmatrixeta^{-1} \Loading}^2
  &=\tr(\Loading'\longrunvarmatrixeta^{-1}\sum_{t=1}^{T}\expec[\crosseta{t}\crosseta{t}']\longrunvarmatrixeta^{-1}\Loading)
  \le T\Frob{\Loading}^2\rnorm{\longrunvarmatrixeta^{-1}}^2\rnorm{D},
  \end{align}
  which is indeed $O(nT)$ thanks to \cref{assumption:factor_loadings,assumption:serial correlation}.
  \end{proof}

\subsection{Proofs of \cref{sec:limexp}}
\begin{proof}[Proof of \cref{lemma:PANICCS}]
In the following all probabilities and expectations are evaluated under $\lawPANIC$. To obtain the desired result, we consider the difference between the two central sequences $\etaCS-\CSorigPANIC$ and the difference between the two Fisher informations $\FIorigPANIC-\frac{1}{2}$.
We show that expectations and variances of both differences converge to zero, implying $L_2$ convergence. 

\emph{Part A:}
Under the null, $\Delta E=\eta$ and hence
\begin{align}
  \etaCS-\CSorigPANIC= \frac{1}{\sqrt n T} \eta\trans\calA\trans (\LRetaT^{-1}-\vareta^{-1}) \eta -\frac{1}{\sqrt n}\sum_{i=1}^n \frac{\oslongrunvaretaT{i}}{\longrunvaretaT{i}}.
\end{align}
 We first show that the difference has mean zero. We have, using $\tr(\calA)=0$ and block diagonality of $\vareta$,
\begin{align}
  \expec [\etaCS-\CSorigPANIC]&=\frac{1}{\sqrt n T}\tr(\calA\trans (\LRetaT^{-1}-\vareta^{-1})\vareta) -\frac{1}{\sqrt n}\sum_{i=1}^n \frac{\oslongrunvaretaT{i}}{\longrunvaretaT{i}}\\
  &=\frac{1}{\sqrt n T}\tr(\calA\trans \LRetaT^{-1}\vareta) -\frac{1}{\sqrt n}\sum_{i=1}^n \frac{\oslongrunvaretaT{i}}{\longrunvaretaT{i}}\\
  &=\frac{1}{\sqrt n T}\tr(({\longrunvarmatrixetaT}^{-1}\kron A')\vareta) -\frac{1}{\sqrt n}\sum_{i=1}^n \frac{\oslongrunvaretaT{i}}{\longrunvaretaT{i}}\\
  &=\frac{1}{\sqrt{n}}
\frac{1}{T}
\sum_{i=1}^n \frac{1}{\longrunvaretaT{i}} \tr\left[\A\trans \vareta[,i]\right]-\frac{1}{\sqrt n}\sum_{i=1}^n \frac{\oslongrunvaretaT{i}}{\longrunvaretaT{i}}=0,
\end{align}
as $\tr\left[\A\trans \vareta[,i]\right]=T\oslongrunvaretaT{i}$.

To show that the variance of $\CSorigPANIC-\etaCS$ goes to zero, observe 
\begin{align}\label{varformula}
nT^2\var (\CSorigPANIC-\etaCS)=&\var(\eta\trans C_\eta  \eta)=\tr[ C_\eta  \vareta C_\eta  \vareta ] +  \tr[ C_\eta  \vareta C_\eta \trans \vareta]\\
\le&\Frob{C_\eta  \vareta}^2+\Frob{C_\eta\vareta}\Frob{\vareta C_\eta },
\end{align}
with $C_\eta =
\calA\trans (\LRetaT^{-1}-\vareta^{-1}) 
$. Hence, it suffices to show $\Frob{C_\eta \vareta }=o(\sqrt{n}T)$ and $\Frob{\vareta C_\eta}=o(\sqrt{n}T)$. Since $\LReta^{-1}$ and $\calA\trans$ commute, we obtain 
\begin{align}
  \Frob{C_\eta\vareta}
  =&\Frob{\calA'\LRetaT^{-1}(\vareta-\LRetaT)}
  \le& \rnorm{\LRetaT^{-1}}\Frob{\calA\trans(\vareta-\LRetaT)},
\end{align}
which is indeed $o(\sqrt{n}T)$ by \cref{lemma:eigenvalues,lemma:etaLRV}.
 For $\Frob{\vareta C_\eta }$, we first  have to approximate $\calA\vareta$ with $\calA\LRetaT$ before we can use the commutativity as above:
 \begin{align}
  \Frob{\vareta C_\eta }
  &\le \Frob{\LRetaT C_\eta }
  +\Frob{C_\eta '(\vareta-\LRetaT)}\\
  &=\Frob{\calA' (\vareta-\LRetaT)\vareta^{-1}}
 +\Frob{ \left( \LRetaT^{-1}-\vareta^{-1}     \right)\calA(\vareta-\LRetaT)}\\
 &\le \rnorm{\vareta^{-1}}\Frob{\calA\trans (\LRetaT-\vareta)}\\
  &+\left(    \rnorm{\LRetaT^{-1}}+\rnorm{\vareta^{-1}}   \right)\Frob{\calA(\vareta-\LRetaT)}
  =o(\sqrt{n}T).
\end{align}

\emph{Part B:}
First, we show that the expectation of $\FIorigPANIC$ converges to $\frac{1}{2}$. We have
\begin{align}
  nT^2\expec \FIorigPANIC
  &=\tr\left[  \calA\trans \vareta^{-1} \calA \vareta  \right]
  =\tr\left[  \calA\trans \LRetaT^{-1} \calA \vareta  \right]-
  \tr\left[  \calA\trans C_\eta \trans \vareta \right]\\
  &=\tr [ \calA\trans\calA]
  +\tr  [\calA\trans \LRetaT^{-1}\calA(\vareta-\LRetaT)]
  -\tr [\vareta C_\eta \calA].
\end{align}
This implies that the leading term is $\frac{1}{2}nT^2$, since the final two terms are $o(nT^2)$: use the arguments already presented in Part A together with the relation between the trace and the Frobenius norm and  
\begin{align}
\frac{1}{nT^2}\Frob{\calA}^2=
  \frac{1}{nT^2}\tr\left[  \calA\trans\calA \right]
  =\frac{1}{T^2}\tr\left[  A\trans A\right]
  =\frac{T(T-1)}{2T^2}\to \frac{1}{2}.
\end{align}

Next, we show that the variance converges to zero.  By the arguments in \cref{varformula}, with $D_\eta=\calA\trans\vareta^{-1}\calA$,
\begin{align}
  n^2T^4 \var(\FIorigPANIC)\le 2\Frob{\vareta D_\eta}^2.
\end{align}
%
The required order is now easily verified, since 
\begin{align}
  \Frob{\vareta D_\eta} 
  \le & \Frob{\calA\trans \LRetaT^{-1}\calA \vareta}+\Frob{\vareta C_\eta  \calA}\\
  \le & \Frob{\calA\trans\calA}+\Frob{\calA\trans\LRetaT^{-1}\calA(\vareta-\LRetaT)}+\Frob{\vareta C_\eta  \calA}
\end{align}
and $\Frob{\calA'\calA}=\sqrt{n}\Frob{A'A}\le \sqrt{n}\Frob{A}^2=\sqrt{n}T(T-1)/2.$
\end{proof}
%
\begin{proof}[Proof of \Cref{lemma:proxyCS}]\label{sec:proofproxyCS}
In the following all probabilities and expectations are evaluated under $\lawMP$.
The proof of this lemma follows the idea of the proof of \cref{lemma:PANICCS} by considering means and variances. The proof that $\FIorigMP$ converges to $\frac{1}{2}$ in $L_2$ is almost identical to its counterpart in the proof of \cref{lemma:PANICCS}: just replace $\eta$ by $\varepsilon$, $\Sigma_\eta$ by $\Sigma_\varepsilon$, $C_\eta$ by $C_\varepsilon$ etc. The same replacements yield that the variance of $\CSMPLRV-\CSorigMP$ converges to zero, by applying them to the arguments starting at \cref{varformula}. We are left to show that the expectation of  $\CSMPLRV-\CSorigMP$ converges to zero. This remaining expectation is more complicated since the variance matrices $\Sigma_\varepsilon$ and $\proxvareps$ have additional terms due to the presence of unobservable factors. 

Recall, under $\lawMP$, $\Delta Y=\varepsilon$ and note
\begin{align}
\CSMPLRV-\CSorigMP&=\frac{1}{\sqrt{n}}\left( \frac{1}{T} \varepsilon \trans \calA\trans \left(  \proxvareps^{-1} -\vareps^{-1}  \right) \varepsilon
-\sum_{i=1}^n \frac{\oslongrunvaretaT{i}}{
\longrunvaretaT{i} } 
\right).
\end{align}
%
%
Thus, we have
\begin{align}
	\expec [\CSMPLRV-\CSorigMP]
&=\frac{1}{\sqrt{n}T}\tr[
\calA\trans \proxvarepsT^{-1}\vareps]
-\frac{1}{\sqrt{n}}\sum_{i=1}^n \frac{\oslongrunvaretaT{i}}{
\longrunvaretaT{i} }\\
&=\frac{1}{\sqrt{n}T}\tr[
\calA\trans \LRetaT^{-1}\vareta]
-\frac{1}{\sqrt{n}}\sum_{i=1}^n \frac{\oslongrunvaretaT{i}}{
\longrunvaretaT{i} }\\
&+\frac{1}{\sqrt{n}T} \sum_{k=1}^K  \tr\left[   \proxvarepsgen^{-1} \loading{k}\loading{k}\trans \kron \A\trans \varF[,k]  \right]\\
&+\frac{1}{\sqrt{n}T}
\tr\left[    \left( (\proxvarepsgen^{-1}-\Omega_\eta^{-1}) \kron \A\trans \right) \vareta \right]
=: I +II+III.
\end{align}
In the proof of \cref{lemma:PANICCS} we have established that the first term equals zero. Therefore, the current proof is complete once we show the final two terms converge to zero. 

Convergence to zero of $II$ follows from $\frac{1}{T}\tr(A'\varF[,k])=\oslongrunvarfT{k}=O(1)$
 in combination with
\begin{align}
  \sum_{k=1}^K\tr\left[   \proxvarepsgen^{-1} \loading{k}\loading{k}\trans\right]
  =&\tr[\Lambda'\proxvarepsgen^{-1}\Lambda]
  =\tr\left[
\longrunvarmatrixF^{-1}
-
\longrunvarmatrixF^{-1}
 \left( \longrunvarmatrixF^{-1} + \Loading\trans
\longrunvarmatrixeta^{-1}\Loading\right)^{-1} \longrunvarmatrixF^{-1}
\right]\\
\leq& \tr\left[
\longrunvarmatrixF^{-1}
\right]=\sum_{k=1}^K \frac{1}{\longrunvarFT{k}}
\to\sum_{k=1}^K \frac{1}{\longrunvarF{k}}<\infty.\label{eq:loadingsBound}
\end{align}
Convergence to zero of $III$ follows from
\begin{align}
  |III|
\le&\frac{1}{\sqrt{n}T}\sum_{i=1}^{n}
\left(
\longrunvarmatrixeta^{-1}\Loading\left( \longrunvarmatrixF^{-1} +\Loading\trans \longrunvarmatrixeta^{-1}\Loading \right)^{-1}\Loading\trans \longrunvarmatrixeta^{-1}
\right)_{i,i}|\tr\left[\A\trans
 \vareta[,i]
\right]|\\
\le& \frac{1}{\sqrt{n}T}\tr(\longrunvarmatrixeta^{-1}\Loading\left( \longrunvarmatrixF^{-1} +\Loading\trans \longrunvarmatrixeta^{-1}\Loading \right)^{-1}\Loading\trans \longrunvarmatrixeta^{-1})
	\sup_i |\tr\left[\A\trans\vareta[,i]\right]|.\\
	\le&\frac{1}{\sqrt{n}T}\Frob{  \Loading}^2\rnorm{ \longrunvarmatrixeta^{-1} }^2
	\rnorm{ \left( \longrunvarmatrixF^{-1} +\Loading\trans \longrunvarmatrixeta^{-1}\Loading \right)^{-1}  }
	\sup_i |\tr\left[\A\trans\vareta[,i]\right]|.\label{smallcconv}
\end{align}
Observe $\sup_i \tr\left[\A\trans\vareta[,i]\right]=O(T)$ by \cref{part:generalOSLRV} of \cref{lemma:generalLRV}.
From \cref{assumption:factor_loadings} we get $\Frob{  \Loading}=O(\sqrt{n})$
and 
\begin{align}
\MoveEqLeft
n\rnorm{ \left( \longrunvarmatrixF^{-1} +\Loading\trans \longrunvarmatrixeta^{-1}\Loading \right)^{-1}  }
=
\rnorm{ \left( \frac{1}{n} \longrunvarmatrixF^{-1} + \frac{1}{n} \Loading\trans \longrunvarmatrixeta^{-1}\Loading \right)^{-1}  }\\
&=\lambda_{min}^{-1}( \frac{1}{n} \longrunvarmatrixF^{-1} + \frac{1}{n} \Loading\trans \longrunvarmatrixeta^{-1}\Loading)
\le\lambda_{min}^{-1}(\frac{1}{n} \Loading\trans \longrunvarmatrixeta^{-1}\Loading)\\
&\le\lambda_{min}^{-1}(\frac{1}{n} \Loading\trans\Loading)\sup_{i\in\mathds{N}}\longrunvaretaT{i} 
\to \lambda_{min}^{-1}(\Psi_\Lambda) \sup_{i\in\mathds{N}}\omega^2_{\eta,i}<\infty.
\end{align}
A combination of these observations with the penultimate display yields $III=o(1)$.
\end{proof}
\begin{proof}[Proof of \cref{lemma:auxCS}]
We have 
		\begin{align}
|\auxCS - \CSMPLRV | 
&=\frac{1}{\sqrt n T} |\tr(A\tilde\varepsilon({\auxproxvarepsgen}^{-1}-{\proxvarepsgen}^{-1})\tilde\varepsilon')|\\
&\le\frac{1}{\sqrt n T} \Frob{{\auxproxvarepsgen}^{-1}-{\proxvarepsgen}^{-1}}\Frob{\tilde\varepsilon'A\tilde\varepsilon}.
\end{align}
We consider each norm separately. We have 
\begin{align}
  \Frob{{\auxproxvarepsgen}^{-1}-{\proxvarepsgen}^{-1}}
  &\le \rnorm{(
\Loading\trans  \longrunvarmatrixeta^{-1} \Loading +\longrunvarmatrixF )^{-1}-
(
\Loading\trans  \longrunvarmatrixeta^{-1} \Loading)^{-1}}\rnorm{\longrunvarmatrixeta^{-1}}^2\Frob{\Loading}^2\\
&=O(n^{-2})O(1)O(n)=O(n^{-1}),
\end{align}
as $\Frob{\Loading}=O(\sqrt{n})$ by \cref{assumption:factor_loadings}, $\rnorm{\longrunvarmatrixeta^{-1}}=O(1)$ by \cref{assumption:serial correlation}, and
\begin{align}
  \MoveEqLeft n\rnorm{(
\Loading\trans  \longrunvarmatrixeta^{-1} \Loading +\longrunvarmatrixF )^{-1}-
(
\Loading\trans  \longrunvarmatrixeta^{-1} \Loading)^{-1}}\\
=&\rnorm{\left(
\frac{\Loading\trans  \longrunvarmatrixeta^{-1}\Loading}{n} +\frac{\longrunvarmatrixF}{n} \right)^{-1}-
\left(
\frac{\Loading\trans  \longrunvarmatrixeta^{-1}\Loading}{n}\right)^{-1}}\\
=&\rnorm{-\left(
\frac{\Loading\trans  \longrunvarmatrixeta^{-1}\Loading}{n} +\frac{\longrunvarmatrixF}{n} \right)^{-1}\frac{\longrunvarmatrixF}{n}
\left(
\frac{\Loading\trans  \longrunvarmatrixeta^{-1}\Loading}{n} \right)^{-1}}\\
\le & \rnorm{\frac{\longrunvarmatrixF}{n}}\rnorm{\left(
\frac{\Loading\trans  \longrunvarmatrixeta^{-1}\Loading}{n} +\frac{\longrunvarmatrixF}{n} \right)^{-1}}\rnorm{
\left(
\frac{\Loading\trans  \longrunvarmatrixeta^{-1}\Loading}{n} \right)^{-1}},
\end{align}
which is $O(n^{-1})$: the second norm converges to the third, which is O(1) by \cref{lem:bddFacInv} of \cref{lemma:extrasFeasible}.
For $\Frob{\tilde\varepsilon'A\tilde\varepsilon}$, we note that $\Frob{\tilde\varepsilon'A\tilde\varepsilon}=\Frob{\tilde\varepsilon'\frac{A+A'}{2}\tilde\varepsilon}$ and recall that $A+A'=\iota\iota'-I_T$, so that
\begin{align}\label{epsilonY}
 2 \Frob{\tilde\varepsilon'A\tilde\varepsilon}=\Frob{\tilde\varepsilon'(\iota\iota'-I_T)\tilde\varepsilon}
  \le \Frob{\iota'\tilde\varepsilon}^2+\Frob{\tilde\varepsilon}^2=O_P(nT),
\end{align}
as $\Frob{\tilde\varepsilon}\le \Frob{\Loading}\Frob{\tilde f}+\Frob{\tilde\eta}=O(\sqrt{n})O_P(\sqrt{T})+O_P(\sqrt{nT})$ and, using \cref{sumeta,sumf} of \cref{lemma:extrasFeasible}, a similar bound holds for $\Frob{\iota'\tilde\varepsilon}$. 
 Conclude that the central sequence difference is $O_p(n^{-1/2})$. 
\end{proof}

 \begin{proof}[Proof of \Cref{lemma:NormalityCS}]
 As ${\auxproxvarepsgen}^{-1}$ projects out the factors, we have
 \begin{align}
\auxCS - \etaCS  &=  
\frac{1}{\sqrt n T} \tr(A\tilde\varepsilon{\auxproxvarepsgen}^{-1}\tilde\varepsilon')
-\frac{1}{\sqrt n T}\tr(A\tilde\eta{\longrunvarmatrixeta }^{-1}\tilde\eta')\\
&=\frac{1}{\sqrt n T} \tr(A\tilde\eta({\auxproxvarepsgen}^{-1}-{\longrunvarmatrixeta }^{-1})\tilde\eta').
\end{align}
Note that for a symmetric matrix $B$, 
\begin{align}
  \tr(A\tilde\eta B\tilde\eta')
  =
  \tr(\tilde\eta B\tilde\eta'A')
  =
  \tr(A'\tilde\eta B\tilde\eta')
  =
  \tr\left(\frac{A+A'}{2}\tilde\eta B\tilde\eta'\right),
\end{align}
so, as ${\auxproxvarepsgen}^{-1}$ and $\longrunvarmatrixeta$ are symmetric and $A+A'=\iota\iota'-I_T$, we have
\begin{align}
\MoveEqLeft|\tr(A\tilde\eta({\auxproxvarepsgen}^{-1}-{\longrunvarmatrixeta }^{-1})\tilde\eta')|
=\frac{1}{2}|\tr((\iota\iota'-I_T)\tilde\eta({\auxproxvarepsgen}^{-1}-{\longrunvarmatrixeta }^{-1})\tilde\eta')|\\
&\le |\tr(\iota'\tilde\eta({\auxproxvarepsgen}^{-1}-{\longrunvarmatrixeta }^{-1})\tilde\eta'\iota)|
+|\tr(\tilde\eta({\auxproxvarepsgen}^{-1}-{\longrunvarmatrixeta }^{-1})\tilde\eta')|\\
&\le \Frob{\left(\Loading\trans  \longrunvarmatrixeta^{-1} \Loading \right)^{-1}}
\left(\Frob{\iota'\tilde\eta\longrunvarmatrixeta^{-1}\Loading}^2+\Frob{\tilde\eta\longrunvarmatrixeta^{-1}\Loading}^2\right)\\
&=O(n^{-1})(O_P(nT)+O_P(nT))=O_P(T),
\end{align}
using \cref{sumLoadingeta,eta1Loading,lem:bddFacInv} of \cref{lemma:extrasFeasible}. 
\end{proof}
\section{Additional Derivations}
\begin{proof}[Proof of \cref{lem:FactEst}]
	
%
As \cite{MoonPerron2004}, we take  $H_K=\frac{\tilde f'\tilde f}{T}\frac{\Loading'\bar\Loading}{n}$.
  First note that from the definitions of $H_K$ and $\hat\Lambda$ and using $\tilde \varepsilon=\tilde f\Loading'+\tilde\eta$ we have
\begin{align}\label{expandDiff}
  \hat \Loading-\Loading H_K 
 =\frac{1}{nT}(\tilde\varepsilon'\tilde\varepsilon-\Loading \tilde f'\tilde f\Loading')\bar\Loading
 =\frac{1}{nT}(\tilde\eta'\tilde f\Loading'+\Loading\tilde f'\tilde\eta+\tilde\eta'\tilde\eta)\bar\Loading,
\end{align}
so that 
\begin{align}
  \Frob{\Loading H_K -\hat \Loading}
  \le&\frac{\Frob{\tilde\eta'\tilde f\Loading'\bar\Loading}}{nT}
  +\frac{\Frob{\Loading\tilde f'\tilde\eta\bar\Loading}}{nT}
  +\frac{1}{nT}\Frob{\tilde\eta'\tilde\eta
  \bar\Loading}\\
  \le &2\sqrt{\frac{n}{T}}\frac{\Frob{\tilde\eta'\tilde f}}{\sqrt{nT}}\frac{\Frob{\Loading}}{\sqrt{n}}\frac{\Frob{\bar\Loading}}{\sqrt{n}}
  +\frac{1}{nT}\rnorm{\tilde\eta'\tilde\eta}\Frob{
  \bar\Loading}.\label{eq:estFacDif}
\end{align}
By the definition of $\bar\Loading$, $\Frob{\bar\Loading}=\sqrt{nK}=O(\sqrt{n})$. 
We have
\begin{align}
  \expec \Frob{\tilde\eta'\tilde f}^2
  =&\expec\sum_{k=1}^{K}\sum_{i=1}^{n}\left(\sum_{t=1}^{T}f_{kt}\eta_{it}\right)^2\\
  =&\sum_{k=1}^{K}\sum_{i=1}^{n}\sum_{t=1}^{T}\sum_{s=1}^{T}\ACFgamma{i}(t-s)\ACFgammaF{k}(t-s)\\
  \le& Mn\sum_{k=1}^{K}\sum_{t=1}^{T}\sum_{s=1}^{T}|\ACFgammaF{k}(t-s)|\\
  =& Mn\sum_{k=1}^{K}\sum_{m=-(T-1)}^{T-1}(T-|m|)|\ACFgammaF{k}(m)|=O(nT),
\end{align}
for some finite constant $M$, using that, thanks to \cref{assumption:serial correlation}, $\ACFgamma{i}(t-s)$ is bounded uniformly in $i$ and $t-s$. 
Thus, each term of the first summand in \cref{eq:estFacDif} is $O_p(1)$. 

Finally, we consider the second summand, which is treated differently from \cite{MoonPerron2004}. 
We obtain $\Frob{\Loading H_K -\hat \Loading}= o_p(1)$ if we can indeed show that $\rnorm{\tilde\eta'\tilde\eta}=o_p(\sqrt{n}T)$ (\cite{MoonPerron2004} only use $\Frob{\tilde\eta'\tilde\eta}=O_p(\sqrt{n}T)$). For this, note that $\frac{1}{T}\tilde\eta'\tilde\eta=\frac{1}{T}\sum_{t=1}^{T} \tilde\eta_{\cdot,t}\tilde\eta_{\cdot,t}'$, which can be considered an approximation to $\Gamma_\eta :=\diag(\ACFgamma{1}(0),\dots,\ACFgamma{n}(0))$, the  $n\times n$ cross-sectional covariance matrix of the $\eta$. From \cref{assumption:serial correlation}, $\rnorm{\Gamma_\eta}<\infty$. We now show that indeed the approximation works. Using Isserlis' Theorem to write $\expec[\eta_{i,t}^2\eta_{i,s}^2]=2\ACFgamma{i}(t-s)^2+\expec[\eta_{i,t}^2]\expec[\eta_{i,s}^2]$, we have
\begin{align}
  \MoveEqLeft\expec\Frob{\frac{1}{T}\sum_{t=1}^{T} \tilde\eta_{\cdot,t}\tilde\eta_{\cdot,t}'-\Gamma_\eta}^2
  =\sum_{i=1}^{n}\sum_{j=1}^{n}\expec\left(\frac{1}{T}\sum_{t=1}^{T}\eta_{i,t}\eta_{j,t}-\expec[\eta_{i,t}\eta_{j,t}]\right)^2\\
  &=\sum_{i=1}^{n}\sum_{j=1}^{n}\frac{1}{T^2} \sum_{t=1}^{T}\sum_{s=1}^{T}
  \expec[\eta_{i,t}\eta_{j,t}\eta_{i,s}\eta_{j,s}]-\expec[\eta_{i,t}\eta_{j,t}]\expec[\eta_{i,s}\eta_{j,s}]\\
  &=\sum_{i=1}^{n}\frac{1}{T^2} \sum_{t=1}^{T}\sum_{s=1}^{T}
  2\ACFgamma{i}(t-s)^2\\
  &\quad+\sum_{i\neq j}^{n}\frac{1}{T^2} \sum_{t=1}^{T}\sum_{s=1}^{T}
  \ACFgamma{i}(t-s)\ACFgamma{j}(t-s)\\
  &=O(n/T)+O(n^2/T).
\end{align}
Conclude that the difference in Frobenius norm is $O_p(n/\sqrt{T})$.\footnote{\label{footnoteGaussianity}Note that, even without Gaussianity, this conclusion holds as long as the long-run variances of the $\{\eta_{i,t}^2\}$ are uniformly bounded.} Thus,
\begin{align}
  \rnorm{\tilde\eta'\tilde\eta}
  \le & \Frob{\sum_{t=1}^{T} \tilde\eta_{\cdot,t}\tilde\eta_{\cdot,t}'-T\Gamma_\eta}
  +\rnorm{T\Gamma_\eta}\\
  =&O_p(n\sqrt{T})+O(T)=o_p(\sqrt{n}T).
\end{align}
Finally, we show the boundedness properties of $H_K$. First note that 
\begin{align}
  \Frob{H_K}\le \frac{\Frob{\tilde f'\tilde f}}{T}\frac{\Frob{\Loading}}{\sqrt{n}}\frac{\Frob{\bar\Loading}}{\sqrt{n}}=O_P(1).
\end{align}
To show boundedness of the inverse, we will show that the limiting eigenvalues of $H_K$ are positive. Introduce $\Gamma_f:=\diag(\ACFgammaF{1}(0),\dots,\ACFgammaF{K}(0)) $, the  $K\times K$ covariance matrix of the $f$, and write
\begin{align}
  \rnorm{H_K-\Gamma_f \frac{\Loading'\bar\Loading}{n}}
  \le \Frob{\frac{\Loading'\bar\Loading}{n}}\Frob{\frac{\tilde f'\tilde f}{T}-\Gamma_F}=O_P(1)o_P(1),
\end{align}
where the latter follows from \cref{assumption:serial correlation}. As $\Gamma_F$ has full rank, it is sufficient to show that the eigenvalues of $\frac{\Loading'\bar\Loading}{n}$ are bounded away from zero. $\bar\Loading$ is defined through the eigenvectors of $\tilde\varepsilon'\tilde\varepsilon/(nT)$. As the eigenvalues of $\tilde\varepsilon'\tilde\varepsilon$ are closely related to those of $\Loading \tilde f'\tilde f\Loading'$, we can use this relation to learn about the rank of $\Loading'\bar\Loading$.
Formally, define $D$ to be the $K\times K$ matrix with the $K$ largest eigenvalues of  $\tilde\varepsilon'\tilde\varepsilon/(nT)$. Then, from the definition of $\bar \Loading$,
\begin{align}
  D=\frac{\bar \Loading'}{\sqrt{n}}\frac{\tilde\varepsilon'\tilde\varepsilon}{nT}\frac{\bar \Loading}{\sqrt{n}}.
\end{align}
Recalling some of the above results we obtain
\begin{align}\label{approxCovm}
  \rnorm{\frac{\tilde\varepsilon'\tilde\varepsilon}{nT}-\frac{\Loading \tilde f'\tilde f\Loading'}{nT}}=o_P(n^{-1/2}),
\end{align}
so that 
\begin{align}
  D=&\frac{\bar \Loading'}{\sqrt{n}}\frac{\Loading \tilde f'\tilde f\Loading'}{nT}\frac{\bar \Loading}{\sqrt{n}}+o_P(n^{-1/2})
  =\frac{\bar \Loading'\Loading}{n}\Gamma_f\frac{\Loading'\bar \Loading}{n}+o_P(1).
\end{align}
As the $K$th largest eigenvalue of $\tilde\varepsilon'\tilde\varepsilon/(nT)$ is bounded away from zero (using \cref{approxCovm} the nonzero limiting eigenvalues are given by those of $\limitfactorloadings \Gamma_F$, a product of two rank $K$ matrices), so must the limit of $\frac{\Loading'\bar \Loading}{n}$ and thus $H_K$. 
\end{proof}

\clearpage

\begin{proof}[Proof of \Cref{lemma:hatCS}]
We split the difference in three parts: one for replacing $\auxproxvarepsgen$ with $\hatproxvarepsgen$, one to take care of the initial value, and one for estimating the correction term. Thus
$\hatCS -\auxCS = I-II - III$, with 
\begin{align}
I & = \frac{1}{\sqrt n T} \tr(A'\tilde \varepsilon \left(\hatproxvarepsgen^{-1} - {\auxproxvarepsgen}^{-1}\right)\tilde\varepsilon') \\
II  
&=\frac{1}{\sqrt n T} \sum_{t=2}^T \varepsilon_{\cdot, 1}\trans \hatproxvarepsgen^{-1}\varepsilon_{\cdot, t}
\\
III & = \frac{1}{\sqrt n}\sum_{i=1}^n \left( \frac{\hatoslongrunvareta{i}}{\hatlongrunvareta{i}} - \frac{\oslongrunvareta{i}}{\longrunvareta{i}} \right)
.
\end{align}

For part $I$, insert \cref{auxinvdev,hatinvdev} to find
\begin{align}
  |I|&=\frac{1}{\sqrt n T}|\tr(\tilde\varepsilon'A'\tilde\varepsilon(\hatproxvarepsgen^{-1} - {\auxproxvarepsgen}^{-1}))|\\
  &\le \frac{1}{\sqrt n T}|\tr(\tilde\varepsilon'A'\tilde\varepsilon(\hatlongrunvarmatrixeta^{-1} - \longrunvarmatrixeta^{-1})|\\
  &+\frac{1}{\sqrt n T} \left|\tr\left(\hat\Loading\trans \hatlongrunvarmatrixeta^{-1}\tilde\varepsilon'A'\tilde\varepsilon\hatlongrunvarmatrixeta^{-1}\hat\Loading\left(\hat\Loading\trans \hatlongrunvarmatrixeta^{-1}\hat\Loading \right)^{-1} 
-\Loading\trans \longrunvarmatrixeta^{-1}\tilde\varepsilon'A'\tilde\varepsilon\longrunvarmatrixeta^{-1}\Loading\left(\Loading\trans \longrunvarmatrixeta^{-1}\Loading \right)^{-1}  \right)\right|\\
  &\le \frac{1}{\sqrt n T}\tr(|\tilde\varepsilon'A'\tilde\varepsilon|)\sup_i |\hatlongrunvareta{i} - \longrunvareta{i}|\\
  &+\frac{1}{\sqrt n T} \Frob{\hatlongrunvarmatrixeta^{-1}\hat\Loading\left(\hat\Loading\trans \hatlongrunvarmatrixeta^{-1}\hat\Loading \right)^{-1} \hat\Loading\trans \hatlongrunvarmatrixeta^{-1}
-\longrunvarmatrixeta^{-1}\Loading\left(\Loading\trans \longrunvarmatrixeta^{-1}\Loading \right)^{-1}\Loading\trans \longrunvarmatrixeta^{-1}}\Frob{\tilde\varepsilon'A'\tilde\varepsilon}.\\
\end{align}
The first summand is $o_P(1)$ thanks to \cref{assumption:estimators} and, similar to the arguments in \cref{epsilonY},
\begin{align}
 2\tr(|\tilde\varepsilon'A'\tilde\varepsilon|) 
 =
 \tr(|\tilde\varepsilon'(\iota\iota'-I_T)\tilde\varepsilon|) 
  \le\Frob{\iota'\tilde\varepsilon}^2
  +\Frob{\tilde\varepsilon}^2
  =O_P(nT),
\end{align}
as $\Frob{\tilde\varepsilon}\le \Frob{\Loading}\Frob{\tilde f}+\Frob{\tilde\eta}=O(\sqrt{n})O_P(\sqrt{T})+O_P(\sqrt{nT})$ and, using \cref{sumeta,sumf} of \cref{lemma:extrasFeasible}, a similar bound holds for $\Frob{\iota'\tilde\varepsilon}$.
The second part is $o_P(1)$ thanks to \cref{epsilonY} and \cref{LamOmUp} of \cref{lem:estimatedTerms}.

For $II$, we have
\begin{align}
\sqrt{n}T  II
  &\le \rnorm{\hatproxvarepsgen^{-1}}\Frob{\varepsilon_{\cdot, 1}}(\Frob{\iota'\tilde\varepsilon}+\Frob{\varepsilon_{\cdot, 1}})\\
  &=O_p(1)O_P(\sqrt{n})(O_P(\sqrt{nT})+O_P(\sqrt{n})=O_P(n\sqrt{T}),
\end{align}
where $\Frob{\varepsilon_{\cdot, 1}}\le \Frob{\Loading}\Frob{\tilde f_{\cdot, 1}}+\Frob{\tilde\eta{\cdot, 1}}=O_P(\sqrt{n})$ and $\rnorm{\hatproxvarepsgen^{-1}}=O_P(1)$ follows from \cref{assumption:estimators} and \cref{LamOmUp} of \cref{lem:estimatedTerms} implying
\begin{align}
\rnorm{\hatproxvarepsgen^{-1}-{\auxproxvarepsgen}^{-1}}&=O_P(n^{-1/2})\text{ and}\\
  \rnorm{{\auxproxvarepsgen}^{-1}}
  &\le\rnorm{\longrunvarmatrixeta^{-1}}+\rnorm{\longrunvarmatrixeta^{-1}}^2 \Frob{\Loading}^2 \Frob{\left(
\Loading\trans  \longrunvarmatrixeta^{-1} \Loading \right)^{-1}}\\
&=O(1)+O(1)O(n)O(n^{-1})=O(1),
\end{align}
using \cref{assumption:factor_loadings,assumption:serial correlation} and \cref{lem:bddFacInv} of \cref{lemma:extrasFeasible}. We conclude that $II=O_P\left(\frac{\sqrt{n}}{\sqrt{T}}\right)=o_P(1)$.

Finally, to show $III = o_P(1)$ we first note that \cref{assumption:serial correlation} together with \cref{assumption:estimators} imply
\begin{align}\label{eq:estimate_longrunvar_bounds}
0< \inf_{i\in\SN} \hatlongrunvareta{i}\leq \sup_{i\in\SN} \hatlongrunvareta{i}<\infty \quad \text{and} \quad \sup_{i\in\SN} \hatoslongrunvareta{i}<\infty.\end{align}
Then, we rewrite $III$ as
\begin{align}III = \frac{1}{\sqrt n} \sum_{i=1}^n \frac{1}{\longrunvareta{i}}(\hatoslongrunvareta{i} - \oslongrunvareta{i}) + \frac{1}{\sqrt n }\sum_{i=1}^n \frac{\hatoslongrunvareta{i}}{\hatlongrunvareta{i} \longrunvareta{i}}(\longrunvareta{i} - \hatlongrunvareta{i}) \end{align}
to see that both parts converge to zero in probability.
\end{proof}

\subsection{Auxiliary Lemmas}
\begin{lemma}\label{lem:estimatedTerms}
Consider the factor estimates and the $H_K$ from \cref{lem:FactEst}. Then, under \cref{assumption:estimators,assumption:factor_loadings,assumption:framework,assumption:rates,assumption:serial correlation}, under $\lawMP$ or $\lawPANIC$ and as $n,T \to \infty$, we have
\begin{enumerate}
	\item $\Frob{\left(\hat\Loading\trans \hatlongrunvarmatrixeta^{-1}\hat\Loading \right)^{-1} -\left(H_K\trans\Loading\trans \longrunvarmatrixeta^{-1}\Loading H_K\right)^{-1}}=o_P(n^{-3/2})$,\label{difInv} and
	\item $\Frob{\hatlongrunvarmatrixeta^{-1}\hat\Loading\left(\hat\Loading\trans \hatlongrunvarmatrixeta^{-1}\hat\Loading \right)^{-1} \hat\Loading\trans \hatlongrunvarmatrixeta^{-1}
-\longrunvarmatrixeta^{-1}\Loading\left(\Loading\trans \longrunvarmatrixeta^{-1}\Loading \right)^{-1}\Loading\trans \longrunvarmatrixeta^{-1}}=o_P(n^{-1/2})$. \label{LamOmUp}
	\end{enumerate}
\end{lemma}
\begin{proof}
We start by noting that $\Frob{H_K\trans\Loading\trans \longrunvarmatrixeta^{-1}\Loading H_K-\hat\Loading\trans \hatlongrunvarmatrixeta^{-1}\hat\Loading}=o_P(\sqrt{n})$: This follows from $\Frob{\Loading H_K-\hat\Loading}=o_P(1)$ (\cref{lem:FactEst}) and $\rnorm{\longrunvarmatrixeta^{-1}-\hatlongrunvarmatrixeta^{-1}}=o_P(n^{-1/2})$ (\cref{assumption:estimators}) in combination with $H_K$ being bounded and $\Frob{\Loading}=O(\sqrt{n})$. 
Next, we have that 
\begin{align}
  \Frob{\frac{1}{n}H_K\trans\Loading\trans \longrunvarmatrixeta^{-1}\Loading H_K}
  \le\Frob{H_K}^2\frac{\Frob{\Loading}^2}{n}\rnorm{\longrunvarmatrixeta^{-1}}=O(1),
\end{align}
and 
\begin{align}
\mineig{\frac{1}{n}H_K\trans\Loading\trans \longrunvarmatrixeta^{-1}\Loading H_K}
=&\rnorm{H_K^{-1}\left(\frac{1}{n}\Loading\trans \longrunvarmatrixeta^{-1}\Loading\right)^{-1} (H_K')^{-1}}^{-1} \\ \ge&\Frob{H_K^{-1}}^{-2}\rnorm{\left(\frac{1}{n}\Loading\trans \longrunvarmatrixeta^{-1}\Loading\right)^{-1}}^{-1},
\end{align}
which is bounded away from zero thanks to $\Frob{H_K^{-1}}$ being bounded and \cref{lem:bddFacInv} of \cref{lemma:extrasFeasible}. Thus, we can restrict attention to a compact subset of the invertible matrices on $\mathds{R}^K$, on which the matrix inverse is uniformly continuous. Therefore, $\Frob{\frac{1}{n}H_K\trans\Loading\trans \longrunvarmatrixeta^{-1}\Loading H_K-\frac{1}{n}\hat\Loading\trans \hatlongrunvarmatrixeta^{-1}\hat\Loading}=o_P(n^{-1/2})$ implies the same for 
$\Frob{\left(\frac{1}{n}H_K\trans\Loading\trans \longrunvarmatrixeta^{-1}\Loading H_K\right)^{-1}-\left(\frac{1}{n}\hat\Loading\trans \hatlongrunvarmatrixeta^{-1}\hat\Loading \right)^{-1} }$.

For \cref{LamOmUp}, let $a= \longrunvarmatrixeta^{-1}\Loading H_K$ and $b= \left(H_K'\Loading\trans \longrunvarmatrixeta^{-1}\Loading H_K\right)^{-1}$ and define $\hat a=\hat\longrunvarmatrixeta^{-1}\hat\Loading$ and $\hat b=\left(\hat\Loading\trans \hat\longrunvarmatrixeta^{-1}\hat\Loading \right)^{-1}$ analogously. Thus \begin{align}
\MoveEqLeft\Frob{\hatlongrunvarmatrixeta^{-1}\hat\Loading\left(\hat\Loading\trans \hatlongrunvarmatrixeta^{-1}\hat\Loading \right)^{-1} \hat\Loading\trans \hatlongrunvarmatrixeta^{-1}
-\longrunvarmatrixeta^{-1}\Loading\left(\Loading\trans \longrunvarmatrixeta^{-1}\Loading \right)^{-1}\Loading\trans \longrunvarmatrixeta^{-1}}\\
&=\Frob{\hat a\hat b\hat a'-aba'}\\
&\le \Frob{\hat a -a}\Frob{\hat b}\Frob{\hat a}+\Frob{a}\Frob{\hat b -b}\Frob{\hat a}
+\Frob{a}\Frob{b}\Frob{\hat a -a}.
  \end{align}
  From \cref{assumption:factor_loadings} and $H_K$ being bounded it follows that $\Frob{b}=O_P(n^{-1})$ and in combination with \cref{assumption:serial correlation} we obtain 
  \begin{align}
  \Frob{a}\le \rnorm{\longrunvarmatrixeta^{-1}}\Frob{\Loading}\Frob{H_K}=O_P(\sqrt{n}).
\end{align}
From \cref{difInv}, $\Frob{\hat b -b}=o_P(n^{-3/2})$ so that also $\Frob{\hat b}=O_P(n^{-1})$. Finally, we have
  \begin{align}
  \Frob{\hat a -a}\le& \rnorm{\hatlongrunvarmatrixeta^{-1}-\longrunvarmatrixeta^{-1}}\Frob{\hat\Loading}\Frob{H_K}+\rnorm{\longrunvarmatrixeta^{-1}}\Frob{\hat\Loading-\Loading H_K}\\
  =&o_P(n^{-1/2})O_P(\sqrt{n})O_P(1)+O(1)o_P(1)=o_P(1),
\end{align}
where $\rnorm{\hatlongrunvarmatrixeta^{-1}-\longrunvarmatrixeta^{-1}}=o_P(n^{-1/2})$ by \cref{assumption:estimators} and $\Frob{\hat\Loading-\Loading H_K}=o_P(1)$ by \cref{lem:FactEst}. Combining all these results indeed yields the correct rate. 
\end{proof}

\begin{proof}[Independent proof of \cref{prop:local_powers_BN_and_MP}]\label{prooflocalpower}
Here we demonstrate the joint asymptotic normality required to apply the second part of \cref{equivalence}.
We divide the proof into two parts. In Part~A, we prove the theorem for $P_a$ while in Part~B we discuss $t_a$. We omit the proofs concerning $P_b$ and $t_b$ as they follow along  the same lines.

\emph{Part A:}
First, we establish the joint convergence, under $\lawMP$ and $\lawPANIC$, of $P_a$ and the local likelihood ratio. As already hinted at in \cref{remark:equivalence}, the results in \cref{expandPANIC,expaMP} imply that we only have to show this convergence once to get the powers in both experiments, as both likelihood ratios are asymptotically equivalent and the models coincide under the hypothesis. Having established this joint convergence, an application of Le Cam's third lemma will lead to the asymptotic distribution of $P_a$ under $\lawMP[h]$ and $\lawPANIC[h]$. 

Specifically, \cref{lemma:NormalityCS,lemma:PANICCS} imply that the limiting distributions of $\left(P_a,\log \frac{ \rd \lawPANIC[h]}{\rd \lawPANIC}\right)$ and $\left(P_a,\log \frac{ \rd \lawMP[h]}{\rd \lawMP}\right)$ are equal to that of $\left(P_a,h\etaCS-\frac{1}{4}h^2\right)$, under $\lawMP$ and $\lawPANIC$. From Lemma 1 and Lemma 2 in \cite{BaiNg2010} we see that $P_a$ is adaptive with respect to the estimation of nuisance parameters while Lemma A.2 in \cite{MoonPerron2004} shows that $\frac{1}{n T^2}\sum_{i=1}^n E_{i,-1}^\prime E_{i,-1}$ converges in probability to $\frac{1}{2} \omega^2$. Therefore, $P_a$ is asymptotically equivalent to $\tilde P_ a = \frac{\frac{1}{\sqrt n T} \sum_{i=1}^n E_{i,-1}^\prime\Delta E_{i}- \frac{1}{\sqrt n} \sum_{i=1}^n \oslongrunvareta{i}}{\sqrt{\phi^4/2}}$.

Under $\lawMP$ or $\lawPANIC$, we can compute the asymptotic distribution of all possible linear combinations of $\tilde P_a$ and $\etaCS$ by an application of \cref{lemma:normality}. For all $\alpha$, $\beta$ in $\mathds R$, we find, using $a_{i,n,T}=\alpha \frac{\longrunvaretaT{i} }{\sqrt{\phi^4/2}}+\beta$ in \cref{lemma:normality},
\begin{align}
\alpha \tilde P_a + \beta \etaCS \vto N\left(0, \left(\alpha^2 + \alpha\beta \sqrt{\frac{2 \omega^4}{\phi^4}} + \frac{\beta^2}{2}\right)\right).
\end{align}

Thus, the Cram\'er-Wold theorem and the asymptotic equivalence of $P_a$ and $\tilde P_a$, yield, still under $\lawMP$ or $\lawPANIC$,
\begin{align}\left( P_a ,  \etaCS \right)
\vto
N\left( \left( \begin{array}{c} 0 \\ 0\end{array}\right) , \left( \begin{array}{cc} 1 &   \sqrt{\frac{\omega^4}{ 2\phi^4}}\\  \sqrt{\frac{\omega^4}{ 2\phi^4}} & 1/2\end{array}\right)  \right).\end{align}
Equivalently,
\begin{align}\left( P_a ,  \log\frac{ \rd \law[h]}{\rd \law} \right)
\vto
N\left( \left( \begin{array}{c} 0 \\ -\frac{1}{4}h^2\end{array}\right) , \left( \begin{array}{cc} 1 &   h\sqrt{\frac{\omega^4}{ 2\phi^4}}\\  h\sqrt{\frac{\omega^4}{ 2\phi^4}} & 1/2 h^2\end{array}\right)  \right).\end{align}

Applying Le Cam's third lemma, we obtain $P_a \vto N\left(h\sqrt{\frac{\omega^4}{2\phi^4}},1\right)$ under $\lawMP[h]$ or $\lawPANIC[h]$.

\emph{Part B:}
As far as $t_a$ is concerned, we recall that $t_a$ is adaptive with respect to the estimation of nuisance parameters (see proofs of Theorem 2a) and b) in \cite{MoonPerron2004}) and that $\frac{1}{nT^2}\sum_{t=1}^T Y_{\cdot,t-1}^\prime Q_\gamma Y_{\cdot,t-1}$ converges in probability to $\frac{1}{2}\omega^2$ under $\lawMP[0]$ . Thus,
$t_a$ is asymptotically equivalent to 
\begin{align}
  \tilde t_a =  \frac{\frac{1}{\sqrt n T}\sum_{i=1}^n  Y_{\cdot,t-1}^\prime Q_\Loading \Delta Y_{\cdot,t-1}- \sqrt n \sum_{i=1}^n \oslongrunvareta{i}}{\sqrt{\phi^4/2}}.
\end{align}
  Moreover, we have \begin{align}
  \MoveEqLeft\frac{1}{\sqrt n T} \sum_{t=1}^T Y_{\cdot,t}^\prime Q_\Loading \Delta Y_{\cdot,t-1} 
  = \frac{1}{\sqrt n T} \sum_{t=1}^T E_{\cdot,t}^\prime Q_\Loading \Delta E_{\cdot,t-1}\\
   & = \frac{1}{\sqrt n T} \sum_{t=1}^T E_{\cdot,t}^\prime \Delta E_{\cdot,t-1} - \frac{1}{\sqrt n T}  \sum_{t=1}^T E_{\cdot,t}^\prime \Loading(\Loading^\prime\Loading)^{-1}\Loading \Delta E_{\cdot,t-1} \\
   &= \frac{1}{\sqrt n T} \sum_{i=1}^n E_{-1,i}^\prime\Delta E_{i} + o_P(1),
\end{align}
 where the last equality follows from the proof of Lemma~2 c) in \cite{MoonPerron2004}. Therefore, $t_a$ is asymptotically equivalent to $\tilde P_a$. Thus, following the same steps as in Part~A,  we find $t_a \vto N\left(h \sqrt{\frac{\omega^4}{ 2\phi^4}}, 1 \right)$ under $\lawMP[h]$ or $\lawPANIC[h]$.
\end{proof}

\clearpage
\setcounter{page}{1}
\section*{Additional Monte-Carlo Results}
In this appendix we present sizes and powers for additional DGPs and additional long-run variance estimates. The first subsection provides sizes and powers for additional DGPs. In the second subsection, we consider the same DGPs as in \cref{sec:MC,appendix:DGP}, but with long-run variances estimated using the \cite{NeweyWest1994} bandwidth. \Cref{table:sizesNe1_0.4_0.4_BN_,table:sizesNeThree3_0.4_0.4_BN_,table:sizesNeT1_0.4_0.4_BN_} are analogous to \cref{table:sizesAn1_0.4_0.4_BN_,table:sizesAnThree3_0.4_0.4_BN_,table:sizesAnT1_0.4_0.4_BN_}. \Cref{DifMPNpowersMPDifNeiidBN0.811,NpowersNeiidBN0.811,NpowersStatNeiidBN0.811,DifNpowersNeiidBN111,NpowersCorrNeMABN0.811,NpowersCorrNeARBN0.811,NpowersThreeNeiidBN0.831,NpowersHetNeiidBN0.810,NpowersTNeiidBN0.811} are analogous to \cref{DifMPNpowersMPDifiidBN0.811,NpowersiidBN0.811,NpowersStatiidBN0.811,DifNpowersiidBN111,NpowersCorrMABN0.811,NpowersCorrARBN0.811,NpowersThreeiidBN0.831,NpowersHetiidBN0.810,NpowersTiidBN0.811}. In general, the sizes for the MA case are slightly better controlled with the \cite{NeweyWest1994} bandwidth, at the expense of slightly lower power for small sample sizes.
\subsection*{Sizes and Powers in Additional DGPs}
\label{appendix:DGP}

First, \cref{NpowersCorrMABN0.811,NpowersCorrARBN0.811} consider the powers in the presence of  MA and AR serial correlation, respectively. The results are similar to those for i.i.d\ innovations. \Cref{NpowersStatiidBN0.811} shows the results when the factor innovations are overdifferenced, i.e., the factor is stationary under the hypothesis. The powers appear to be unaffected. \Cref{NpowersThreeiidBN0.831} considers the case of the dependence being generated by three factors, with the corresponding sizes reported in \cref{table:sizesAnThree3_0.4_0.4_BN_}. For very small sample sizes, powers of both tests are affected, but generally the results are similar also here. 

We now consider deviations from our assumptions.
\cref{NpowersHetiidBN0.810} reports the size-corrected powers of our tests against heterogeneous alternatives of the form 
\begin{align}\label{localAlternatives}
	\rho_{i}=1+\frac{hU_i}{\sqrt{n}T},
\end{align}
where the $U_i$ are i.i.d.\ random variables with mean one. We draw the $U_i$ from a Uniform(0.2,1.8) distribution. Once again, the finite-sample behaviour does not appear to be affected significantly, for both small and large samples.

Finally, we consider non-Gaussian innovations. \cref{NpowersTiidBN0.811} reports size corrected powers with the innovations drawn from a $t$ distribution with five degrees of freedom. The corresponding sizes are reported in \cref{table:sizesAnT1_0.4_0.4_BN_}. Also here, the conclusions remain the same.

\begin{figure}[ht]
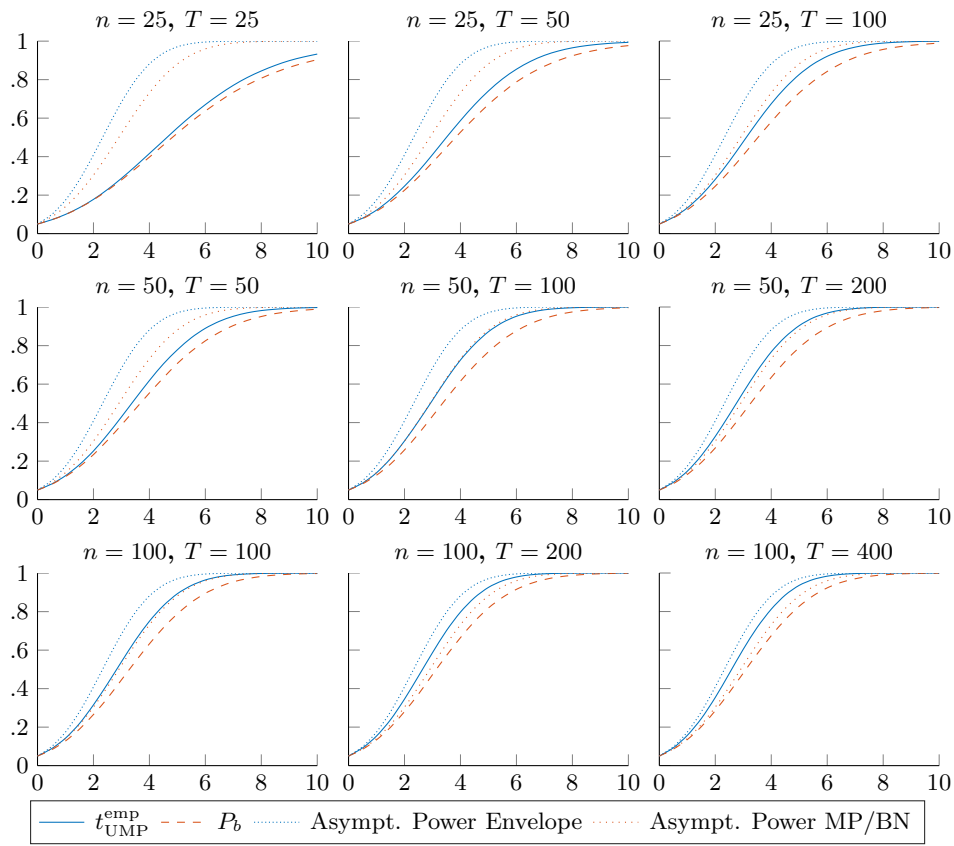
\footnotesize\begin{subfigure}[b]{0.3466\textwidth}\input{graphs/stamps/NpowersCorr25_25_1_0.8_MA_0.4_MA_0.4_BN_.tex}\end{subfigure}\begin{subfigure}[b]{0.3267\textwidth}\input{graphs/stamps/NpowersCorr50_25_1_0.8_MA_0.4_MA_0.4_BN_.tex}\end{subfigure}\begin{subfigure}[b]{0.3267\textwidth}\input{graphs/stamps/NpowersCorr100_25_1_0.8_MA_0.4_MA_0.4_BN_.tex}\end{subfigure}\\ \begin{subfigure}[b]{0.3466\textwidth}\input{graphs/stamps/NpowersCorr50_50_1_0.8_MA_0.4_MA_0.4_BN_.tex}\end{subfigure}\begin{subfigure}[b]{0.3267\textwidth}\input{graphs/stamps/NpowersCorr100_50_1_0.8_MA_0.4_MA_0.4_BN_.tex}\end{subfigure}\begin{subfigure}[b]{0.3267\textwidth}\input{graphs/stamps/NpowersCorr200_50_1_0.8_MA_0.4_MA_0.4_BN_.tex}\end{subfigure}\\ \begin{subfigure}[b]{0.3466\textwidth}\input{graphs/stamps/NpowersCorr100_100_1_0.8_MA_0.4_MA_0.4_BN_.tex}\end{subfigure}\begin{subfigure}[b]{0.3267\textwidth}\input{graphs/stamps/NpowersCorr200_100_1_0.8_MA_0.4_MA_0.4_BN_.tex}\end{subfigure}\begin{subfigure}[b]{0.3267\textwidth}\input{graphs/stamps/NpowersCorr400_100_1_0.8_MA_0.4_MA_0.4_BN_.tex}\end{subfigure}\\ \begin{subfigure}[b]{\textwidth}\centering\input{graphs/stamps/lNpowersCorr400_100_1_0.8_MA_0.4_MA_0.4_BN_.tex}\end{subfigure}\normalsize\caption{\label{NpowersCorrMABN0.811}Size-corrected power of unit-root tests as a function of $-h$ for varying sample sizes in the PANIC framework with MA factor innovations and MA idiosyncratic parts and $\sqrt{\omega^4/\phi^4}=0.8$. Based on \num{100000} replications.}\end{figure}
\begin{figure}[ht]\footnotesize\begin{subfigure}[b]{0.3466\textwidth}\input{graphs/stamps/NpowersCorr25_25_1_0.8_AR_0.4_AR_0.4_BN_.tex}\end{subfigure}\begin{subfigure}[b]{0.3267\textwidth}\input{graphs/stamps/NpowersCorr50_25_1_0.8_AR_0.4_AR_0.4_BN_.tex}\end{subfigure}\begin{subfigure}[b]{0.3267\textwidth}\input{graphs/stamps/NpowersCorr100_25_1_0.8_AR_0.4_AR_0.4_BN_.tex}\end{subfigure}\\ \begin{subfigure}[b]{0.3466\textwidth}\input{graphs/stamps/NpowersCorr50_50_1_0.8_AR_0.4_AR_0.4_BN_.tex}\end{subfigure}\begin{subfigure}[b]{0.3267\textwidth}\input{graphs/stamps/NpowersCorr100_50_1_0.8_AR_0.4_AR_0.4_BN_.tex}\end{subfigure}\begin{subfigure}[b]{0.3267\textwidth}\input{graphs/stamps/NpowersCorr200_50_1_0.8_AR_0.4_AR_0.4_BN_.tex}\end{subfigure}\\ \begin{subfigure}[b]{0.3466\textwidth}\input{graphs/stamps/NpowersCorr100_100_1_0.8_AR_0.4_AR_0.4_BN_.tex}\end{subfigure}\begin{subfigure}[b]{0.3267\textwidth}\input{graphs/stamps/NpowersCorr200_100_1_0.8_AR_0.4_AR_0.4_BN_.tex}\end{subfigure}\begin{subfigure}[b]{0.3267\textwidth}\input{graphs/stamps/NpowersCorr400_100_1_0.8_AR_0.4_AR_0.4_BN_.tex}\end{subfigure}\\ \begin{subfigure}[b]{\textwidth}\centering\input{graphs/stamps/lNpowersCorr400_100_1_0.8_AR_0.4_AR_0.4_BN_.tex}\end{subfigure}\normalsize\caption{\label{NpowersCorrARBN0.811}Size-corrected power of unit-root tests as a function of $-h$ for varying sample sizes in the PANIC framework with AR factor innovations and AR idiosyncratic parts and $\sqrt{\omega^4/\phi^4}=0.8$. Based on \num{100000} replications.}\end{figure}
\begin{figure}[ht]\footnotesize\begin{subfigure}[b]{0.3466\textwidth}\input{graphs/stamps/NpowersStat25_25_1_0.8_iid_0.4_overdiffiid_0.4_BN_.tex}\end{subfigure}\begin{subfigure}[b]{0.3267\textwidth}\input{graphs/stamps/NpowersStat50_25_1_0.8_iid_0.4_overdiffiid_0.4_BN_.tex}\end{subfigure}\begin{subfigure}[b]{0.3267\textwidth}\input{graphs/stamps/NpowersStat100_25_1_0.8_iid_0.4_overdiffiid_0.4_BN_.tex}\end{subfigure}\\ \begin{subfigure}[b]{0.3466\textwidth}\input{graphs/stamps/NpowersStat50_50_1_0.8_iid_0.4_overdiffiid_0.4_BN_.tex}\end{subfigure}\begin{subfigure}[b]{0.3267\textwidth}\input{graphs/stamps/NpowersStat100_50_1_0.8_iid_0.4_overdiffiid_0.4_BN_.tex}\end{subfigure}\begin{subfigure}[b]{0.3267\textwidth}\input{graphs/stamps/NpowersStat200_50_1_0.8_iid_0.4_overdiffiid_0.4_BN_.tex}\end{subfigure}\\ \begin{subfigure}[b]{0.3466\textwidth}\input{graphs/stamps/NpowersStat100_100_1_0.8_iid_0.4_overdiffiid_0.4_BN_.tex}\end{subfigure}\begin{subfigure}[b]{0.3267\textwidth}\input{graphs/stamps/NpowersStat200_100_1_0.8_iid_0.4_overdiffiid_0.4_BN_.tex}\end{subfigure}\begin{subfigure}[b]{0.3267\textwidth}\input{graphs/stamps/NpowersStat400_100_1_0.8_iid_0.4_overdiffiid_0.4_BN_.tex}\end{subfigure}\\ \begin{subfigure}[b]{\textwidth}\centering\input{graphs/stamps/lNpowersStat400_100_1_0.8_iid_0.4_overdiffiid_0.4_BN_.tex}\end{subfigure}\normalsize\caption{\label{NpowersStatiidBN0.811}Size-corrected power of unit-root tests as a function of $-h$ for varying sample sizes in the PANIC framework with overdifferenced i.i.d.\ factor innovations and i.i.d.\ idiosyncratic parts and $\sqrt{\omega^4/\phi^4}=0.8$. The factor is stationary.  Based on \num{100000} replications.}\end{figure}
\begin{figure}[ht]\footnotesize\begin{subfigure}[b]{0.3466\textwidth}\input{graphs/stamps/NpowersThree25_25_3_0.8_iid_0.4_iid_0.4_BN_.tex}\end{subfigure}\begin{subfigure}[b]{0.3267\textwidth}\input{graphs/stamps/NpowersThree50_25_3_0.8_iid_0.4_iid_0.4_BN_.tex}\end{subfigure}\begin{subfigure}[b]{0.3267\textwidth}\input{graphs/stamps/NpowersThree100_25_3_0.8_iid_0.4_iid_0.4_BN_.tex}\end{subfigure}\\ \begin{subfigure}[b]{0.3466\textwidth}\input{graphs/stamps/NpowersThree50_50_3_0.8_iid_0.4_iid_0.4_BN_.tex}\end{subfigure}\begin{subfigure}[b]{0.3267\textwidth}\input{graphs/stamps/NpowersThree100_50_3_0.8_iid_0.4_iid_0.4_BN_.tex}\end{subfigure}\begin{subfigure}[b]{0.3267\textwidth}\input{graphs/stamps/NpowersThree200_50_3_0.8_iid_0.4_iid_0.4_BN_.tex}\end{subfigure}\\ \begin{subfigure}[b]{0.3466\textwidth}\input{graphs/stamps/NpowersThree100_100_3_0.8_iid_0.4_iid_0.4_BN_.tex}\end{subfigure}\begin{subfigure}[b]{0.3267\textwidth}\input{graphs/stamps/NpowersThree200_100_3_0.8_iid_0.4_iid_0.4_BN_.tex}\end{subfigure}\begin{subfigure}[b]{0.3267\textwidth}\input{graphs/stamps/NpowersThree400_100_3_0.8_iid_0.4_iid_0.4_BN_.tex}\end{subfigure}\\ \begin{subfigure}[b]{\textwidth}\centering\input{graphs/stamps/lNpowersThree400_100_3_0.8_iid_0.4_iid_0.4_BN_.tex}\end{subfigure}\normalsize\caption{\label{NpowersThreeiidBN0.831}Size-corrected power of unit-root tests as a function of $-h$ for varying sample sizes in the PANIC framework with i.i.d.\ factor innovations and i.i.d.\ idiosyncratic parts and $\sqrt{\omega^4/\phi^4}=0.8$. Dependence based on three factors. Based on \num{100000} replications.}\end{figure}

\begin{figure}[ht]\footnotesize\begin{subfigure}[b]{0.3466\textwidth}\input{graphs/stamps/NpowersHet25_25_1_0.8_iid_0.4_iid_0.4_BN__Uniform.tex}\end{subfigure}\begin{subfigure}[b]{0.3267\textwidth}\input{graphs/stamps/NpowersHet50_25_1_0.8_iid_0.4_iid_0.4_BN__Uniform.tex}\end{subfigure}\begin{subfigure}[b]{0.3267\textwidth}\input{graphs/stamps/NpowersHet100_25_1_0.8_iid_0.4_iid_0.4_BN__Uniform.tex}\end{subfigure}\\ \begin{subfigure}[b]{0.3466\textwidth}\input{graphs/stamps/NpowersHet50_50_1_0.8_iid_0.4_iid_0.4_BN__Uniform.tex}\end{subfigure}\begin{subfigure}[b]{0.3267\textwidth}\input{graphs/stamps/NpowersHet100_50_1_0.8_iid_0.4_iid_0.4_BN__Uniform.tex}\end{subfigure}\begin{subfigure}[b]{0.3267\textwidth}\input{graphs/stamps/NpowersHet200_50_1_0.8_iid_0.4_iid_0.4_BN__Uniform.tex}\end{subfigure}\\ \begin{subfigure}[b]{0.3466\textwidth}\input{graphs/stamps/NpowersHet100_100_1_0.8_iid_0.4_iid_0.4_BN__Uniform.tex}\end{subfigure}\begin{subfigure}[b]{0.3267\textwidth}\input{graphs/stamps/NpowersHet200_100_1_0.8_iid_0.4_iid_0.4_BN__Uniform.tex}\end{subfigure}\begin{subfigure}[b]{0.3267\textwidth}\input{graphs/stamps/NpowersHet400_100_1_0.8_iid_0.4_iid_0.4_BN__Uniform.tex}\end{subfigure}\\ \begin{subfigure}[b]{\textwidth}\centering\input{graphs/stamps/lNpowersHet400_100_1_0.8_iid_0.4_iid_0.4_BN__Uniform.tex}\end{subfigure}\normalsize\caption{\label{NpowersHetiidBN0.810}Size-corrected power of unit-root tests as a function of $-h$ for varying sample sizes in the PANIC framework with i.i.d.\ factor innovations and i.i.d.\ idiosyncratic parts and $\sqrt{\omega^4/\phi^4}=0.8$. Alternatives drawn from a Uniform(0.2,1.8) distribution. Based on \num{100000} replications.}\end{figure}

\begin{figure}[ht]
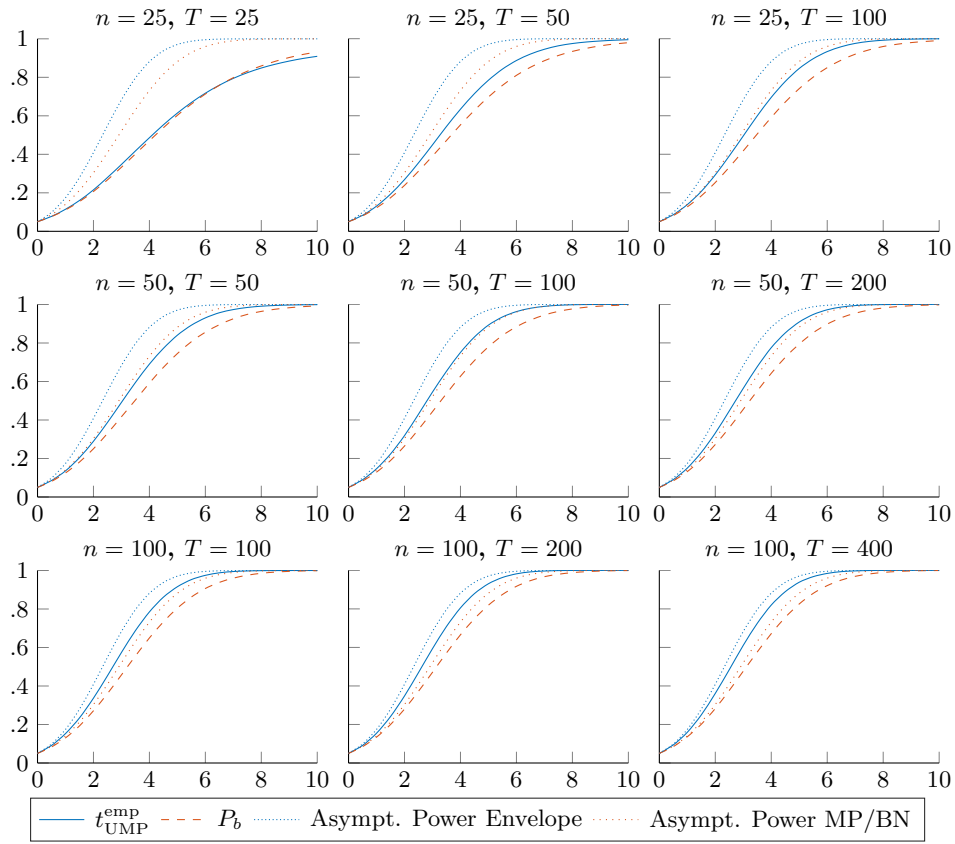
\footnotesize\begin{subfigure}[b]{0.3466\textwidth}\input{graphs/stamps/NpowersT25_25_1_0.8_iid_0.4_iid_0.4_BN_.tex}\end{subfigure}\begin{subfigure}[b]{0.3267\textwidth}\input{graphs/stamps/NpowersT50_25_1_0.8_iid_0.4_iid_0.4_BN_.tex}\end{subfigure}\begin{subfigure}[b]{0.3267\textwidth}\input{graphs/stamps/NpowersT100_25_1_0.8_iid_0.4_iid_0.4_BN_.tex}\end{subfigure}\\ \begin{subfigure}[b]{0.3466\textwidth}\input{graphs/stamps/NpowersT50_50_1_0.8_iid_0.4_iid_0.4_BN_.tex}\end{subfigure}\begin{subfigure}[b]{0.3267\textwidth}\input{graphs/stamps/NpowersT100_50_1_0.8_iid_0.4_iid_0.4_BN_.tex}\end{subfigure}\begin{subfigure}[b]{0.3267\textwidth}\input{graphs/stamps/NpowersT200_50_1_0.8_iid_0.4_iid_0.4_BN_.tex}\end{subfigure}\\ \begin{subfigure}[b]{0.3466\textwidth}\input{graphs/stamps/NpowersT100_100_1_0.8_iid_0.4_iid_0.4_BN_.tex}\end{subfigure}\begin{subfigure}[b]{0.3267\textwidth}\input{graphs/stamps/NpowersT200_100_1_0.8_iid_0.4_iid_0.4_BN_.tex}\end{subfigure}\begin{subfigure}[b]{0.3267\textwidth}\input{graphs/stamps/NpowersT400_100_1_0.8_iid_0.4_iid_0.4_BN_.tex}\end{subfigure}\\ \begin{subfigure}[b]{\textwidth}\centering\input{graphs/stamps/lNpowersT400_100_1_0.8_iid_0.4_iid_0.4_BN_.tex}\end{subfigure}\normalsize\caption{\label{NpowersTiidBN0.811}Size-corrected power of unit-root tests as a function of $-h$ for varying sample sizes in the PANIC framework with i.i.d.\ factor innovations and i.i.d.\ idiosyncratic parts and $\sqrt{\omega^4/\phi^4}=0.8$. Innovations drawn from a $t_5$ distribution. Note that the power envelopes refer to the Gaussian experiment.  Based on \num{100000} replications.}\end{figure}
\clearpage

\begin{table}
\resizebox{\textwidth}{!}{
\begin{tabular}{@{}S[table-figures-decimal=0,table-figures-integer=3]S[table-figures-decimal=0,table-figures-integer=3]S[table-figures-decimal=1,table-figures-integer=1]cSSScSSScSSS@{}}
\toprule
&&&&\multicolumn{3}{c}{i.i.d.}&&\multicolumn{3}{c}{AR(1)}&&\multicolumn{3}{c}{MA(1)}\\ \cmidrule{5-7} \cmidrule{9-11} \cmidrule{13-15}
{$n$} & {$T$} & {$\sqrt{\omega^4/\phi^4}$} & {} & {$\tUMP$} & {$\tUMPEMP$} & {$P_b$} & {} & {$\tUMP$} & {$\tUMPEMP$} & {$P_b$} & {} & {$\tUMP$} & {$\tUMPEMP$} & {$P_b$} \\
\midrule
25.0000 & 25.0000 & 0.6000 &  & 0.7962 & 3.8669 & 5.8818 &  & 4.6350 & 10.4648 & 9.4674 &  & 3.9021 & 10.8189 & 9.5658 \\
25.0000 & 50.0000 & 0.6000 &  & 1.3906 & 5.6856 & 6.6231 &  & 3.1102 & 8.2077 & 6.4387 &  & 4.1678 & 12.0357 & 9.3665 \\
25.0000 & 100.0000 & 0.6000 &  & 1.7551 & 6.5244 & 7.1257 &  & 3.5090 & 9.2923 & 6.4749 &  & 5.1133 & 13.7465 & 9.9041 \\
50.0000 & 50.0000 & 0.6000 &  & 1.6533 & 4.3996 & 4.7311 &  & 4.8251 & 8.2410 & 5.6254 &  & 6.7739 & 12.8545 & 8.4252 \\
50.0000 & 100.0000 & 0.6000 &  & 2.1212 & 5.0920 & 5.1464 &  & 4.2693 & 7.9506 & 4.8411 &  & 7.4065 & 14.0242 & 8.4378 \\
50.0000 & 200.0000 & 0.6000 &  & 2.4416 & 5.4761 & 5.3857 &  & 4.6366 & 8.5042 & 4.9744 &  & 6.3889 & 11.8987 & 7.2555 \\
100.0000 & 100.0000 & 0.6000 &  & 2.9127 & 4.9994 & 4.6021 &  & 5.4476 & 7.7822 & 4.7116 &  & 11.3178 & 16.5656 & 9.2628 \\
100.0000 & 200.0000 & 0.6000 &  & 3.1384 & 5.1736 & 4.8487 &  & 5.0329 & 7.4031 & 4.5095 &  & 8.4932 & 12.5348 & 7.3641 \\
100.0000 & 400.0000 & 0.6000 &  & 3.2521 & 5.2887 & 4.9561 &  & 5.7283 & 8.3145 & 4.9024 &  & 6.0023 & 8.8612 & 5.6634 \\
25.0000 & 25.0000 & 0.8000 &  & 0.9810 & 3.7338 & 5.2208 &  & 4.9076 & 9.8435 & 9.6164 &  & 4.1192 & 10.0258 & 9.4909 \\
25.0000 & 50.0000 & 0.8000 &  & 1.9054 & 5.6916 & 5.9999 &  & 2.8183 & 6.6715 & 6.0226 &  & 3.9827 & 10.0882 & 8.9819 \\
25.0000 & 100.0000 & 0.8000 &  & 2.5064 & 6.6163 & 6.5990 &  & 2.9499 & 7.0354 & 6.0180 &  & 4.7119 & 11.1431 & 9.5114 \\
50.0000 & 50.0000 & 0.8000 &  & 2.4165 & 4.9557 & 4.9910 &  & 4.5332 & 7.0930 & 6.5310 &  & 6.6538 & 11.4154 & 9.8946 \\
50.0000 & 100.0000 & 0.8000 &  & 3.0251 & 5.6452 & 5.4536 &  & 3.6292 & 6.1694 & 5.3056 &  & 6.8389 & 11.7291 & 9.6216 \\
50.0000 & 200.0000 & 0.8000 &  & 3.3481 & 5.9672 & 5.7878 &  & 3.7025 & 6.2602 & 5.3175 &  & 5.7277 & 9.7367 & 8.1180 \\
100.0000 & 100.0000 & 0.8000 &  & 3.5691 & 5.4105 & 4.9983 &  & 4.6138 & 6.3163 & 5.6847 &  & 10.1845 & 14.2208 & 11.5615 \\
100.0000 & 200.0000 & 0.8000 &  & 3.8384 & 5.6195 & 5.2732 &  & 3.9635 & 5.6051 & 4.9913 &  & 7.4150 & 10.4451 & 8.5821 \\
100.0000 & 400.0000 & 0.8000 &  & 3.8806 & 5.5888 & 5.3605 &  & 4.3856 & 6.1650 & 5.3575 &  & 5.2138 & 7.3246 & 6.3555 \\
25.0000 & 25.0000 & 1.0000 &  & 1.2397 & 3.9659 & 5.2410 &  & 5.0684 & 9.5982 & 10.2110 &  & 4.3615 & 9.8044 & 10.1019 \\
25.0000 & 50.0000 & 1.0000 &  & 2.3887 & 6.0448 & 6.0907 &  & 2.8032 & 6.2196 & 6.3289 &  & 4.1320 & 9.5884 & 9.4867 \\
25.0000 & 100.0000 & 1.0000 &  & 3.0960 & 6.9557 & 6.7638 &  & 2.8186 & 6.2345 & 6.1083 &  & 4.7735 & 10.4260 & 10.0881 \\
50.0000 & 50.0000 & 1.0000 &  & 2.8943 & 5.3350 & 5.3508 &  & 4.4731 & 6.7677 & 7.6859 &  & 6.5797 & 10.8650 & 11.4672 \\
50.0000 & 100.0000 & 1.0000 &  & 3.4405 & 5.8994 & 5.6885 &  & 3.3951 & 5.5735 & 5.8338 &  & 6.6612 & 10.9441 & 10.9341 \\
50.0000 & 200.0000 & 1.0000 &  & 3.7978 & 6.2282 & 6.0694 &  & 3.3643 & 5.5223 & 5.6364 &  & 5.5642 & 8.9847 & 8.8904 \\
100.0000 & 100.0000 & 1.0000 &  & 3.8792 & 5.6321 & 5.2958 &  & 4.3799 & 5.8997 & 6.6096 &  & 9.8836 & 13.5523 & 13.9209 \\
100.0000 & 200.0000 & 1.0000 &  & 4.0833 & 5.7237 & 5.5148 &  & 3.6800 & 5.1107 & 5.4257 &  & 7.1744 & 9.8918 & 9.9169 \\
100.0000 & 400.0000 & 1.0000 &  & 4.2268 & 5.8105 & 5.6839 &  & 4.0533 & 5.5600 & 5.7140 &  & 4.9708 & 6.7904 & 6.8131 \\
\midrule\multicolumn{3}{c}{Mean abs. dev. from 5\%}  &  &2.2601 &0.7588 &0.64551 &  &0.99337 &2.178 &1.2401 &  &1.6767 &6.1228 &4.2216\\
\bottomrule
\end{tabular}}
\caption{Sizes (in percent) of nominal 5\% level tests  with no heterogeneity in the alternatives. Based on \num{1000000} replications. Andrews Bandwidth, three factors.}
\label{table:sizesAnThree3_0.4_0.4_BN_}
\end{table}

\begin{table}
\resizebox{\textwidth}{!}{
\begin{tabular}{@{}S[table-figures-decimal=0,table-figures-integer=3]S[table-figures-decimal=0,table-figures-integer=3]S[table-figures-decimal=1,table-figures-integer=1]cSSScSSScSSS@{}}
\toprule
&&&&\multicolumn{3}{c}{i.i.d.}&&\multicolumn{3}{c}{AR(1)}&&\multicolumn{3}{c}{MA(1)}\\ \cmidrule{5-7} \cmidrule{9-11} \cmidrule{13-15}
{$n$} & {$T$} & {$\sqrt{\omega^4/\phi^4}$} & {} & {$\tUMP$} & {$\tUMPEMP$} & {$P_b$} & {} & {$\tUMP$} & {$\tUMPEMP$} & {$P_b$} & {} & {$\tUMP$} & {$\tUMPEMP$} & {$P_b$} \\
\midrule
25.0000 & 25.0000 & 0.6000 &  & 0.6519 & 2.9429 & 3.3228 &  & 1.9529 & 4.9187 & 4.5589 &  & 2.3252 & 7.1691 & 5.9148 \\
25.0000 & 50.0000 & 0.6000 &  & 1.3766 & 4.7112 & 4.2022 &  & 1.7975 & 5.0220 & 3.7168 &  & 3.1576 & 9.0570 & 6.4136 \\
25.0000 & 100.0000 & 0.6000 &  & 1.8412 & 5.5369 & 4.7090 &  & 2.3311 & 6.1002 & 4.1882 &  & 3.8906 & 10.1311 & 6.7980 \\
50.0000 & 50.0000 & 0.6000 &  & 1.9896 & 4.3363 & 3.7307 &  & 2.5779 & 4.6328 & 3.5677 &  & 5.2546 & 9.9898 & 6.7567 \\
50.0000 & 100.0000 & 0.6000 &  & 2.6096 & 5.1453 & 4.2821 &  & 2.8677 & 5.2536 & 3.7584 &  & 6.0728 & 10.9300 & 7.0383 \\
50.0000 & 200.0000 & 0.6000 &  & 2.9079 & 5.4449 & 4.5221 &  & 3.3588 & 5.9291 & 4.1363 &  & 5.2537 & 9.1803 & 6.0980 \\
100.0000 & 100.0000 & 0.6000 &  & 3.2333 & 4.9976 & 4.1734 &  & 3.3417 & 4.8772 & 3.8013 &  & 9.0868 & 13.1803 & 8.2036 \\
100.0000 & 200.0000 & 0.6000 &  & 3.5501 & 5.3180 & 4.4490 &  & 3.6401 & 5.2962 & 4.0341 &  & 6.9468 & 9.9451 & 6.6518 \\
100.0000 & 400.0000 & 0.6000 &  & 3.6748 & 5.4193 & 4.6036 &  & 4.3762 & 6.1913 & 4.5044 &  & 4.9468 & 7.0638 & 5.1842 \\
25.0000 & 25.0000 & 0.8000 &  & 0.8633 & 3.1380 & 3.5236 &  & 1.9518 & 4.5172 & 4.9306 &  & 2.4244 & 6.7597 & 6.4840 \\
25.0000 & 50.0000 & 0.8000 &  & 1.7525 & 5.0215 & 4.5947 &  & 1.7399 & 4.4680 & 4.0614 &  & 3.1456 & 8.2590 & 7.1735 \\
25.0000 & 100.0000 & 0.8000 &  & 2.3157 & 5.8541 & 5.2189 &  & 2.1850 & 5.2841 & 4.5616 &  & 3.8778 & 9.2635 & 7.7405 \\
50.0000 & 50.0000 & 0.8000 &  & 2.3259 & 4.5832 & 4.1542 &  & 2.4437 & 4.2243 & 4.2674 &  & 5.2027 & 9.4108 & 8.3227 \\
50.0000 & 100.0000 & 0.8000 &  & 2.9815 & 5.3636 & 4.8061 &  & 2.6466 & 4.6723 & 4.3354 &  & 5.8894 & 10.1058 & 8.5486 \\
50.0000 & 200.0000 & 0.8000 &  & 3.3378 & 5.7085 & 5.1843 &  & 3.0297 & 5.1846 & 4.6531 &  & 4.9821 & 8.3556 & 7.1979 \\
100.0000 & 100.0000 & 0.8000 &  & 3.5279 & 5.2200 & 4.6670 &  & 3.0675 & 4.4142 & 4.4180 &  & 8.7215 & 12.4005 & 10.4291 \\
100.0000 & 200.0000 & 0.8000 &  & 3.8300 & 5.4703 & 5.0033 &  & 3.3368 & 4.7260 & 4.5346 &  & 6.5951 & 9.2644 & 7.9484 \\
100.0000 & 400.0000 & 0.8000 &  & 3.9354 & 5.5291 & 5.1134 &  & 3.9441 & 5.4998 & 5.0133 &  & 4.6953 & 6.5273 & 5.8729 \\
25.0000 & 25.0000 & 1.0000 &  & 0.9758 & 3.2925 & 3.8269 &  & 1.9857 & 4.4192 & 5.5800 &  & 2.4779 & 6.6578 & 7.3304 \\
25.0000 & 50.0000 & 1.0000 &  & 1.9868 & 5.1818 & 5.0885 &  & 1.7063 & 4.2346 & 4.4774 &  & 3.2573 & 8.1076 & 8.2202 \\
25.0000 & 100.0000 & 1.0000 &  & 2.6054 & 6.0263 & 5.8432 &  & 2.1568 & 5.0555 & 5.0845 &  & 3.9255 & 9.0100 & 8.8553 \\
50.0000 & 50.0000 & 1.0000 &  & 2.4841 & 4.6933 & 4.5562 &  & 2.4197 & 4.1105 & 5.0323 &  & 5.1279 & 9.1164 & 9.9805 \\
50.0000 & 100.0000 & 1.0000 &  & 3.1490 & 5.4369 & 5.2120 &  & 2.5772 & 4.4302 & 4.7827 &  & 5.8110 & 9.8907 & 9.9904 \\
50.0000 & 200.0000 & 1.0000 &  & 3.4869 & 5.8121 & 5.6425 &  & 2.9910 & 4.9786 & 5.1586 &  & 4.9494 & 8.1371 & 8.1176 \\
100.0000 & 100.0000 & 1.0000 &  & 3.6000 & 5.2589 & 4.9303 &  & 2.9881 & 4.2744 & 4.9690 &  & 8.5781 & 12.0980 & 12.6370 \\
100.0000 & 200.0000 & 1.0000 &  & 3.8518 & 5.4593 & 5.2434 &  & 3.2298 & 4.5700 & 4.9094 &  & 6.4189 & 8.9621 & 9.0305 \\
100.0000 & 400.0000 & 1.0000 &  & 4.0782 & 5.6390 & 5.5125 &  & 3.7884 & 5.2570 & 5.4063 &  & 4.5581 & 6.3042 & 6.3770 \\
\midrule\multicolumn{3}{c}{Mean abs. dev. from 5\%}  &  &2.2991 &0.59818 &0.55586 &  &2.2433 &0.46686 &0.55957 &  &1.3832 &4.0843 &2.7524\\
\bottomrule
\end{tabular}}
\caption{Sizes (in percent) of nominal 5\% level tests  with no heterogeneity in the alternatives. Based on \num{1000000} replications. Andrews Bandwidth, $t$-distribution with five degrees of freedom.}
\label{table:sizesAnT1_0.4_0.4_BN_}
\end{table}

\clearpage

\subsection*{Finite-Sample Results with the \cite{NeweyWest1994} Bandwidth}
\label{appendix:Newey}

\begin{figure}[ht]\footnotesize\begin{subfigure}[b]{0.36\textwidth}\input{graphs/stamps/NpowersMPDifNe25_25_1_0.8_iid_0.4_iid_0.4_MP_DifMP.tex}\end{subfigure}\begin{subfigure}[b]{0.32\textwidth}\input{graphs/stamps/NpowersMPDifNe50_25_1_0.8_iid_0.4_iid_0.4_MP_DifMP.tex}\end{subfigure}\begin{subfigure}[b]{0.32\textwidth}\input{graphs/stamps/NpowersMPDifNe100_25_1_0.8_iid_0.4_iid_0.4_MP_DifMP.tex}\end{subfigure}\\ \begin{subfigure}[b]{0.36\textwidth}\input{graphs/stamps/NpowersMPDifNe50_50_1_0.8_iid_0.4_iid_0.4_MP_DifMP.tex}\end{subfigure}\begin{subfigure}[b]{0.32\textwidth}\input{graphs/stamps/NpowersMPDifNe100_50_1_0.8_iid_0.4_iid_0.4_MP_DifMP.tex}\end{subfigure}\begin{subfigure}[b]{0.32\textwidth}\input{graphs/stamps/NpowersMPDifNe200_50_1_0.8_iid_0.4_iid_0.4_MP_DifMP.tex}\end{subfigure}\\ \begin{subfigure}[b]{0.36\textwidth}\input{graphs/stamps/NpowersMPDifNe100_100_1_0.8_iid_0.4_iid_0.4_MP_DifMP.tex}\end{subfigure}\begin{subfigure}[b]{0.32\textwidth}\input{graphs/stamps/NpowersMPDifNe200_100_1_0.8_iid_0.4_iid_0.4_MP_DifMP.tex}\end{subfigure}\begin{subfigure}[b]{0.32\textwidth}\input{graphs/stamps/NpowersMPDifNe400_100_1_0.8_iid_0.4_iid_0.4_MP_DifMP.tex}\end{subfigure}\\ \begin{subfigure}[b]{\textwidth}\centering\input{graphs/stamps/lNpowersMPDifNe400_100_1_0.8_iid_0.4_iid_0.4_MP_DifMP.tex}\end{subfigure}\normalsize\caption{\label{DifMPNpowersMPDifNeiidBN0.811}Difference between powers in the MP vs the PANIC framework as a function of $-h$ with i.i.d.\ factor innovations and i.i.d.\ idiosyncratic parts and $\sqrt{\omega^4/\phi^4}=0.8$. Based on \num{1000000} replications.}\end{figure}
\begin{figure}[ht]\footnotesize\begin{subfigure}[b]{0.3466\textwidth}\input{graphs/stamps/NpowersNe25_25_1_0.8_iid_0.4_iid_0.4_BN_.tex}\end{subfigure}\begin{subfigure}[b]{0.3267\textwidth}\input{graphs/stamps/NpowersNe50_25_1_0.8_iid_0.4_iid_0.4_BN_.tex}\end{subfigure}\begin{subfigure}[b]{0.3267\textwidth}\input{graphs/stamps/NpowersNe100_25_1_0.8_iid_0.4_iid_0.4_BN_.tex}\end{subfigure}\\ \begin{subfigure}[b]{0.3466\textwidth}\input{graphs/stamps/NpowersNe50_50_1_0.8_iid_0.4_iid_0.4_BN_.tex}\end{subfigure}\begin{subfigure}[b]{0.3267\textwidth}\input{graphs/stamps/NpowersNe100_50_1_0.8_iid_0.4_iid_0.4_BN_.tex}\end{subfigure}\begin{subfigure}[b]{0.3267\textwidth}\input{graphs/stamps/NpowersNe200_50_1_0.8_iid_0.4_iid_0.4_BN_.tex}\end{subfigure}\\ \begin{subfigure}[b]{0.3466\textwidth}\input{graphs/stamps/NpowersNe100_100_1_0.8_iid_0.4_iid_0.4_BN_.tex}\end{subfigure}\begin{subfigure}[b]{0.3267\textwidth}\input{graphs/stamps/NpowersNe200_100_1_0.8_iid_0.4_iid_0.4_BN_.tex}\end{subfigure}\begin{subfigure}[b]{0.3267\textwidth}\input{graphs/stamps/NpowersNe400_100_1_0.8_iid_0.4_iid_0.4_BN_.tex}\end{subfigure}\\ \begin{subfigure}[b]{\textwidth}\centering\input{graphs/stamps/lNpowersNe400_100_1_0.8_iid_0.4_iid_0.4_BN_.tex}\end{subfigure}\normalsize\caption{\label{NpowersNeiidBN0.811}Size-corrected power of unit-root tests as a function of $-h$ for varying sample sizes in the PANIC framework with i.i.d.\ factor innovations and i.i.d.\ idiosyncratic parts and $\sqrt{\omega^4/\phi^4}=0.8$. Based on \num{100000} replications.}\end{figure}

\begin{figure}[ht]\footnotesize\begin{subfigure}[b]{0.3466\textwidth}\input{graphs/stamps/NpowersNe25_25_1_1_iid_0.4_iid_0.4_BN_Dif.tex}\end{subfigure}\begin{subfigure}[b]{0.3267\textwidth}\input{graphs/stamps/NpowersNe50_25_1_1_iid_0.4_iid_0.4_BN_Dif.tex}\end{subfigure}\begin{subfigure}[b]{0.3267\textwidth}\input{graphs/stamps/NpowersNe100_25_1_1_iid_0.4_iid_0.4_BN_Dif.tex}\end{subfigure}\\ \begin{subfigure}[b]{0.3466\textwidth}\input{graphs/stamps/NpowersNe50_50_1_1_iid_0.4_iid_0.4_BN_Dif.tex}\end{subfigure}\begin{subfigure}[b]{0.3267\textwidth}\input{graphs/stamps/NpowersNe100_50_1_1_iid_0.4_iid_0.4_BN_Dif.tex}\end{subfigure}\begin{subfigure}[b]{0.3267\textwidth}\input{graphs/stamps/NpowersNe200_50_1_1_iid_0.4_iid_0.4_BN_Dif.tex}\end{subfigure}\\ \begin{subfigure}[b]{0.3466\textwidth}\input{graphs/stamps/NpowersNe100_100_1_1_iid_0.4_iid_0.4_BN_Dif.tex}\end{subfigure}\begin{subfigure}[b]{0.3267\textwidth}\input{graphs/stamps/NpowersNe200_100_1_1_iid_0.4_iid_0.4_BN_Dif.tex}\end{subfigure}\begin{subfigure}[b]{0.3267\textwidth}\input{graphs/stamps/NpowersNe400_100_1_1_iid_0.4_iid_0.4_BN_Dif.tex}\end{subfigure}\\ \begin{subfigure}[b]{\textwidth}\centering\input{graphs/stamps/lNpowersNe400_100_1_1_iid_0.4_iid_0.4_BN_Dif.tex}\end{subfigure}\normalsize\caption{\label{DifNpowersNeiidBN111}(Size-corrected) power gains from using $\tUMPEMP$ over $P_b$ for varying values of $\sqrt{\omega^4/\phi^4}$ and sample sizes in the PANIC framework with i.i.d.\ factor innovations and i.i.d.\ idiosyncratic parts. Based on \num{0} replications.}\end{figure}
\begin{figure}[ht]\footnotesize\begin{subfigure}[b]{0.3466\textwidth}\input{graphs/stamps/NpowersCorrNe25_25_1_0.8_MA_0.4_MA_0.4_BN_.tex}\end{subfigure}\begin{subfigure}[b]{0.3267\textwidth}\input{graphs/stamps/NpowersCorrNe50_25_1_0.8_MA_0.4_MA_0.4_BN_.tex}\end{subfigure}\begin{subfigure}[b]{0.3267\textwidth}\input{graphs/stamps/NpowersCorrNe100_25_1_0.8_MA_0.4_MA_0.4_BN_.tex}\end{subfigure}\\ \begin{subfigure}[b]{0.3466\textwidth}\input{graphs/stamps/NpowersCorrNe50_50_1_0.8_MA_0.4_MA_0.4_BN_.tex}\end{subfigure}\begin{subfigure}[b]{0.3267\textwidth}\input{graphs/stamps/NpowersCorrNe100_50_1_0.8_MA_0.4_MA_0.4_BN_.tex}\end{subfigure}\begin{subfigure}[b]{0.3267\textwidth}\input{graphs/stamps/NpowersCorrNe200_50_1_0.8_MA_0.4_MA_0.4_BN_.tex}\end{subfigure}\\ \begin{subfigure}[b]{0.3466\textwidth}\input{graphs/stamps/NpowersCorrNe100_100_1_0.8_MA_0.4_MA_0.4_BN_.tex}\end{subfigure}\begin{subfigure}[b]{0.3267\textwidth}\input{graphs/stamps/NpowersCorrNe200_100_1_0.8_MA_0.4_MA_0.4_BN_.tex}\end{subfigure}\begin{subfigure}[b]{0.3267\textwidth}\input{graphs/stamps/NpowersCorrNe400_100_1_0.8_MA_0.4_MA_0.4_BN_.tex}\end{subfigure}\\ \begin{subfigure}[b]{\textwidth}\centering\input{graphs/stamps/lNpowersCorrNe400_100_1_0.8_MA_0.4_MA_0.4_BN_.tex}\end{subfigure}\normalsize\caption{\label{NpowersCorrNeMABN0.811}Size-corrected power of unit-root tests as a function of $-h$ for varying sample sizes in the PANIC framework with MA factor innovations and MA idiosyncratic parts and $\sqrt{\omega^4/\phi^4}=0.8$. Based on \num{100000} replications.}\end{figure}
\begin{figure}[ht]\footnotesize\begin{subfigure}[b]{0.3466\textwidth}\input{graphs/stamps/NpowersCorrNe25_25_1_0.8_AR_0.4_AR_0.4_BN_.tex}\end{subfigure}\begin{subfigure}[b]{0.3267\textwidth}\input{graphs/stamps/NpowersCorrNe50_25_1_0.8_AR_0.4_AR_0.4_BN_.tex}\end{subfigure}\begin{subfigure}[b]{0.3267\textwidth}\input{graphs/stamps/NpowersCorrNe100_25_1_0.8_AR_0.4_AR_0.4_BN_.tex}\end{subfigure}\\ \begin{subfigure}[b]{0.3466\textwidth}\input{graphs/stamps/NpowersCorrNe50_50_1_0.8_AR_0.4_AR_0.4_BN_.tex}\end{subfigure}\begin{subfigure}[b]{0.3267\textwidth}\input{graphs/stamps/NpowersCorrNe100_50_1_0.8_AR_0.4_AR_0.4_BN_.tex}\end{subfigure}\begin{subfigure}[b]{0.3267\textwidth}\input{graphs/stamps/NpowersCorrNe200_50_1_0.8_AR_0.4_AR_0.4_BN_.tex}\end{subfigure}\\ \begin{subfigure}[b]{0.3466\textwidth}\input{graphs/stamps/NpowersCorrNe100_100_1_0.8_AR_0.4_AR_0.4_BN_.tex}\end{subfigure}\begin{subfigure}[b]{0.3267\textwidth}\input{graphs/stamps/NpowersCorrNe200_100_1_0.8_AR_0.4_AR_0.4_BN_.tex}\end{subfigure}\begin{subfigure}[b]{0.3267\textwidth}\input{graphs/stamps/NpowersCorrNe400_100_1_0.8_AR_0.4_AR_0.4_BN_.tex}\end{subfigure}\\ \begin{subfigure}[b]{\textwidth}\centering\input{graphs/stamps/lNpowersCorrNe400_100_1_0.8_AR_0.4_AR_0.4_BN_.tex}\end{subfigure}\normalsize\caption{\label{NpowersCorrNeARBN0.811}Size-corrected power of unit-root tests as a function of $-h$ for varying sample sizes in the PANIC framework with AR factor innovations and AR idiosyncratic parts and $\sqrt{\omega^4/\phi^4}=0.8$. Based on \num{100000} replications.}\end{figure}
\begin{figure}[ht]\footnotesize\begin{subfigure}[b]{0.3466\textwidth}\input{graphs/stamps/NpowersStatNe25_25_1_0.8_iid_0.4_overdiffiid_0.4_BN_.tex}\end{subfigure}\begin{subfigure}[b]{0.3267\textwidth}\input{graphs/stamps/NpowersStatNe50_25_1_0.8_iid_0.4_overdiffiid_0.4_BN_.tex}\end{subfigure}\begin{subfigure}[b]{0.3267\textwidth}\input{graphs/stamps/NpowersStatNe100_25_1_0.8_iid_0.4_overdiffiid_0.4_BN_.tex}\end{subfigure}\\ \begin{subfigure}[b]{0.3466\textwidth}\input{graphs/stamps/NpowersStatNe50_50_1_0.8_iid_0.4_overdiffiid_0.4_BN_.tex}\end{subfigure}\begin{subfigure}[b]{0.3267\textwidth}\input{graphs/stamps/NpowersStatNe100_50_1_0.8_iid_0.4_overdiffiid_0.4_BN_.tex}\end{subfigure}\begin{subfigure}[b]{0.3267\textwidth}\input{graphs/stamps/NpowersStatNe200_50_1_0.8_iid_0.4_overdiffiid_0.4_BN_.tex}\end{subfigure}\\ \begin{subfigure}[b]{0.3466\textwidth}\input{graphs/stamps/NpowersStatNe100_100_1_0.8_iid_0.4_overdiffiid_0.4_BN_.tex}\end{subfigure}\begin{subfigure}[b]{0.3267\textwidth}\input{graphs/stamps/NpowersStatNe200_100_1_0.8_iid_0.4_overdiffiid_0.4_BN_.tex}\end{subfigure}\begin{subfigure}[b]{0.3267\textwidth}\input{graphs/stamps/NpowersStatNe400_100_1_0.8_iid_0.4_overdiffiid_0.4_BN_.tex}\end{subfigure}\\ \begin{subfigure}[b]{\textwidth}\centering\input{graphs/stamps/lNpowersStatNe400_100_1_0.8_iid_0.4_overdiffiid_0.4_BN_.tex}\end{subfigure}\normalsize\caption{\label{NpowersStatNeiidBN0.811}Size-corrected power of unit-root tests as a function of $-h$ for varying sample sizes in the PANIC framework with overdifferenced i.i.d.\ factor innovations and i.i.d.\ idiosyncratic parts and $\sqrt{\omega^4/\phi^4}=0.8$. Based on \num{100000} replications.}\end{figure}

\begin{figure}[ht]\footnotesize\begin{subfigure}[b]{0.3466\textwidth}\input{graphs/stamps/NpowersThreeNe25_25_3_0.8_iid_0.4_iid_0.4_BN_.tex}\end{subfigure}\begin{subfigure}[b]{0.3267\textwidth}\input{graphs/stamps/NpowersThreeNe50_25_3_0.8_iid_0.4_iid_0.4_BN_.tex}\end{subfigure}\begin{subfigure}[b]{0.3267\textwidth}\input{graphs/stamps/NpowersThreeNe100_25_3_0.8_iid_0.4_iid_0.4_BN_.tex}\end{subfigure}\\ \begin{subfigure}[b]{0.3466\textwidth}\input{graphs/stamps/NpowersThreeNe50_50_3_0.8_iid_0.4_iid_0.4_BN_.tex}\end{subfigure}\begin{subfigure}[b]{0.3267\textwidth}\input{graphs/stamps/NpowersThreeNe100_50_3_0.8_iid_0.4_iid_0.4_BN_.tex}\end{subfigure}\begin{subfigure}[b]{0.3267\textwidth}\input{graphs/stamps/NpowersThreeNe200_50_3_0.8_iid_0.4_iid_0.4_BN_.tex}\end{subfigure}\\ \begin{subfigure}[b]{0.3466\textwidth}\input{graphs/stamps/NpowersThreeNe100_100_3_0.8_iid_0.4_iid_0.4_BN_.tex}\end{subfigure}\begin{subfigure}[b]{0.3267\textwidth}\input{graphs/stamps/NpowersThreeNe200_100_3_0.8_iid_0.4_iid_0.4_BN_.tex}\end{subfigure}\begin{subfigure}[b]{0.3267\textwidth}\input{graphs/stamps/NpowersThreeNe400_100_3_0.8_iid_0.4_iid_0.4_BN_.tex}\end{subfigure}\\ \begin{subfigure}[b]{\textwidth}\centering\input{graphs/stamps/lNpowersThreeNe400_100_3_0.8_iid_0.4_iid_0.4_BN_.tex}\end{subfigure}\normalsize\caption{\label{NpowersThreeNeiidBN0.831}Size-corrected power of unit-root tests as a function of $-h$ for varying sample sizes in the PANIC framework with i.i.d.\ factor innovations and i.i.d.\ idiosyncratic parts and $\sqrt{\omega^4/\phi^4}=0.8$. Dependence based on three factors. Based on \num{100000} replications.}\end{figure}
\begin{figure}[ht]\footnotesize\begin{subfigure}[b]{0.3466\textwidth}\input{graphs/stamps/NpowersHetNe25_25_1_0.8_iid_0.4_iid_0.4_BN__Uniform.tex}\end{subfigure}\begin{subfigure}[b]{0.3267\textwidth}\input{graphs/stamps/NpowersHetNe50_25_1_0.8_iid_0.4_iid_0.4_BN__Uniform.tex}\end{subfigure}\begin{subfigure}[b]{0.3267\textwidth}\input{graphs/stamps/NpowersHetNe100_25_1_0.8_iid_0.4_iid_0.4_BN__Uniform.tex}\end{subfigure}\\ \begin{subfigure}[b]{0.3466\textwidth}\input{graphs/stamps/NpowersHetNe50_50_1_0.8_iid_0.4_iid_0.4_BN__Uniform.tex}\end{subfigure}\begin{subfigure}[b]{0.3267\textwidth}\input{graphs/stamps/NpowersHetNe100_50_1_0.8_iid_0.4_iid_0.4_BN__Uniform.tex}\end{subfigure}\begin{subfigure}[b]{0.3267\textwidth}\input{graphs/stamps/NpowersHetNe200_50_1_0.8_iid_0.4_iid_0.4_BN__Uniform.tex}\end{subfigure}\\ \begin{subfigure}[b]{0.3466\textwidth}\input{graphs/stamps/NpowersHetNe100_100_1_0.8_iid_0.4_iid_0.4_BN__Uniform.tex}\end{subfigure}\begin{subfigure}[b]{0.3267\textwidth}\input{graphs/stamps/NpowersHetNe200_100_1_0.8_iid_0.4_iid_0.4_BN__Uniform.tex}\end{subfigure}\begin{subfigure}[b]{0.3267\textwidth}\input{graphs/stamps/NpowersHetNe400_100_1_0.8_iid_0.4_iid_0.4_BN__Uniform.tex}\end{subfigure}\\ \begin{subfigure}[b]{\textwidth}\centering\input{graphs/stamps/lNpowersHetNe400_100_1_0.8_iid_0.4_iid_0.4_BN__Uniform.tex}\end{subfigure}\normalsize\caption{\label{NpowersHetNeiidBN0.810}Size-corrected power of unit-root tests as a function of $-h$ for varying sample sizes in the PANIC framework with i.i.d.\ factor innovations and i.i.d.\ idiosyncratic parts and $\sqrt{\omega^4/\phi^4}=0.8$. Alternatives drawn from a Uniform(0.2,1.8) distribution. Based on \num{100000} replications.}\end{figure}
\begin{figure}[ht]\footnotesize\begin{subfigure}[b]{0.3466\textwidth}\input{graphs/stamps/NpowersTNe25_25_1_0.8_iid_0.4_iid_0.4_BN_.tex}\end{subfigure}\begin{subfigure}[b]{0.3267\textwidth}\input{graphs/stamps/NpowersTNe50_25_1_0.8_iid_0.4_iid_0.4_BN_.tex}\end{subfigure}\begin{subfigure}[b]{0.3267\textwidth}\input{graphs/stamps/NpowersTNe100_25_1_0.8_iid_0.4_iid_0.4_BN_.tex}\end{subfigure}\\ \begin{subfigure}[b]{0.3466\textwidth}\input{graphs/stamps/NpowersTNe50_50_1_0.8_iid_0.4_iid_0.4_BN_.tex}\end{subfigure}\begin{subfigure}[b]{0.3267\textwidth}\input{graphs/stamps/NpowersTNe100_50_1_0.8_iid_0.4_iid_0.4_BN_.tex}\end{subfigure}\begin{subfigure}[b]{0.3267\textwidth}\input{graphs/stamps/NpowersTNe200_50_1_0.8_iid_0.4_iid_0.4_BN_.tex}\end{subfigure}\\ \begin{subfigure}[b]{0.3466\textwidth}\input{graphs/stamps/NpowersTNe100_100_1_0.8_iid_0.4_iid_0.4_BN_.tex}\end{subfigure}\begin{subfigure}[b]{0.3267\textwidth}\input{graphs/stamps/NpowersTNe200_100_1_0.8_iid_0.4_iid_0.4_BN_.tex}\end{subfigure}\begin{subfigure}[b]{0.3267\textwidth}\input{graphs/stamps/NpowersTNe400_100_1_0.8_iid_0.4_iid_0.4_BN_.tex}\end{subfigure}\\ \begin{subfigure}[b]{\textwidth}\centering\input{graphs/stamps/lNpowersTNe400_100_1_0.8_iid_0.4_iid_0.4_BN_.tex}\end{subfigure}\normalsize\caption{\label{NpowersTNeiidBN0.811}Size-corrected power of unit-root tests as a function of $-h$ for varying sample sizes in the PANIC framework with i.i.d.\ factor innovations and i.i.d.\ idiosyncratic parts and $\sqrt{\omega^4/\phi^4}=0.8$. Based on \num{100000} replications.}\end{figure}

\clearpage
\begin{table}
\resizebox{0.95\textwidth}{!}{
\begin{tabular}{@{}S[table-figures-decimal=0,table-figures-integer=3]S[table-figures-decimal=0,table-figures-integer=3]S[table-figures-decimal=1,table-figures-integer=1]cSSScSSScSSS@{}}
\toprule
&&&&\multicolumn{3}{c}{i.i.d.}&&\multicolumn{3}{c}{AR(1)}&&\multicolumn{3}{c}{MA(1)}\\ \cmidrule{5-7} \cmidrule{9-11} \cmidrule{13-15}
{$n$} & {$T$} & {$\sqrt{\omega^4/\phi^4}$} & {} & {$\tUMP$} & {$\tUMPEMP$} & {$P_b$} & {} & {$\tUMP$} & {$\tUMPEMP$} & {$P_b$} & {} & {$\tUMP$} & {$\tUMPEMP$} & {$P_b$} \\
\midrule
25.0000 & 25.0000 & 0.6000 &  & 0.3388 & 1.1623 & 1.5457 &  & 1.3131 & 3.4462 & 3.6370 &  & 1.0063 & 3.2835 & 3.6236 \\
25.0000 & 50.0000 & 0.6000 &  & 0.6440 & 2.4586 & 2.3327 &  & 1.4229 & 4.1603 & 3.0997 &  & 1.2706 & 4.1106 & 3.3459 \\
25.0000 & 100.0000 & 0.6000 &  & 1.3066 & 4.2087 & 3.5685 &  & 2.2698 & 6.0257 & 3.9663 &  & 2.2416 & 6.2762 & 4.4190 \\
50.0000 & 50.0000 & 0.6000 &  & 0.8959 & 2.0875 & 1.8857 &  & 2.1089 & 3.8543 & 2.9740 &  & 1.8969 & 3.8332 & 3.1266 \\
50.0000 & 100.0000 & 0.6000 &  & 1.8630 & 3.8484 & 3.0815 &  & 2.8624 & 5.2761 & 3.5671 &  & 3.1365 & 5.9815 & 4.1554 \\
50.0000 & 200.0000 & 0.6000 &  & 2.3599 & 4.5786 & 3.6983 &  & 3.3664 & 5.9994 & 3.8594 &  & 2.8088 & 5.1477 & 3.6055 \\
100.0000 & 100.0000 & 0.6000 &  & 2.2961 & 3.6674 & 2.8402 &  & 3.4368 & 5.0565 & 3.6375 &  & 4.0777 & 6.1400 & 4.3134 \\
100.0000 & 200.0000 & 0.6000 &  & 2.9105 & 4.4274 & 3.4651 &  & 3.7872 & 5.5406 & 3.7939 &  & 3.1585 & 4.6620 & 3.4292 \\
100.0000 & 400.0000 & 0.6000 &  & 3.2298 & 4.8187 & 3.8973 &  & 4.2232 & 6.0455 & 4.1226 &  & 3.1457 & 4.5824 & 3.4705 \\
25.0000 & 25.0000 & 0.8000 &  & 0.4225 & 1.3391 & 1.6874 &  & 1.3718 & 3.2440 & 4.1068 &  & 1.1106 & 3.1819 & 4.0738 \\
25.0000 & 50.0000 & 0.8000 &  & 0.8897 & 2.8047 & 2.6400 &  & 1.3979 & 3.7231 & 3.4075 &  & 1.3729 & 3.8814 & 3.7385 \\
25.0000 & 100.0000 & 0.8000 &  & 1.7360 & 4.6057 & 4.0250 &  & 2.1290 & 5.2631 & 4.3875 &  & 2.3297 & 5.8709 & 5.0136 \\
50.0000 & 50.0000 & 0.8000 &  & 1.1637 & 2.4323 & 2.0917 &  & 2.0337 & 3.5635 & 3.5755 &  & 1.9149 & 3.6634 & 3.7193 \\
50.0000 & 100.0000 & 0.8000 &  & 2.2364 & 4.1863 & 3.4217 &  & 2.6444 & 4.6693 & 4.0968 &  & 3.1080 & 5.5761 & 4.8208 \\
50.0000 & 200.0000 & 0.8000 &  & 2.7667 & 4.8961 & 4.1759 &  & 3.0927 & 5.2730 & 4.3847 &  & 2.7425 & 4.7416 & 4.0211 \\
100.0000 & 100.0000 & 0.8000 &  & 2.5527 & 3.8837 & 2.9999 &  & 3.2116 & 4.5998 & 4.1842 &  & 3.9520 & 5.8031 & 5.0632 \\
100.0000 & 200.0000 & 0.8000 &  & 3.1914 & 4.6460 & 3.7717 &  & 3.4551 & 4.9476 & 4.1871 &  & 2.9933 & 4.3505 & 3.6884 \\
100.0000 & 400.0000 & 0.8000 &  & 3.5457 & 5.0434 & 4.2708 &  & 3.8760 & 5.4350 & 4.5832 &  & 3.0457 & 4.3041 & 3.7559 \\
25.0000 & 25.0000 & 1.0000 &  & 0.4986 & 1.4582 & 1.8642 &  & 1.4347 & 3.2722 & 4.7637 &  & 1.1320 & 3.1501 & 4.5363 \\
25.0000 & 50.0000 & 1.0000 &  & 1.0655 & 3.0220 & 2.9235 &  & 1.4236 & 3.6271 & 3.9000 &  & 1.4269 & 3.8729 & 4.2192 \\
25.0000 & 100.0000 & 1.0000 &  & 1.9786 & 4.8377 & 4.4509 &  & 2.1016 & 4.9973 & 4.8947 &  & 2.3609 & 5.6851 & 5.5877 \\
50.0000 & 50.0000 & 1.0000 &  & 1.2592 & 2.5226 & 2.1738 &  & 2.0251 & 3.4968 & 4.2043 &  & 1.9637 & 3.6463 & 4.3561 \\
50.0000 & 100.0000 & 1.0000 &  & 2.3839 & 4.2495 & 3.6552 &  & 2.5926 & 4.4860 & 4.5672 &  & 3.1095 & 5.4795 & 5.4354 \\
50.0000 & 200.0000 & 1.0000 &  & 2.9174 & 4.9746 & 4.4111 &  & 2.9749 & 5.0344 & 4.7947 &  & 2.7566 & 4.6517 & 4.3721 \\
100.0000 & 100.0000 & 1.0000 &  & 2.6899 & 4.0159 & 3.0834 &  & 3.0944 & 4.4236 & 4.7045 &  & 3.9036 & 5.6798 & 5.7016 \\
100.0000 & 200.0000 & 1.0000 &  & 3.3241 & 4.7578 & 3.9273 &  & 3.3702 & 4.7801 & 4.5121 &  & 2.9883 & 4.2657 & 3.9064 \\
100.0000 & 400.0000 & 1.0000 &  & 3.6528 & 5.1201 & 4.5369 &  & 3.7611 & 5.2540 & 4.8960 &  & 2.9627 & 4.1590 & 3.8562 \\
\midrule\multicolumn{3}{c}{Mean abs. dev. from 5\%}  &  &2.9956 &1.3064 &1.7991 &  &2.3785 &0.73749 &0.896 &  &2.4846 &0.89999 &0.93512\\
\bottomrule
\end{tabular}}
\caption{Sizes (in percent) of nominal 5\% level tests  with no heterogeneity in the alternatives. Based on \num{1000000} replications. Newey Bandwidth.}
\label{table:sizesNe1_0.4_0.4_BN_}
\end{table}

\begin{table}
\resizebox{\textwidth}{!}{
\begin{tabular}{@{}S[table-figures-decimal=0,table-figures-integer=3]S[table-figures-decimal=0,table-figures-integer=3]S[table-figures-decimal=1,table-figures-integer=1]cSSScSSScSSS@{}}
\toprule
&&&&\multicolumn{3}{c}{i.i.d.}&&\multicolumn{3}{c}{AR(1)}&&\multicolumn{3}{c}{MA(1)}\\ \cmidrule{5-7} \cmidrule{9-11} \cmidrule{13-15}
{$n$} & {$T$} & {$\sqrt{\omega^4/\phi^4}$} & {} & {$\tUMP$} & {$\tUMPEMP$} & {$P_b$} & {} & {$\tUMP$} & {$\tUMPEMP$} & {$P_b$} & {} & {$\tUMP$} & {$\tUMPEMP$} & {$P_b$} \\
\midrule
25.0000 & 25.0000 & 0.6000 &  & 0.4962 & 1.5444 & 3.3224 &  & 2.9578 & 7.1775 & 7.9194 &  & 1.6268 & 4.9505 & 6.3470 \\
25.0000 & 50.0000 & 0.6000 &  & 0.5508 & 2.6096 & 4.1764 &  & 2.3702 & 6.7117 & 5.5572 &  & 1.6001 & 5.6003 & 5.5287 \\
25.0000 & 100.0000 & 0.6000 &  & 1.1492 & 4.7448 & 5.7383 &  & 3.3728 & 9.0958 & 6.2475 &  & 2.8726 & 8.6625 & 6.8781 \\
50.0000 & 50.0000 & 0.6000 &  & 0.6023 & 1.7404 & 2.5553 &  & 3.8536 & 6.8190 & 4.8049 &  & 2.3722 & 5.0253 & 4.3409 \\
50.0000 & 100.0000 & 0.6000 &  & 1.3413 & 3.5093 & 3.8275 &  & 4.1974 & 7.9236 & 4.6344 &  & 3.6692 & 7.5570 & 5.1173 \\
50.0000 & 200.0000 & 0.6000 &  & 1.8385 & 4.3726 & 4.4307 &  & 4.6179 & 8.5630 & 4.7291 &  & 3.2509 & 6.5819 & 4.3732 \\
100.0000 & 100.0000 & 0.6000 &  & 1.8665 & 3.3818 & 3.1584 &  & 5.5232 & 7.9254 & 4.5001 &  & 5.0200 & 7.7833 & 4.9222 \\
100.0000 & 200.0000 & 0.6000 &  & 2.3987 & 4.1007 & 3.7974 &  & 5.2103 & 7.6945 & 4.2308 &  & 3.6942 & 5.7748 & 3.7743 \\
100.0000 & 400.0000 & 0.6000 &  & 2.7867 & 4.6380 & 4.2494 &  & 5.6598 & 8.2504 & 4.5355 &  & 3.7122 & 5.6787 & 3.8055 \\
25.0000 & 25.0000 & 0.8000 &  & 0.5455 & 1.5147 & 2.7888 &  & 3.2335 & 6.8826 & 8.1472 &  & 1.7727 & 4.6477 & 6.2022 \\
25.0000 & 50.0000 & 0.8000 &  & 0.8098 & 2.8020 & 3.6265 &  & 2.1924 & 5.4822 & 5.1573 &  & 1.6382 & 4.7384 & 5.0447 \\
25.0000 & 100.0000 & 0.8000 &  & 1.7726 & 5.0763 & 5.1882 &  & 2.8438 & 6.9046 & 5.8022 &  & 2.7530 & 7.0275 & 6.3298 \\
50.0000 & 50.0000 & 0.8000 &  & 1.0307 & 2.2782 & 2.5762 &  & 3.7059 & 5.9373 & 5.5867 &  & 2.5361 & 4.6978 & 4.8684 \\
50.0000 & 100.0000 & 0.8000 &  & 2.1445 & 4.2225 & 3.9575 &  & 3.6027 & 6.1652 & 5.0887 &  & 3.5613 & 6.4271 & 5.5642 \\
50.0000 & 200.0000 & 0.8000 &  & 2.7145 & 5.0159 & 4.6926 &  & 3.7158 & 6.3520 & 5.0357 &  & 3.0855 & 5.4563 & 4.6310 \\
100.0000 & 100.0000 & 0.8000 &  & 2.5199 & 3.9437 & 3.3039 &  & 4.7029 & 6.4571 & 5.4348 &  & 4.7398 & 6.8616 & 5.7698 \\
100.0000 & 200.0000 & 0.8000 &  & 3.1329 & 4.7019 & 4.0220 &  & 4.1174 & 5.8601 & 4.6602 &  & 3.3139 & 4.8555 & 4.0193 \\
100.0000 & 400.0000 & 0.8000 &  & 3.4686 & 5.0439 & 4.5277 &  & 4.3461 & 6.1262 & 4.8763 &  & 3.3071 & 4.7319 & 4.0322 \\
25.0000 & 25.0000 & 1.0000 &  & 0.6544 & 1.6998 & 2.7389 &  & 3.3817 & 6.8133 & 8.6953 &  & 1.9387 & 4.6574 & 6.5738 \\
25.0000 & 50.0000 & 1.0000 &  & 1.0930 & 3.1774 & 3.6107 &  & 2.2096 & 5.1455 & 5.4346 &  & 1.7649 & 4.6148 & 5.2181 \\
25.0000 & 100.0000 & 1.0000 &  & 2.3010 & 5.4987 & 5.2511 &  & 2.7322 & 6.1530 & 5.8943 &  & 2.8672 & 6.6794 & 6.5638 \\
50.0000 & 50.0000 & 1.0000 &  & 1.3430 & 2.6593 & 2.6549 &  & 3.6903 & 5.7102 & 6.6092 &  & 2.6447 & 4.6126 & 5.5601 \\
50.0000 & 100.0000 & 1.0000 &  & 2.5544 & 4.5465 & 4.0289 &  & 3.3788 & 5.5997 & 5.5944 &  & 3.5488 & 6.1215 & 6.0997 \\
50.0000 & 200.0000 & 1.0000 &  & 3.1982 & 5.3621 & 4.8240 &  & 3.3973 & 5.6042 & 5.3284 &  & 3.0723 & 5.1269 & 4.8269 \\
100.0000 & 100.0000 & 1.0000 &  & 2.8577 & 4.2410 & 3.3702 &  & 4.4877 & 6.0596 & 6.3150 &  & 4.6587 & 6.5959 & 6.5963 \\
100.0000 & 200.0000 & 1.0000 &  & 3.4348 & 4.9097 & 4.0996 &  & 3.8409 & 5.3476 & 5.0376 &  & 3.2979 & 4.6861 & 4.2967 \\
100.0000 & 400.0000 & 1.0000 &  & 3.8467 & 5.3387 & 4.7443 &  & 4.0185 & 5.5378 & 5.1632 &  & 3.1722 & 4.3971 & 4.0748 \\
\midrule\multicolumn{3}{c}{Mean abs. dev. from 5\%}  &  &3.0573 &1.2962 &1.1886 &  &1.3491 &1.6037 &0.81768 &  &1.9833 &0.97667 &0.83067\\
\bottomrule
\end{tabular}}
\caption{Sizes (in percent) of nominal 5\% level tests  with no heterogeneity in the alternatives. Based on \num{1000000} replications. Newey Bandwidth, three factors.}
\label{table:sizesNeThree3_0.4_0.4_BN_}
\end{table}

\begin{table}
\resizebox{\textwidth}{!}{
\begin{tabular}{@{}S[table-figures-decimal=0,table-figures-integer=3]S[table-figures-decimal=0,table-figures-integer=3]S[table-figures-decimal=1,table-figures-integer=1]cSSScSSScSSS@{}}
\toprule
&&&&\multicolumn{3}{c}{i.i.d.}&&\multicolumn{3}{c}{AR(1)}&&\multicolumn{3}{c}{MA(1)}\\ \cmidrule{5-7} \cmidrule{9-11} \cmidrule{13-15}
{$n$} & {$T$} & {$\sqrt{\omega^4/\phi^4}$} & {} & {$\tUMP$} & {$\tUMPEMP$} & {$P_b$} & {} & {$\tUMP$} & {$\tUMPEMP$} & {$P_b$} & {} & {$\tUMP$} & {$\tUMPEMP$} & {$P_b$} \\
\midrule
25.0000 & 25.0000 & 0.6000 &  & 0.3409 & 1.2142 & 1.6505 &  & 1.4658 & 3.7385 & 3.9776 &  & 1.0648 & 3.4138 & 3.8460 \\
25.0000 & 50.0000 & 0.6000 &  & 0.6516 & 2.4646 & 2.4427 &  & 1.4798 & 4.2709 & 3.2114 &  & 1.3073 & 4.1964 & 3.4938 \\
25.0000 & 100.0000 & 0.6000 &  & 1.3262 & 4.2486 & 3.6417 &  & 2.2889 & 6.0418 & 4.0332 &  & 2.2535 & 6.3537 & 4.4937 \\
50.0000 & 50.0000 & 0.6000 &  & 0.9047 & 2.1034 & 1.9321 &  & 2.1749 & 3.9629 & 3.0134 &  & 1.9077 & 3.8500 & 3.1830 \\
50.0000 & 100.0000 & 0.6000 &  & 1.8729 & 3.8735 & 3.1185 &  & 2.8676 & 5.2993 & 3.6095 &  & 3.1404 & 5.9591 & 4.1796 \\
50.0000 & 200.0000 & 0.6000 &  & 2.3727 & 4.5805 & 3.6723 &  & 3.3682 & 5.9888 & 3.9003 &  & 2.8298 & 5.1637 & 3.5972 \\
100.0000 & 100.0000 & 0.6000 &  & 2.2972 & 3.6694 & 2.8567 &  & 3.4544 & 5.0675 & 3.6437 &  & 4.0526 & 6.1198 & 4.2833 \\
100.0000 & 200.0000 & 0.6000 &  & 2.8961 & 4.4491 & 3.4489 &  & 3.7737 & 5.5150 & 3.7563 &  & 3.1528 & 4.6599 & 3.4039 \\
100.0000 & 400.0000 & 0.6000 &  & 3.2887 & 4.8759 & 3.9439 &  & 4.2975 & 6.1233 & 4.1309 &  & 3.1535 & 4.5657 & 3.4830 \\
25.0000 & 25.0000 & 0.8000 &  & 0.4281 & 1.3461 & 1.7030 &  & 1.4810 & 3.4516 & 4.3192 &  & 1.1370 & 3.2526 & 4.1147 \\
25.0000 & 50.0000 & 0.8000 &  & 0.8791 & 2.8029 & 2.6092 &  & 1.4316 & 3.8096 & 3.4923 &  & 1.3680 & 3.8928 & 3.7450 \\
25.0000 & 100.0000 & 0.8000 &  & 1.7403 & 4.6442 & 3.9925 &  & 2.1275 & 5.2467 & 4.4039 &  & 2.3131 & 5.8618 & 4.9925 \\
50.0000 & 50.0000 & 0.8000 &  & 1.1484 & 2.3988 & 2.0396 &  & 2.0682 & 3.6285 & 3.6244 &  & 1.9738 & 3.7413 & 3.7599 \\
50.0000 & 100.0000 & 0.8000 &  & 2.2293 & 4.1505 & 3.4328 &  & 2.6583 & 4.7169 & 4.1595 &  & 3.1390 & 5.5939 & 4.8661 \\
50.0000 & 200.0000 & 0.8000 &  & 2.7940 & 4.9094 & 4.1489 &  & 3.0365 & 5.2509 & 4.3682 &  & 2.7455 & 4.7664 & 4.0542 \\
100.0000 & 100.0000 & 0.8000 &  & 2.6206 & 3.9694 & 3.0386 &  & 3.1945 & 4.6022 & 4.2217 &  & 3.9703 & 5.8404 & 5.0614 \\
100.0000 & 200.0000 & 0.8000 &  & 3.2240 & 4.6662 & 3.7743 &  & 3.4582 & 4.9381 & 4.1899 &  & 3.0392 & 4.3850 & 3.7296 \\
100.0000 & 400.0000 & 0.8000 &  & 3.5446 & 5.0466 & 4.3073 &  & 3.8903 & 5.4523 & 4.5547 &  & 3.0375 & 4.2743 & 3.7346 \\
25.0000 & 25.0000 & 1.0000 &  & 0.4851 & 1.4446 & 1.8301 &  & 1.5121 & 3.3971 & 4.8935 &  & 1.1821 & 3.2211 & 4.6313 \\
25.0000 & 50.0000 & 1.0000 &  & 1.0329 & 2.9825 & 2.8721 &  & 1.4215 & 3.6448 & 3.8487 &  & 1.4502 & 3.8940 & 4.2347 \\
25.0000 & 100.0000 & 1.0000 &  & 1.9995 & 4.8589 & 4.4433 &  & 2.1293 & 5.0200 & 4.9077 &  & 2.3971 & 5.7525 & 5.6331 \\
50.0000 & 50.0000 & 1.0000 &  & 1.2678 & 2.5350 & 2.1625 &  & 2.0655 & 3.5552 & 4.2583 &  & 1.9666 & 3.6485 & 4.3241 \\
50.0000 & 100.0000 & 1.0000 &  & 2.4012 & 4.2840 & 3.6406 &  & 2.6038 & 4.4901 & 4.5941 &  & 3.1090 & 5.4972 & 5.4502 \\
50.0000 & 200.0000 & 1.0000 &  & 2.9656 & 5.0245 & 4.4532 &  & 2.9955 & 5.0497 & 4.8265 &  & 2.7473 & 4.6747 & 4.3900 \\
100.0000 & 100.0000 & 1.0000 &  & 2.6830 & 4.0208 & 3.0671 &  & 3.1179 & 4.4591 & 4.7456 &  & 3.9433 & 5.6975 & 5.7287 \\
100.0000 & 200.0000 & 1.0000 &  & 3.2646 & 4.6956 & 3.8540 &  & 3.3639 & 4.7820 & 4.5191 &  & 2.9477 & 4.2286 & 3.8715 \\
100.0000 & 400.0000 & 1.0000 &  & 3.7035 & 5.1913 & 4.5774 &  & 3.7416 & 5.2167 & 4.8686 &  & 2.9444 & 4.1339 & 3.8330 \\
\midrule\multicolumn{3}{c}{Mean abs. dev. from 5\%}  &  &2.9866 &1.299 &1.7906 &  &2.353 &0.6972 &0.84918 &  &2.4714 &0.89039 &0.91217\\
\bottomrule
\end{tabular}}
\caption{Sizes (in percent) of nominal 5\% level tests  with no heterogeneity in the alternatives. Based on \num{1000000} replications. Newey Bandwidth, $t$-distribution with five degrees of freedom.}
\label{table:sizesNeT1_0.4_0.4_BN_}
\end{table}

\end{appendices}
\end{document}